\documentclass[amsmath,a4paper, twocolumn, superscriptaddress,floatfix,nofootinbib]{revtex4-1}

\usepackage{ucs}
\usepackage{graphicx}
\usepackage{url}

\usepackage{fullpage}

\usepackage[ colorlinks = true, 
             linkcolor = MidnightBlue,
             urlcolor  = MidnightBlue,
             citecolor = orange,
             anchorcolor = ForestGreen,
]{hyperref} 

\usepackage[sans]{dsfont} 
\usepackage{bbm}
\usepackage[mathscr]{euscript}

\usepackage[usenames,dvipsnames]{color}

\usepackage{tcolorbox}
\tcbuselibrary{theorems} 

\usepackage{amsmath}
\usepackage{amssymb}
\usepackage{amsthm}
\usepackage{units}
\usepackage{framed}
\usepackage{mdframed}
\usepackage{thmtools, thm-restate} 

\usepackage{cancel} 

\makeatletter
\def\l@subsubsection#1#2{}
\makeatother

\tcbset{colback=orange!5,colframe=orange!75!yellow, 
float= htb,
fonttitle= \bfseries}

\newtcbtheorem[]{bigexample}{Example}{float*=htb, width=\textwidth}{ex}

\newtcolorbox[blend into=figures]{boxfigure}[3][]
{float*=htb, width=\textwidth,lower separated=false, center upper, 
center title,title={#2},every float=\centering,label= fig:#3,#1}  

\newtcolorbox[]{boxfigmap}[2][]
{width=\textwidth,lower separated=false, center upper,
colframe=orange!75!yellow,
coltitle=orange!75!yellow, 
colback=white,
,#1} 

\declaretheoremstyle[shaded={rulecolor=MidnightBlue,rulewidth=1pt, bgcolor={rgb}{1,1,1}}]{boxed}

\declaretheoremstyle[shaded={rulecolor=Thistle,rulewidth=1pt, bgcolor={rgb}{1,1,1}}]{secboxed}

\declaretheoremstyle[shaded={rulecolor=YellowOrange,rulewidth=1pt, bgcolor={rgb}{1,1,1}}]{terboxed}

\declaretheoremstyle[shaded={rulecolor=Green,rulewidth=1pt, bgcolor={rgb}{1,1,1}}]{tetraboxed}

\declaretheorem[style=boxed]{theorem}

\declaretheorem[within=section]{lemma}
\declaretheorem[sibling=lemma]{remark}
\declaretheorem[sibling=lemma]{corollary}

\declaretheorem[sibling=lemma]{example}
\declaretheorem[sibling=lemma, style=tetraboxed]{proposition}
\declaretheorem[style=secboxed, sibling=lemma]{definition}

	\newcommand{\green}[1]{\textcolor{OliveGreen}{#1}}

	\newcommand{\flag}[1]{\green{ [#1]}}

    \newcommand{\ket}[1]{\vert  #1 \rangle}
    \newcommand{\bra}[1]{\langle #1 |}
    
	\newcommand{\proj}[2]{\ket{#1}\bra{#2}}
	\newcommand{\pure}[1]{\proj{#1}{#1}}

	\newcommand{\hilbert}{\mathcal{H}}

	\newcommand{\Sys}{\operatorname{Sys}}

\newcommand*{\Lump}{\operatorname{Lump}}	

    \newcommand{\com}[1]{ \overline{#1}\, }
	\newcommand{\bic}[1]{ \com{\overline{#1}}  }

	\newcommand{\id}{\mathbbm{1}}

	\newcommand{\ball}{\mathcal{B}}
	\newcommand{\appspace}[1]{\{\omega^\epsilon\}_{\omega\in #1,\epsilon\in\E}}

	\newcommand{\tr}{\operatorname{Tr}  }
	
\newcommand*{\I}{\mathcal{I}}
\newcommand*{\E}{\mathcal{E}}
\newcommand*{\cM}{\mathcal{M}}
\newcommand{\cT}{\mathcal{T}}
\newcommand{\cF}{\mathcal{F}}

\newcommand*{\emb}{{\bf e}}
\newcommand*{\e}{{\bf e}}
\newcommand*{\h}{{\bf h}}

\renewcommand{\P}{{\mathbb P}} 

\newcommand{\local}[1]{\widehat{#1}}

   \newcommand{\mix}[3]{ f_{#1}\left(#2,#3\right)}
   
   \newcommand{\mixt}[3]{ \widetilde f_{#1}\left(#2,#3\right)} 
 
   \newcommand{\mixg}[3]{ f'_{#1}\left(#2,#3\right)} 
   
   \newcommand{\mixtg}[3]{ \widetilde f'_{#1}\left(#2,#3\right)} 
   
   \newcommand{\mixf}[3]{ c_{#1}\left(#2,#3\right)}
   \newcommand{\mixft}[3]{ \widetilde c_{#1}\left(#2,#3\right)}

	\renewcommand{\vec}[1]{\mathbf{#1}}

	\newcommand*{\eps}{\varepsilon}
	\newcommand*{\half}{\frac{1}{2}}

\makeatletter
\def\l@subsection#1#2{}
\def\l@subsubsection#1#2{}

\makeatother

\begin{document}

\title{Resource theories of knowledge}

\author{L\'idia \surname{del Rio}}
\thanks{These authors contributed equally to this work.}
\affiliation{School of Physics, University of Bristol, United Kingdom}
\affiliation{Institute for Theoretical Physics, ETH Zurich, Switzerland}

\author{Lea \surname{Kr\"amer}}
\thanks{These authors contributed equally to this work.}
\affiliation{Institute for Theoretical Physics, ETH Zurich, Switzerland}

\author{Renato \surname{Renner}}
\affiliation{Institute for Theoretical Physics, ETH Zurich, Switzerland}

\date{26th November 2015}

\begin{abstract}
How far can we take the resource theoretic approach to explore physics?
Resource theories like LOCC, reference frames and quantum thermodynamics
have proven a powerful tool to study how agents who are subject to certain constraints can act on physical systems.  This approach has advanced our understanding of fundamental physical principles, such as the second law of thermodynamics, and provided operational measures to quantify resources such as entanglement or information content. 
In this work, we significantly extend the approach and range of applicability of resource theories.
Firstly we generalize the notion of resource theories to include any description or knowledge that  agents may have of a physical state, beyond the density operator formalism. 
We show how to relate theories that differ in the language used to describe resources, like micro and macroscopic thermodynamics.  Finally, we  take a top-down approach to locality, in which a subsystem structure is derived from a global theory rather than assumed. 
The extended framework introduced here enables us to formalize  new tasks in the language of resource theories, ranging from tomography, cryptography, thermodynamics and foundational questions, both within and beyond quantum theory.

\end{abstract}

\maketitle

\begin{boxfigmap}{Structure of this paper}
\includegraphics[width= 1\textwidth]{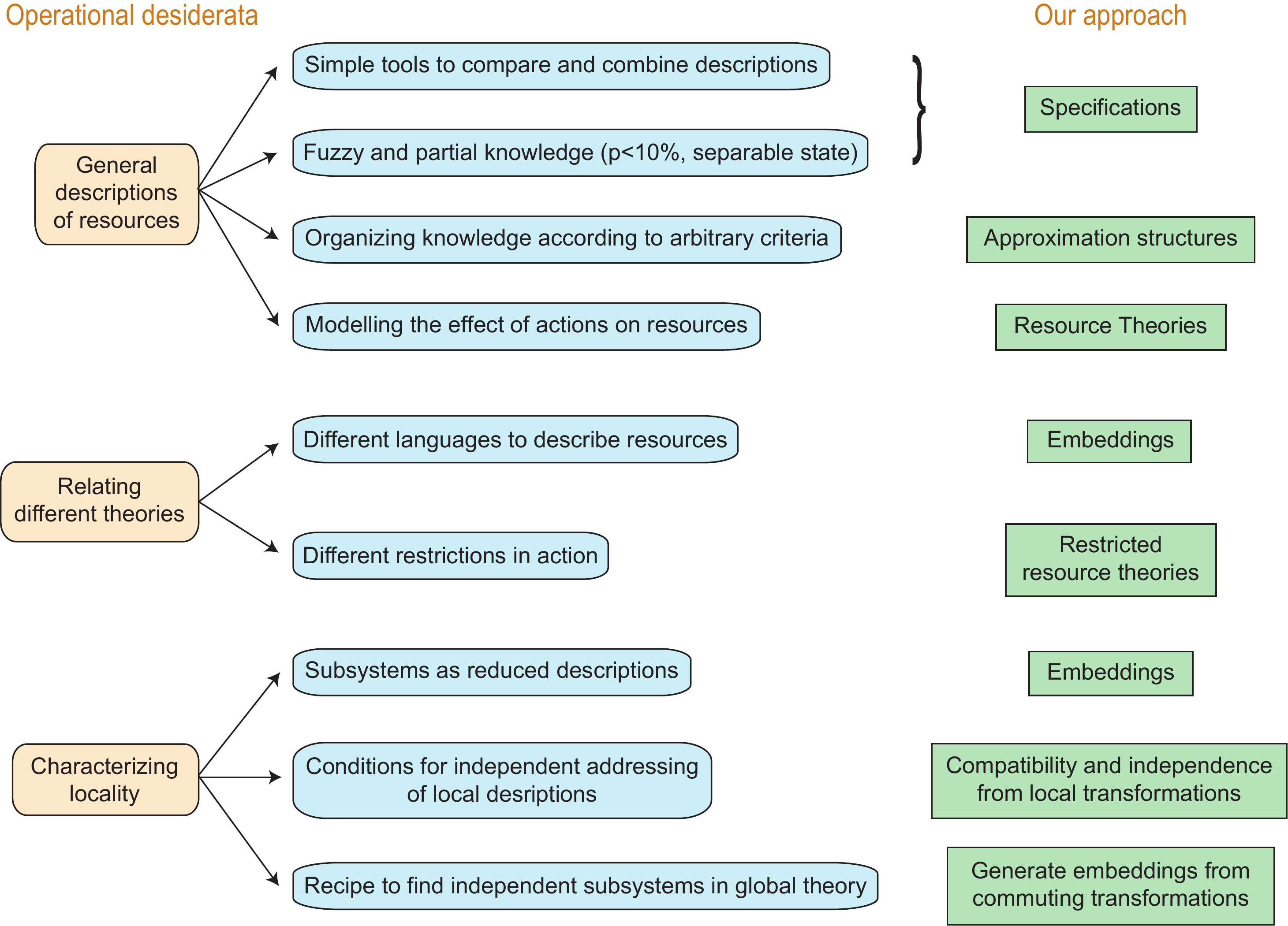}

\tcblower
Summary of this paper: This work is motivated by the three desiderata on the left, which are detailed in the centre (see Section~\ref{sec:desiderata}). On the right, the tools developed here to address them. 

\end{boxfigmap}

\clearpage

{\twocolumngrid
 \hypersetup{linkcolor=black}
\tableofcontents
}

\newpage

\section{Introduction}
\label{sec:intro}
Modern physics and information theory highlighted a fundamental connection between knowledge and action. Consider for instance the thought experiment of Maxwell's demon \cite{Maxwell1871}, which concerns a box filled with a gas and partitioned by a wall with a small door.
The same system could be described both as a `gas of volume $V$, temperature $T$ and pressure $P$' and through a precise microscopic specification of all the gas particles|there is no contradiction in the descriptions, which simply represent the different depths of knowledge of two possible observers. Crucially, this difference affects how different agents may exploit the box, even if they have access to the same set of operations (like opening and closing the door at will).
The demon is an agent who knows the position and velocity of every particle, and can selectively let  certain molecules pass between the two sides, in order to send fast particles to the right and slow ones to the left. This creates a difference in temperature between the two sides, and now the demon can extract work by letting the hot side expand. For an agent who knows only the volume, temperature and pressure of the gas (and does not know when to open the door), such a process seems impossible in line with the second law of thermodynamics.

This example sounds like a perfect candidate to be framed in terms of resource theories, which study precisely the possibilities accessible to constrained agents. 
Instead, it  highlights a fundamental limitation of current frameworks: 
since the two agents  have different descriptions of the same physical system, there is no way to consistently model both within the same theory without  extra assumptions. For example, in statistical physics it is   assumed that the macroscopic observer assigns a particular probability distribution over all microstates compatible with her description,  in order to represent her knowledge as a mixed state (a uniform distribution gives rise to the usual statistical ensembles).
This assumption is both morally unjustified and limiting, because the statements made on this basis apply only to an agent with this particular prior distribution and may not be robust. Despite this, traditional statistical physics claims that the statements based on a uniform prior apply to all agents with a given coarse-grained description (specified e.g.\ by volume and pressure). This may serendipitously  hold  in the thermodynamic regime  of large, uncorrelated systems, for typical observations, but is not necessarily true in general.

In this work, we overcome this limitation by 
generalizing resource theories
to model an agent's knowledge explicitly. We draw from the idea of coarse-grained descriptions such as an $\eps$-ball of quantum states or the specification of a few properties (e.g.\ macroscopic observables,  or that two systems are uncorrelated), which may be compatible with several microstates. We do not impose any distribution over those microstates, which results in concise descriptions, without all the extra information of a probability measure. These \emph{specifications}  allow us to find simple answers to particular questions about the system without checking all the states involved explicitly (e.g.\ common properties of gases at a given pressure and volume, or of separable states).
Such statements naturally  hold true for all Bayesian agents with different prior distributions. 
Furthermore, this approach enables us to explicitly relate  resource theories that differ in the language used to describe resources (like  micro and macroscopic thermodynamics), as well as  local agents acting in different parts of a global space.

Finally, we address another shortcoming of current resource theories, which concerns the subsystem structure of resources.
Traditional formulations of resource theories assume that local resources are well-defined and can be composed as building blocks (e.g.\ by taking the tensor product). Figures of merit like conversion rates and monotones are often expressed in terms of number of copies of local resources (e.g.\  the entanglement of formation is defined as the number of Bell pairs necessary to create a bipartite state with LOCC). 
However, the assumption of a rich subsystem structure is very strong and in natural settings unrealistic. In real life, 
subsystems correspond merely to operational classifications within a global space. Here we take a top-down approach  in which the local structure is derived from a global theory rather than assumed.
We identify two key aspects of traditional subsystems: local descriptions correspond to  coarse-grainings that should be independent of both  other local descriptions and local transformations elsewhere. We express these properties in a global theory and find minimal conditions to make meaningful statements about local resources. Finally, we show how to find an operational subsystem structure within a global theory, based only on commutation relations between transformations. 

Our approach can be applied to a large class of problems, beyond quantum theory. The  concepts introduced here allow us to formalize a range of new tasks in the language of resource theories, ranging from tomography, cryptography, thermodynamics and foundational questions.

In the remainder of this introduction, we first review the notion of resource theories, with examples from quantum information theory, elaborate on the  limitations of current approaches, and then outline the desiderata of a generalized framework.

\subsection{Resource theories}
Resource theories are a framework to study the possible actions of agents, given some constraints, like the laws of physics, technical limitations or the rules of a game.
A resource theory is usually defined by a set of allowed transformations on a state space, which correspond to operations that are easy or cheap to implement.

We illustrate this with a paradigmatic example: the resource theory of \emph{local operations and classical communication}, {\sc LOCC}~\cite{Bennett1996} (see \cite{Horodecki2009} for a review). 
The idea behind {\sc LOCC} is to study entanglement under realistic conditions: in the absence of perfect quantum state distribution, what can be achieved with resources such as a lossy quantum channel or a slightly entangled state?
To study this resource conversion rigorously, one formulates a theory in which agents  can freely perform local operations and classical communication, and everything else is a resource that must be accounted for~\cite{Bennett1996}.
Such an operational motivation is key to most resource theories in quantum information theory, from reference frames to thermodynamics (see Example \ref{ex:quantum_resource_theories}).

The set of free (or allowed) operations imposes a \emph{pre-order}\footnote{A pre-order is a binary relation that is reflexive, $R \to R$ and transitive,  $R \to S\  \& \ S \to T \implies R \to T$.} `$\to$' on the space of resources: we write $R \to S$ if we can transform a resource $R$ into another, $S$, using only free operations. For example, we can use a Bell pair $R$ and some {\sc LOCC} to perform teleportation of one qubit, and therefore simulate a single-qubit use of a perfect quantum channel, $S$. The pre-order induced by {\sc LOCC} on bipartite pure states is given by the majorization relation of the reduced density matrices: $\ket{\psi}_{AB} \to \ket{\phi}_{AB} \iff \tr_A (\pure{\phi}_{AB}) \prec  \tr_A (\pure{\phi}_{AB}) $.

We may then characterize the ordered space. 
For instance, we may identify a set of \emph{free resources}: these are the resources that can always be prepared with the allowed transformations, even when we don't know the initial state. In {\sc LOCC}, this corresponds to separable states and classical channels.
We may also look for necessary and sufficient conditions for state transformations, try to assign value to different resources, study conversion rates between many copies of two resources, find order monotones like squashed entanglement, discuss resource inequalities, and phenomena like  catalysis. These concepts are generalized and explored in the context of more abstract resource theories in \cite{Brandao2015, Coecke2014, Fritz2015}.  
Most results to date are  largely based on the notions of subsystems and composition of local resources (like asking how many copies of a Bell pair are needed to create a given state). In the present work we analyse and generalize those notions, thus paving the way to a more interesting discussion on quantification of resources, which will be treated in the upcoming Part II of this work.

\begin{bigexample}{Quantum resource theories}{quantum_resource_theories}

In quantum resource theories, the state-space $\Omega$  is usually the set of density matrices over a fixed Hilbert space $\hilbert$. The set of allowed operations $\cT$ is a monoid of  trace-preserving completely positive maps (TPCPMs).  In general, ignoring a subsystem by taking the partial trace is always allowed, so operations have the form $\tr_S (\cF(\rho))$.  Let us see some examples.

\paragraph{Quantum theory.}
All TPCPMs are allowed;  all quantum resource theories correspond to restrictions on this theory. \emph{Unitary quantum theory} only allows for unitaries, $\cF(\rho) = U \ \omega \ U^\dagger$.

\paragraph{Locality constraints.}
One of the first examples of quantum resource theories model agents subject to locality constraints, in order to characterize non-local correlations.

We fix some partitions of the global Hilbert space, $\hilbert = A \otimes B \otimes \dots$, and only allow for local TPCPMs in those partitions, $\cF(\omega) = [\cF_A \otimes \cF_B \otimes \dots ](\omega)$. We may also allow for classical channels across the partitions, in which case we recover the resource theory of local operations and classical communication (LOCC). Many variations exist, for instance in which we allow for a limited amount of quantum communication; see \cite{Horodecki2009} for a review.

For local operations, any  state of product form across the partitions is a free resource, and in LOCC all separable states are free. On the other hand, Bell pairs are extremely valuable for bipartite LOCC, because we can use them to create any other state.
Common ways to quantify the value of resources are for instance entanglement of formation (number of Bell pairs necessary to create a state via LOCC) and distillation entanglement (number of Bell pairs that can be obtained from a state), and quantities like the mutual information are monotones for LO.

\paragraph{Thermal operations.}
This resource theory  intends to formalize quantum thermodynamics by studying both energy and entropy flows \cite{Janzing2000, Brandao2013b, Horodecki2013a}.
We fix a global Hamiltonian $H$, and allow for operations that affect neither entropy nor energy (that is, unitaries that commute with $H$). 
If we want to model thermal operations in the presence of a free heat bath at a fixed temperature $T$, we may additionally allow for maps that simulate thermalization of subsystems, for instance by replacing local states with thermal states at temperature $T$. The general form of an allowed operation is therefore
$\cF(\omega) = U \left(\pi_A \otimes  \tr_A (\omega) \right) U^\dagger: \ [U, H]=0,$
where $\pi_A \propto  e^{- \frac {H_A} {k T} }$ is a local Gibbs state. This assumes a weakly interacting Hamiltonian, $H = H_A \otimes \id_{\com A} + \id_A \otimes H_{\com A} + H_{\text{int}}$, with $\| H_{\text{int}} \| \ll \| H \|$.
In this theory, local Gibbs states  are free resources, and quantities like the free energy emerge as monotones (see \cite{Goold2015} for a review).
If the Hamiltonian is fully degenerate, all unitaries are allowed, and we recover the resource theory of \emph{noisy operations}. Extra constraints like momentum conservation may also be added.

Note: It is more common to find the formulation `we may bring in any extra ancilla system in a thermal state, apply a joint unitary, and trace out a subsystem'.  In our framework, the partial trace always comes for free (it is a form of forgetting). Furthermore, we build our specification space such that all the subsystems that could be appended in a fully mixed state are already included in the global Hilbert space. The two views can be related via embeddings.

\paragraph{Quantum computation complexity.} A relevant question for experimentalists is how easy it is to implement a quantum gate. At the moment, single-qubit gates are significantly easier to implement than even two-qubit gates, so one may express the question of complexity of implementation as a resource theory where single-qubit gates are given for free (note that here the state space of resources is made of transformations, not quantum states).
Given that universal sets of gates only require one two-qubit entangling gate (like a CNOT), one can quantify the cost resources (gates) in terms of number of CNOTs needed to implement a gate. For an overview, see for instance~\cite{Iten2015}.

\paragraph{Other examples.}
Aspects of quantum theory such as coherence~\cite{Baumgratz2013,Lostaglio2015,Winter2015}, asymmetry, and reference frames~\cite{Gour2008,Marvian2013,Marvian2014,Marvian2015} have also be studied through the lens of resource theories. For a brief overview see \cite{Coecke2014}.

\end{bigexample}

\subsection{Desiderata for a general framework}
\label{sec:desiderata}

In this section we explore limitations of current formulations of resource theories in more detail, and lay out the operational desiderata that we want a general framework to satisfy. We will refer to these guidelines throughout this work as we set up a simple and intuitive framework that respects them.

\subsubsection{Representation of resources: formalizing subjective knowledge}
\label{sec:Bayes}

As Maxwell's demon illustrated, the resources available to an agent are not simply the ultimate physical states of a system, but rather what she knows about the system. This knowledge can be formalized with  descriptions such as quantum states (in the resource theories of LOCC and noisy operations), state variables like pressure or temperature (in phenomenological thermodynamics), or  molecules and compounds (in chemistry).   
The descriptions one can make (or equivalently, the knowledge one may hold) about a system respect a certain structure: for example, `acetic acid' and `hydrochloric acid' are clearly different compounds, but the description `acid' fits both;  it is a compatible, although \emph{less specific}  (and potentially less useful) description than the former two. 
For quantum systems, we can think of similar examples: for instance, to say that a system is at least $\eps$-close to a state $\sigma$, i.e.\ $\rho\in\mathcal{B}^\eps(\sigma)$, is less specific than describing its exact state $\rho$ (there are several density operators compatible with that description). Similarly, we could use entanglement witnesses to characterize correlations between two systems, without specifying their joint state. 
In quantum resource theories, resources are usually represented by  density operators. However, the examples of coarse-grained descriptions above cannot be expressed in the density matrix formalism alone. Perhaps the reason why this has not been noticed in the past is that in traditional resource theories, one rarely thinks about resources as descriptions of an agent's knowledge, but rather as physical states, for which density matrices would be the best formalism.

Now, a quantum Bayesian\footnote{In QBism, density operators represent \emph{states of knowledge} and correspond to an agent's best bid about the properties of a system at hand. In this view, quantum states are always subjective descriptions, not objective facts about a physical system \cite{Fuchs2015}.} might argue that
the quantum formalism \emph{does} allow an agent to express  
coarse-grained knowledge, since he can always take a convex mixture of quantum states to create a new one that describes more uncertainty. However, the focus of the Bayesian framework lies on updating probability assignments after \emph{gaining additional knowledge}, and as such its basic principles  (such as Dutch book consistency) provide no guidelines for \emph{forgetting or coarse-graining knowledge}. It then seems that in order to coarse-grain his description, a Bayesian has to choose a particular measure,
such as 
one that corresponds to his prior distribution on the  states involved. Whatever statements he later makes will depend on that prior, and may not hold for another agent with the same coarse-grained knowledge.
This is problematic for example in quantum tomography, where we try to characterize a system by  measuring  a small sample subsystem \cite{Christandl2012}. 
In a Bayesian framework, the state an agent ultimately assigns to the system  depends on the prior distribution that he uses (that is, 
the state that he used to describe the system before taking into account the measurement data). However, more generally one might want to make statements independently of the prior|statements different agents could agree on, and  that accurately represent the knowledge learnt through the tomography, without the extra clutter of the prior distribution. This can be done with \emph{confidence regions}, with which  the knowledge about the system is not represented by means of a density matrix alone, but rather by \emph{neighbourhoods} of states, which take precisely the form $\ball^\eps (\rho)$.

There are also cases in which we may simply not be interested in the details of the actual state of the system, even if we could describe it. For example, our description of success of a protocol could be that  a battery is charged above some energy threshold, and  we would accept any final state that was compatible with that requirement; more prosaically, a child might be happy to receive any dog for her birthday, even though she can distinguish different breeds.

We would like to represent resources with a formalism that captures all aspects of the above examples. It should be both simple and powerful at the same time: it should be possible to compare resources, go to more coarse-grained descriptions by forgetting detailed knowledge about the system, and to combine different descriptions to make up a more detailed one.
With these aspects in mind, we can now formulate the first desideratum for a generalized framework for resource theories.

\begin{restatable}[Modelling resources]{desideratum}{desiknowledge}
 \label{desi_knowledge}
	A general framework  should allow agents' descriptions of  resources to go beyond the assignment of subjective probabilities.
	It should specify simple rules to combine, coarse-grain  and organize knowledge according to arbitrary criteria. 
	Finally, it should specify how to update  knowledge after physical transformations.
\end{restatable}

\subsubsection{Relating different theories}

One of the main aims when generalizing resource theories should be to devise a framework that allows us to relate resource theories that differ in the transformations they allow as well as in their language to describe resources. 
This, on
the one hand, 
means that
it should describe how to combine allowed operations, building more relaxed or more restricted resource theories, such as going from LO (local operations)  to LOCC (local operations and classical communication), or constructing a resource theory of local thermodynamics \cite{Goold2015}. 
On the other hand, the framework should enable us to compare theories that are described by different state spaces. 
For example, we  should be able to compare the viewpoint of Maxwell's demon (whose state space accounts for the positions and velocities of microscopic particles) to that of a macroscopic observer (who describes only macroscopic state variables), and use this to study how standard thermodynamics emerges from microscopic models like quantum thermal operations.
Similarly, the framework should allow us to relate different agents acting within 
a resource theory that have only partial knowledge of each other's actions, for instance in cryptographic scenarios.

\begin{restatable}[Relating different resource theories]{desideratum}{desirelating}
  \label{desi_relating}
 	The framework should allow us to combine and relate different resource theories, which may differ in  both the language used to describe resources and  actions available to agents.
 \end{restatable}

\subsubsection{Operational subsystems}

Resource theories usually follow a `bottom-up' approach to subsystems: they specify building blocks, such as local states and resources, that can be composed and discarded at will. While this is often useful in practice to make relevant statements, it is generally a simplified model of the real processes which occur in one global state space. That such an approach is justified is tacitly assumed in each resource theory and cannot be derived from within.
In real life, however, a scientist confronted with a global theory does not assume that it contains a clean subsystem structure, but rather tries to find and characterize operational subsystems.
For example, in genetics, it was not  {\it a priori} known where genes started and ended within a strand of DNA. It is only through making experiments and analysing the effect of transformations (like the replacement of nucleotides, which affects the encoded proteins) that biologists learn the structure of genes. A `traditional' approach to a resource theory of genetics would have as local resources all the mapped genes, and let agents combine them to engineer new creatures. 

For a generalized framework of resource theories, we now ask to take a top-down approach, and start with a global state space and global transformations. 
Subsystems should now not be assumed, but the theory should give a way to construct descriptions of local resources
and local transformations, 
and specify when such descriptions can be used consistently. 

In this way, each agent may divide the global state space into subsystems and recognize degrees of freedoms based on her understanding of the structure of possible transformations, resulting in potentially different, operational classifications.

\begin{restatable}[Characterizing locality]{desideratum}{desicomposition}
 \label{desi_composition}
	A general framework for resource theories should follow a top-down approach, starting from a global theory and deriving an operational subsystem structure. The framework should allow to characterize
	local resources, and specify conditions for consistent use of local  
	descriptions. 
	It should further show in which circumstances the usual bottom-up, building-block type models are justified and known resource theories are recovered.
\end{restatable}

\subsection{Structure of this paper}

In Section~\ref{sec:knowledge} we introduce the basic formalism of our framework, meeting Desideratum~\ref{desi_knowledge}. In Section~\ref{sec:relating_agents} we see how to  relate the views and possibilities of different agents (Desideratum~\ref{desi_relating}). In Section~\ref{sec:composition} we show how, using these ideas, a subsystem structure can be derived from a global theory, and discuss different notions of locality (Desideratum~\ref{desi_composition}). 
In the remaining of the paper we  explore Desideratum~\ref{desi_knowledge} further: in Section~\ref{sec:approximations} we formalize approximations and more generally the structuring of knowledge according to arbitrary criteria, and in Section~\ref{sec:convexity} we fit the notions of convexity and probability in our framework. 
Finally, we wrap up with conclusions, further directions and relation to existing approaches.

The examples provided give an intuition behind our ideas, as well as direct applications of the framework. All the proofs of propositions and theorems in the main text can be found in the appendix, where the results are further explored and expanded, and where you can also find a summary of basic notions of algebra and order theory. Knowledge of quantum theory helps to follow some of the examples and motivation, but is not necessary to understand the framework.

\section{Modelling resources}
\label{sec:knowledge}
In this section we  introduce a formal definition of resource theories that  satisfies Desideratum~\ref{desi_knowledge}. We will define resource theories through two aspects of the agent: a \emph{specification space} to describe her resources, and the allowed operations that she can perform.

The discerning reader may soon recognize familiar aspects of our framework: it might remind her of possible worlds, algebras of events, linear logic, Kripke structures, indicator functions, rough sets,  Dempster-Shafer theory or a limited version of Bayesianism.  Connections to those approaches are discussed in Section~\ref{sec:literature}. Indeed our specification spaces could also be expressed through some of those theories, which is partially what makes it so intuitive. If the next section feels like old news, keep in mind that it is just the minimal structure needed to discuss our problems of interest: resource theories, relations between  agents and notions of locality.

\subsection{Specification spaces}

\begin{bigexample}{Animal specification spaces}{animal_specifications}

Imagine that we want to formalize our knowledge about animals. The elements of $\Omega$ correspond to the best possible description of different animals|for instance, all the known animal species. An example of a state would be  $\omega=\text{`jaguar'} \in \Omega$. Specifications are subsets of $\Omega$ that characterize uncertainty about the exact species, for instance $V=\{\text{jaguar, leopard}\} \in S^\Omega$.
`Mammal' is another specification: a set that includes all animals that are mammals. 

\begin{center}
 \includegraphics[height=5.5cm]{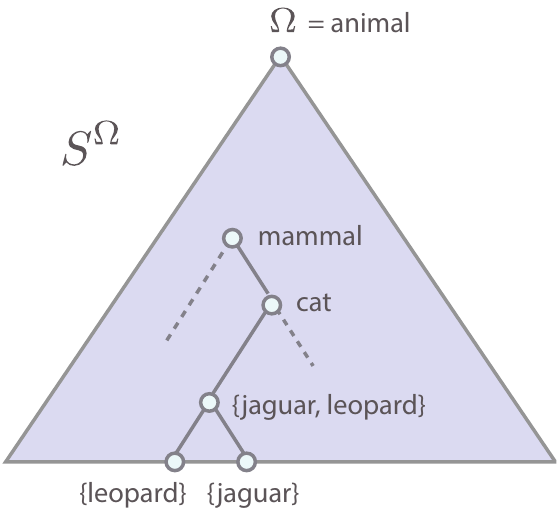}
\end{center}

For an intuitive example of knowledge combination, imagine that Alice and Bob both saw a spotted big cat. Alice knows that the cat lives in Africa, so she thinks that it might be either a cheetah or a leopard, whereas  Bob noticed that the animal is a Panthera, so he knows it could be a jaguar or a leopard.  By combining their knowledge, they reach the conclusion
$
\{\text{cheetah, leopard} \} \cap \{\text{jaguar, leopard}\} = \{\text{leopard}\}.
$
However, if a third observer insisted that the animal was either a puma or a lynx, we would know that at least one of them was wrong, because the three specifications are not compatible. 

For an example of a specification homomorphism,  let $\Sigma$ the set of all colours, and $S^\Sigma$ the corresponding specification space. Now take the function that maps an animal to its colours, $\tilde f: \Omega \to S^\Sigma$; for instance $\tilde f(\mbox{Tiger}) = \{\text{Orange, Black, White} \}$. We can use $\tilde f$ to build a homomorphism $f: S^\Omega \to S^\Sigma$ which maps a specification to all possible colours of its elements. For instance, 
\begin{align*}
 f(\{ \text{Tiger, Leopard}\} ) 
 &=  \tilde f( \text{Tiger} )  \cup  \tilde f( \text{Leopard} ) \\
 &= \{\text{Orange, Black, White} \} 
 \cup \{\text{Yellow, Black, White} \} \\
 &= \{\text{Orange, Black, White, Yellow} \}.
\end{align*}
As an example of a function that is \emph{not} a homomorphism, consider $g: S^\Sigma \to S^\Omega$ to be the function that maps a set of colours to the set of animals that exhibit \emph{precisely all} those colours. Clearly $g(V \cup W) \neq g(V) \cup g(W)$; for instance,
$\text{Tiger} \in g(\{  \text{Orange, Black, White}\})$, but $\text{Tiger} \notin g(\text{Orange}) \cup g(\text{Black}) \cup g(\text{White})$.
Here, $g$ is not a homomorphism because it does not treat specifications as states of knowledge. That is, under $g$, a set $Z \in S^\Sigma$ does not reflect the idea `I have one of these colours, and I do not know which', but rather `I have exactly  all of these colours with certainty'.

\end{bigexample}

\begin{bigexample}{Quantum specification spaces}{quantum_specifications}

For a quantum-mechanical resource theory of a bipartite system, the state space could be the set $\Omega_{AB}$ of all normalized density operators over the global Hilbert space $\hilbert_{AB}$. Specifications, or states of knowledge, are subsets of $\Omega_{AB}$: for instance the specification $\{\rho, \sigma\}$ means that the observer knows that $\hilbert_{AB}$  is either  in state $\rho$ or $\sigma$.  We can also represent fuzzy probabilistic knowledge, of the sort `I know that the probability of state $\sigma$ is at most $5\%$', with the specification $\bigcup_{p \leq 0.05} \{p \ \sigma + (1-p) \ \rho \}$. 
Another intuitive specification would be the neighbourhood of a state, $\ball^\eps(\rho_{AB}) = \{\omega_{AB} \in \Omega_{AB}: \half \ \|\rho_{AB} - \omega_{AB} \|_1 \leq \eps \}$.

\begin{center}
 \includegraphics[height=5.5cm]{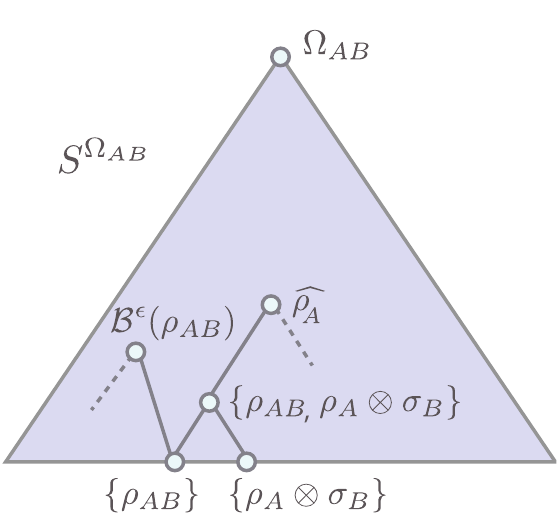}
\end{center}

A specification of particular interest is the description of the reduced state of a subsystem. 
If an agent says that the subsystem $\hilbert_A$ is in state $\rho_A$, this implies that her knowledge is limited to system $\hilbert_A$; the rest of the universe could be in any compatible state. In our framework, this is represented by the  specification $\widehat{\rho_A} \in 2^{\Omega_{AB}}$, which is the set  of all states of $\hilbert_{AB}$ that are compatible with the marginal $\rho_A$,
\begin{align*}
  \widehat{\rho_A} := \{ \omega_{AB} \in \Omega_{AB}: \ \omega_A:= \tr_{B}(\omega_{AB})  = \rho_A\}.
\end{align*}
A relevant example of knowledge combination is when an agent knows two marginals, like $\rho_A$ and $\sigma_B$, but not how they are correlated. His knowledge can be represented as
$$   \widehat{\rho_A}\cap  \widehat{\sigma_B}  = 
\{\omega_{AB} \in \Omega: \ \omega_A = \rho_A,\  \omega_B = \sigma_B \}.
$$
An intuitive example of \emph{forgetting} is induced by taking the partial trace (forgetting the state of a subsystem), as $\{\rho_{AB}\} \subset \widehat{\rho_A}$. Losing precision would be another example, as $\ball^\eps(\rho) \subseteq \ball^{\eps + \delta}(\rho)$.

Trace-preserving completely positive maps (TPCPMs) map density operators to density operators, and therefore may be used to build quantum specification endomorphisms.
For example, a unitary operation $U$ is applied to a quantum specification $W$  as $f_U(W) = \bigcup_{\rho \in W}\ \{U\,\rho\, U^\dag\}$.

\end{bigexample}

The state space of a physical theory can be seen as  the \emph{language} used by a theory to describe resources. 
For instance, the state space of zoology could be the set of all animal species (Example~\ref{ex:animal_specifications}); in quantum theory, all density operators over some Hilbert space (Example~\ref{ex:quantum_specifications}), and in traditional thermodynamics the possible values of macroscopic observables.
In particular, the `states' of a theory are not necessarily the ultimate descriptions of an underlying reality. 
We will see how to relate the state spaces|or languages|of different theories in Section~\ref{sec:relating_agents}.

A particular agent acting within a theory may not have access to the information about the exact state she holds, and may describe her knowledge of the resource at hand through a coarse-grained  \emph{specification}. 
For example, imagine that you are given a box, and told that it holds either a sheep or a fox. How would you describe your knowledge of the animal inside the box? Given that you only know it is one of the two, you could without loss of information express it as a set $\{$sheep, fox$\}$, whose elements are the two possibilities that you admit. 
Many other specifications of knowledge  can be expressed through sets, including  neighbourhoods of a state, uncertainty about a probability distribution over possible states,\footnote{In order to express probabilistic knowledge, we may combine the set formalism with a convex state space, like the space of density matrices. These notions are explored in Section~\ref{sec:convexity}.} and non-local information. They  allow us to frame knowledge of local resources (like a reduced density operator) in the context of a global space (see Example~\ref{ex:quantum_specifications}). 
Furthermore, sets let us formalize learning and forgetting very easily: forgetting corresponds to going to a larger set, and learning, or combining knowledge, is achieved by intersecting sets. 
We formalize the notion of sets|or specifications|as subjective states of knowledge as follows.

\begin{definition}[Specification space] \label{def:specification_space}
The \emph{specification space} $S^\Omega$ of a set $\Omega$ is composed of all non-empty subsets of $\Omega$,
 $$S^\Omega := \{V \subseteq \Omega, \ V \neq \emptyset\} .$$
 Elements $V \in S^\Omega$ are called \emph{specifications}. The set $\Omega$ is called the \emph{state space}, and elements $\omega$ of $\Omega$ are called \emph{states}.

A function on a specification space that achieves forgetting, i.e.\ which satisfies $f(W) \supseteq W$ for all specifications $W$, is called \emph{inflating}.

Two specifications $V, W \in S^\Omega$ are said to be \emph{compatible} if $V \cap W \neq \emptyset$, and in this case we call $V\cap W$ the \emph{combined knowledge} of $V$ and $W$.
\end{definition}

The empty set corresponds to the idea of being wrong|we reach it by combining contradicting knowledge. 

\begin{remark}
A specification space is a partially-ordered set, ordered by inclusion, $(S^\Omega, \subseteq)$.  In particular, it is a join-semilattice $(S^\Omega, \cup)$, where the join operation is the union of sets. We may define the meet of two specifications via set intersection $\cap$, but the structure is not closed under $\cap$ (because we may reach the empty set).
\end{remark}

Although knowledge is represented by general sets of states, some ways of organizing knowledge are more natural than others. For example, some sets correspond to concise descriptions, such as ``mammal'' in the example of the animal specification space, or local density matrices in the case of a quantum specification space (see Example \ref{ex:quantum_specifications}). This idea will be explored  in Section  \ref{sec:approximations}. For now, note that although sets can be large and in general complex, the beauty of specifications is that they can always be seen as a subset of a set that allows a simple description.

As a first example of why this description of knowledge in terms of specifications can be extremely useful, consider again the tomography setting described in Section \ref{sec:Bayes}. In \cite{Christandl2012}, it was noted that the problem of determining the state on a subsystem by means of measurements on a random sample is equivalent to other common formulations of tomography, in which the goal is to determine the ``true'' state $\sigma$ based on measurements on many copies $\sigma^{\otimes n}$. Using specifications, we can now immediately formulate the main result of \cite{Christandl2012}: 
namely, it says that 
for any $n$ and $0<\eps\leq 1$, there exists a non-trivial tomography procedure which, after measurements on $n$ copies of any state $\sigma$, returns a specification $V$ such that $\sigma\in V$ with confidence level $1-\eps$. 
Of course, the size of the specification $V$ will depend on $\epsilon$ and $n$. For precise definitions of the confidence level and the construction of the specification, see \cite{Christandl2012}.

\subsection{Resource theories}

As we saw in the introduction, a resource theory is defined by a set of allowed transformations. These are functions that map elements on specification space to other elements of specification space which are allowed by the theory.

We may use functions on specification space in general to characterize certain aspects of specifications (like the colour of an animal, as in Example~\ref{ex:animal_specifications}) or to represent physical transformations (like the quantum operations of Example~\ref{ex:quantum_specifications}). We will later also use them to relate different specification spaces.
In all these cases, the idea that specifications are states of knowledge is reflected by considering only functions that act individually on each element of a specification.\footnote{This condition, though weaker, is analogous to asking that quantum operations be convexity-preserving.} For example, if you know that two qubits are either in state $\rho$ or $\sigma$, and you apply a $\operatorname{CNOT}$ gate, your knowledge is updated as 
$\operatorname{CNOT}(\{\rho, \sigma\})  = \{\operatorname{CNOT}(\rho), \operatorname{CNOT}(\sigma) \} $. Technically, functions that follow this rule correspond to semilattice homomorphisms, which preserve the structure of a specification space. 

\begin{remark} \label{lemma:homomorphisms}
 Let $S^\Omega$ and $S^\Sigma$ be two specification spaces.
 Then, for any function $f: S^\Omega \to S^\Sigma$, these two statements are equivalent: 
 \begin{enumerate}
 \item $f$ is a join-semilattice homomorphism, that is, for any set $\mathcal W \subseteq S^\Omega$ of specifications,
  $f\left(\bigcup_{W \in \mathcal W}\, W\right) = \bigcup_{W \in \mathcal W}\ f(W);$
 
  \item $f$ is an element-wise function, that is, there exists a function $\tilde f: \Omega \to S^\Sigma$ such that $f(W) =\bigcup_{\omega \in W} \tilde f(\omega) $.
 \end{enumerate}
\end{remark}

The set of transformations in a resource theory, in our framework, forms a monoid of endomorphisms on the specification space.

\begin{definition}[Resource theory] \label{def:resource_theory}
A resource theory is a structure $(S^\Omega, \cT )$, where $S^\Omega$ is a specification space and $\cT$ is a monoid of endomorphisms on $S^\Omega$, equipped with the operation of composition of functions.
The elements of $S^\Omega$ are called \emph{resources} and the elements of $\cT$ are the \emph{transformations}. 
\end{definition}

Now that we have a set of transformations, we can analyze the structure that they impose on the specification space. In other words, we define an operational pre-order $\to$ for resource theories on specification spaces.

\begin{definition}[Reaching specifications] 
  \label{def_construction}
  Let $(S^\Omega, \cT)$ be a resource theory. 
  Given two resources $V, W \in S^\Omega$, we say that a  transformation $f \in \cT$ \emph{reaches} $W$ from $V$, denoted $V \stackrel{f}{\longrightarrow} W$, if $f(V) \subseteq W$. 
  
  More generally, we say that $W$ can be reached from $V$, and denote it by $V \to W$, if there exists a transformation $f \in \cT$ such that $V \stackrel{f}{\longrightarrow} W$.
\end{definition}

In this definition, we use $\subseteq$ instead of equality because an agent is always allowed to forget some information. For instance,  a transformation that turns fish into parrots also turns fish into birds (the less specific description). Similarly, if we start with a quantum state $\rho$ and there is an allowed transformation $f$ such that $f(\rho) = \sigma' \in \ball^\eps (\sigma)$ according to e.g. the trace distance, we may say that this neighbourhood can be reached, $\rho \to \ball^\eps (\sigma)$, although we might not necessarily be able to reach the state $\sigma$ from $\rho$.\footnote{This is the main reason why we exclude the empty set from $S^\Omega$, and allow only for transformations that do not reach it. If we are wrong, we may believe that we can achieve anything: we can go to any specification with the operation of forgetting, since the empty set is contained in all sets. We would therefore obtain a trivial resource theory where everything is possible.}

Note that, in order to have $V \to W$, there must exist a transformation $f$ that transforms every state $\omega \in V$ as $f(\{\omega\}) \subseteq W$. In particular, this implies that knowing more cannot hurt.

\begin{remark}
Let $(S^\Omega,\cT )$ be a resource theory. The relation $\to$ is a pre-order in $S^\Omega$. Now let $V, W \in S^\Omega$ be two compatible specifications, and let $Z\in S^\Omega$. If $V \to Z$ then $V \cap W \to Z$. 
\end{remark}

This formalism gives us a very natural way to define free resources: they are  specifications that can be achieved from any state, in particular even if we do not know the initial state. In other words, a specification that can be reached starting from $\Omega$.
  
\begin{definition}[Free resources]
Let $(S^\Omega, \cT)$ be a resource theory.  $V \in S^\Omega$ is a \emph{free resource} if $\Omega\to V$. 
\end{definition}

In the next sections, we explore and develop this basic structure.

\section{Relating different theories}
\label{sec:relating_agents}
In this section we address Desideratum~\ref{desi_relating}, on relating different resource theories. First we study how to translate between the  languages, or specification spaces, of two theories. We then investigate  how to combine and relate theories that are restricted in the transformations available to agents.

\subsection{Specification embeddings}
\label{sec:embeddings}

\begin{bigexample}{Extensive embeddings}{extensive_embeddings}

In extensive embeddings, the larger state space $\Sigma$ simply contains more states than $\Omega$, but they do not correspond to a more precise description of reality. 

\begin{center}
 \includegraphics[height=1cm]{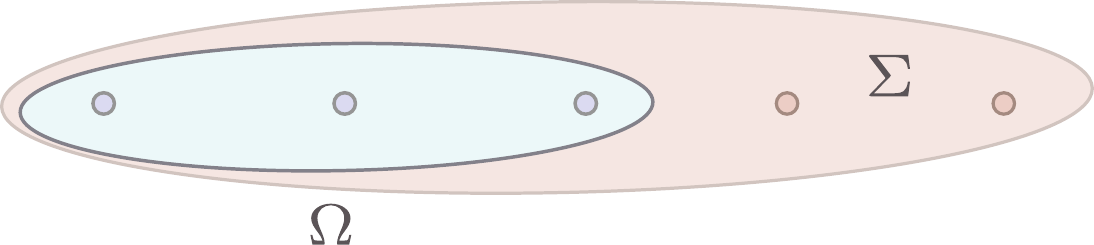}
\end{center}

For example, suppose that you discover a new animal species, {\it Saltuarius eximius}. Then your new state-space of all animal species would be $\Sigma = \Omega\, \cup\, \{${\it Saltuarius eximius}$\}$, and the elements of $\Omega$ would still correspond to basic elements in $\Sigma$.

\begin{center}
  \includegraphics[height=3cm]{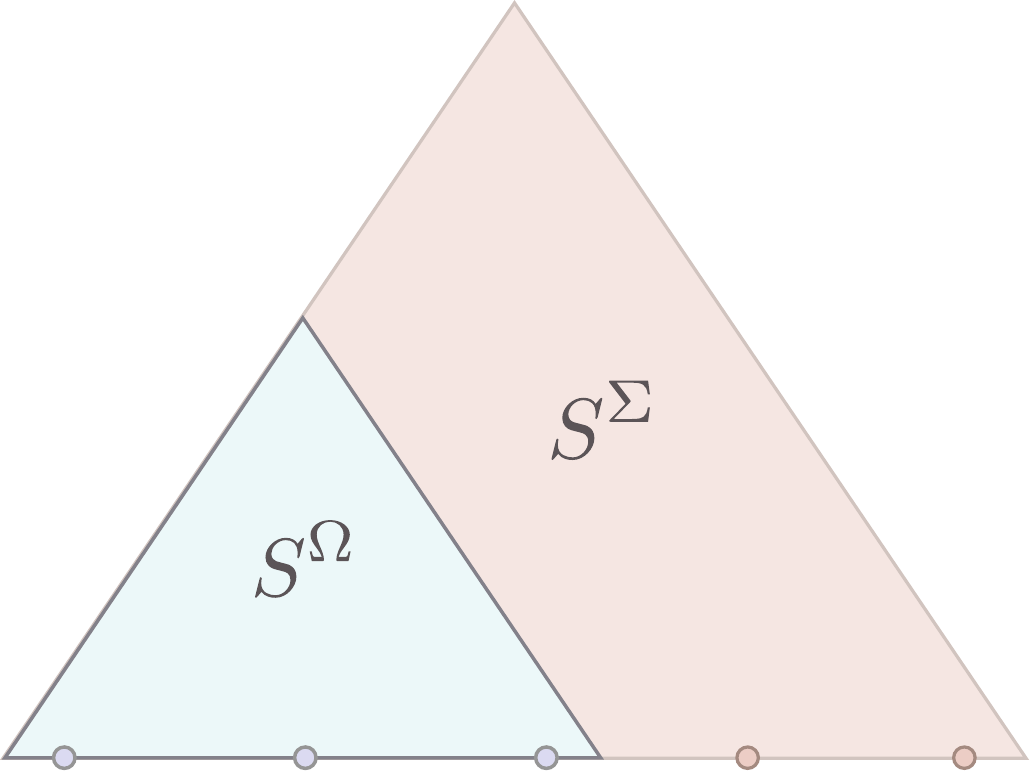}
\end{center}

As an example of extensive embeddings in quantum theory, the set of density operators $\Omega$ has  more elements than the set of density matrices with real entries,  $\Omega(\mathbb R)$. Therefore,  $S^\Omega$ is an extensive embedding of $S^{\Omega(\mathbb R)}$. Fermionic mechanics considers only density operators that correspond to fermions. 
In thermodynamics, an energy shell forms a subset of all possible states of a system, and therefore there is an extensive embedding from the specifications admitted by an observer who knows the total energy of a system to $S^\Omega$.

\end{bigexample}

\begin{bigexample}{Intensive embeddings}{intensive_embeddings}

In intensive embeddings, states of $\Sigma$ correspond to more precise descriptions of states of $\Omega$. In particular, this implies that every specification in $S^\Omega$ has a correspondent description in $S^\Sigma$, and vice-versa. That correspondence is formalized by a Galois insertion $(\e, \h)$ of the reduced space $S^\Omega$ in $S^\Sigma$. 

For animal specifications, an example of an intensive embedding would be, for instance, if a more astute observer could identify subspecies. Her state space $\Sigma$ would have all the subspecies, and states of $\Omega$ (species) would correspond to coarse-grainings of subspecies|in other words, specifications.

\begin{center}
 \includegraphics[height= 1.5cm]{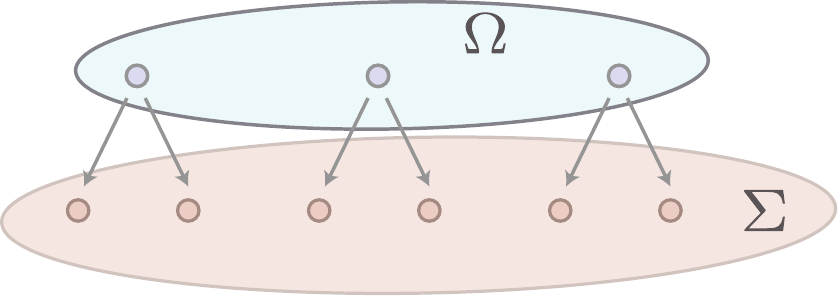} \qquad
  \includegraphics[height=2cm]{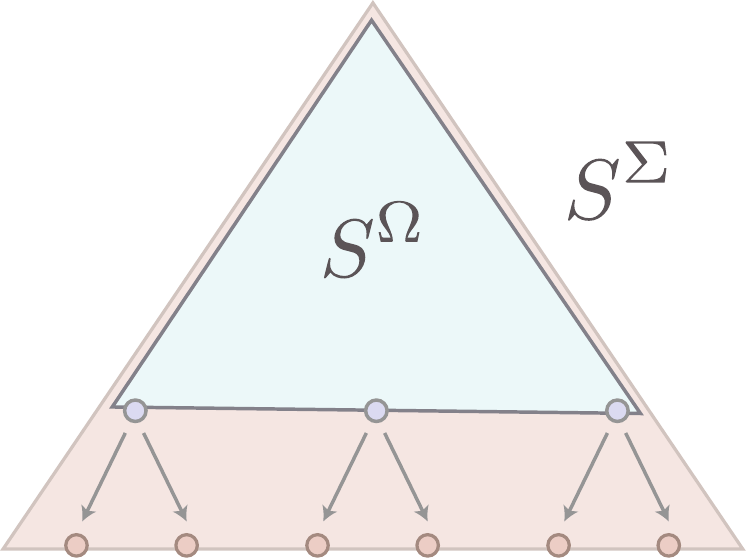}
\end{center}

In quantum theory any coarse-graining of information can also be formulated as an intensive embedding. This is the case of going from the quantum state to all states compatible with some macroscopic observables, or from hidden variables to density matrices.
Another interesting example  concerns local {\it vs} global knowledge. The specification space of a subsystem,  $S^{\Omega_A}$ is embedded in a larger one, $S^{\Omega_{AB}}$: for example, $\e_A (\{\rho_A\} ) =   \widehat{\rho_A} =  \{\omega_{AB}: \ \omega_{A} = \rho_A\}$, and $\h_A(\{\rho_{AB} \}) = \h_A(\{\rho_{A} \otimes \sigma_B\}) = \h_A (\widehat{\rho_A}) = \{\rho_A\}$. Indeed, the lumping map  $\Lump_A = \e_A \circ \h_A$  acts as $\Lump_A(\{\rho_{AB} \}) = \Lump_A (\{\rho_{A} \otimes \sigma_B\}) = \widehat{\rho_A}$.

\begin{center}
 \includegraphics[height=3cm]{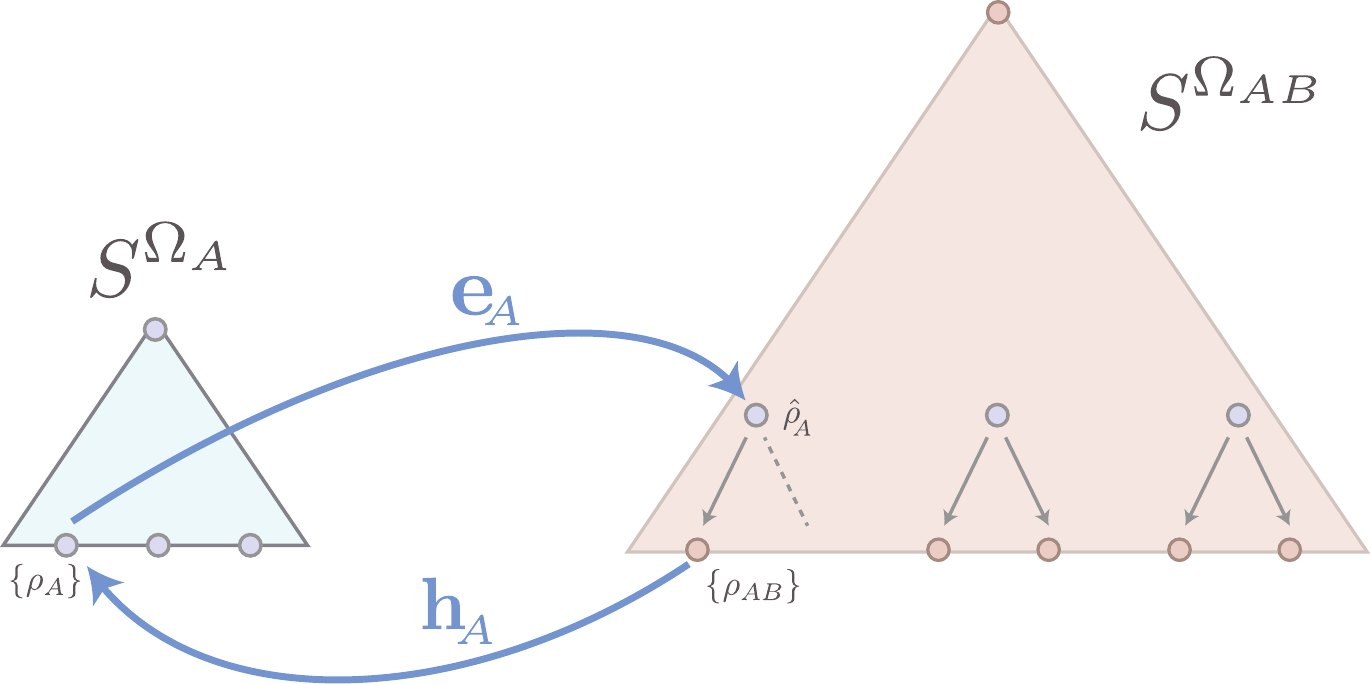}
\end{center}

\end{bigexample}

\begin{bigexample}{Nested embeddings}{nested_embeddings}

\begin{center}
 \includegraphics[height=5cm]{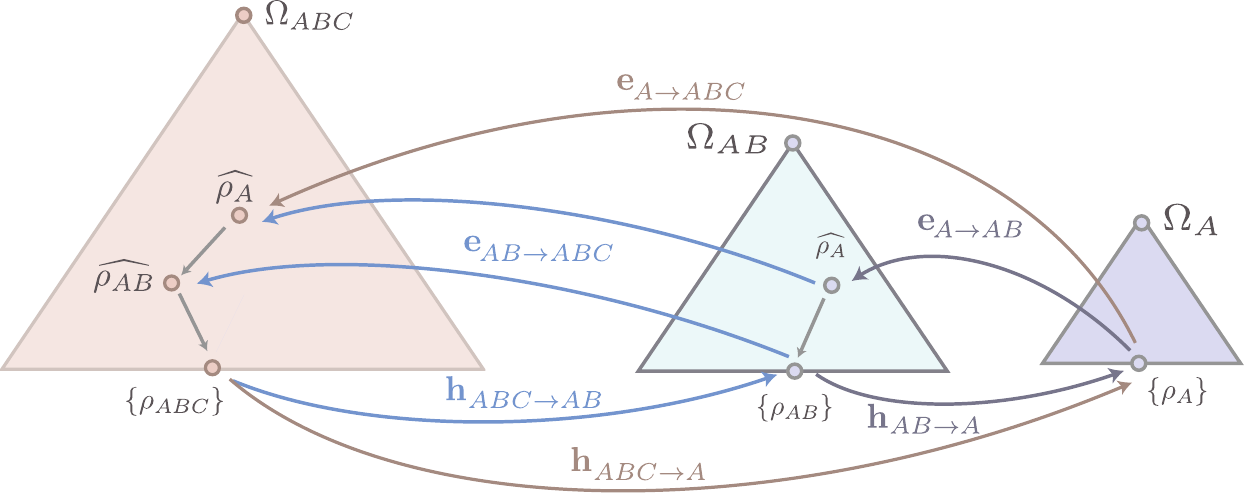}
\end{center}

The authors acknowledge the daunting look of this picture---it depicts a simple concept, if one looks closely. In this example, we want to go from a description of quantum states in a global space $\hilbert_A \otimes \hilbert_B \otimes \hilbert_C$ to a description of reduced states in  $\hilbert_A \otimes \hilbert_B$ and finally in the smaller space, $\hilbert_A$. Proposition \ref{prop:nested} ensures that the embeddings behave naturally: for example, they may be composed as  $\e_{A \to ABC} = \e_{AB \to ABC} \circ \e_{A \to AB}$.

\end{bigexample}

One may find it undesirable that in our framework we always have to specify the most precise state space available within a theory. How can we represent an agent that keeps an open mind to the world and does not assume that her `states' are the most specific descriptions possible?   The answer is two-fold. Firstly, in our framework the choice of state space $\Omega$ by an agent is not connected to a  belief that those states be the ultimate representation of reality. They correspond merely to a choice of language|a level of specificity that the agent is comfortable with. For instance, suppose that an agent wants to engineer molecular interactions. Even if she knows of the existence of quarks, it might be enough for her purposes to model  nucleons as protons and neutrons, ignoring their finer structure.
Secondly, as we now show, an agent can always relate to  more fine-grained descriptions, should they exist, via specification embeddings.

For example, suppose that the state space $\Omega$ corresponds to the most precise description of reality known to an agent, Horatio. This is not necessarily the most complete description possible: another agent may know that there are more things in heaven and earth than are dreamt of in $\Omega$. Additional information about reality can come in two flavours: either Hamlet simply knows of more elements of reality, in which case his state-space is an extension of $\Omega$ (Example~\ref{ex:extensive_embeddings}), or he has more specific descriptions of the elements of $\Omega$ (Example~\ref{ex:intensive_embeddings}). These two cases are formalized as extensive and intensive embeddings of specification spaces.

\begin{definition}[Specification embeddings]
\label{def:specification_embeddings}
Let $S^\Omega$ and $S^\Sigma$ be two specification spaces. 
A function $\e: S^\Omega \to S^\Sigma$ is a \emph{specification embedding} if it is both an order embedding and an order homomorphism. In this case we say that $S^\Omega$ is \emph{embedded} in $S^\Sigma$, and we may call $S^\Omega$ the \emph{reduced} specification space.
We denote the embedding of $S^\Omega$ in $S^\Sigma$ by $\e(S^\Omega) :=\{\e(V)\}_{V \in S^\Omega}$.
 
 A specification embedding is \emph{extensive} if $ |\e(\{\omega \})|=1$ for all $\omega \in \Omega$, and \emph{intensive} if there exists an adjoint homomorphism $\h: S^\Sigma \to S^\Omega$ such that $(\e, \h)$ is a Galois insertion.
\end{definition}

Embeddings and Galois insertions are notions directly inherited from order theory and applied to the join-semilattices  $(S^\Omega, \subseteq)$ and $(S^\Sigma, \subseteq)$. See Appendices~\ref{appendix:algebra_order} and~\ref{appendix:embeddings} for details; for now it suffices to know that an \emph{order embedding}  $\e: S^\Omega \to S^\Sigma$ is a map that satisfies
$V \subseteq W \Leftrightarrow \e(V) \subseteq \e(W)$.
In particular, embeddings are injective and order-preserving. A Galois insertion of $S^\Omega$ in $S^\Sigma$ is a set of two maps $(\e, \h)$, where $\e: S^\Omega \to S^\Sigma$ and $\h: S^\Sigma \to S^\Omega$, such that  $\h \circ \e$ is the identity in $S^\Omega$ and $\e(V) \subseteq Z \ \Leftrightarrow \ V \subseteq \h(Z)$.

The unique characteristic that distinguishes intensive embeddings is that it is always possible to find a reduced specification in $S^\Omega$ from a specification $V \in S^\Sigma$, by taking $\h(V)$. This does not necessarily hold for extensive embeddings, for instance, where we add new states. In particular, intensive embeddings play an important role in defining local specifications from a global one (Example \ref{ex:intensive_embeddings}). 

In fact, it turns out that all specification embeddings can be constructed as a combination of intensive and extensive embeddings. We can use this to compare the languages of  theories that apply to different regimes of the same physical reality, or to  relate the perspectives of  two agents acting within the same underlying theory.
\begin{restatable}[Constructing embeddings]{theorem}{propGeneralEmbeddings}
\label{prop:general_embeddings}
Any specification embedding $\e:\ S^\Omega\to S^\Sigma$ can be written as a combination of an extensive and an intensive embedding, that is $\e=\e_\text{ext}\circ\e_\text{int},$
or likewise
$\e=\e_\text{int}\circ\e_\text{ext}.$
\end{restatable}
For a  constructive proof and further properties of embeddings, see Appendix~\ref{appendix:embeddings}.

In the above, we saw how to relate  different languages used to describe resources. However, we assumed that these languages were already known, and merely described the  embeddings that connect the corresponding specification spaces.
A more general questions would be whether, given a theory and some restricting criteria expressed in its language, we could construct the embeddings and reduced spaces directly. Indeed this is the case.
For extensive embeddings, the  reduced specification spaces can be understood as simply restricting the state space of the bigger specification space (from all animals to only mammals, from all particles to only fermions). On the other hand, for intensive embeddings the reduced spaces correspond to a coarse-graining of the language (from quarks to protons, from species to families). 
These can be formalized via so-called \emph{lumping functions} on the original specification space. Formally, a lumping is an idempotent inflating endomorphism (see Appendix~\ref{appendix:algebra_order}).
For example, this lumping could be a function that maps species to families ($\Lump_{\text{fam}}($jaguar$) = $ Felidae $= \{$leopard, lynx, puma, \dots $\}$) or bipartite states to the knowledge of a marginal ($\Lump_A(\{ \rho_{AB}\}) = \widehat{\rho_A}$; see Example~\ref{ex:quantum_specifications}). 
Lumpings define equivalence classes of specifications ($\Lump_{\text{fam}}($jaguar$) = \Lump_{\text{fam}} ($puma$)$, so we may write jaguar $\sim_{\text{fam}}$ puma),  which in turn induce the reduced specification spaces and intensive embeddings, as shown in the next theorem.

\begin{restatable}[Embeddings and lumpings]{proposition}{propLumpingEmbedding}
\label{thm:lumping}
Let $S^\Sigma$ be a specification space equipped with a lumping map $\Lump$. Then $\Lump$ induces an intensive embedding of a reduced specification space $S^\Omega$ into $S^\Sigma$ with Galois insertion $(\e, \h)$ such that $\Lump = \e \circ \h$. Conversely, through this relation every intensive embedding gives rise to a lumping $\Lump$ on $S^\Omega$.
\end{restatable}

We can use lumpings to find specifications that correspond to reduced knowledge, like the local specification $\widehat{\rho_A}$.
\begin{definition}[Local specification]
\label{def:local_specification}
Let $S^\Omega$ be a specification space intensively embedded in another, $S^\Sigma$, with the corresponding lumping $\Lump$. We say that a specification $V \in S^\Sigma$ is \emph{local} in $\Omega$ (or coarsed to $\Omega$) if $\Lump(V) = V$.
\end{definition}

Proposition~\ref{thm:lumping} also ensures nested embeddings behave well. This is the case of going from a description $\Sigma$ of all subspecies to one of all species, $\Omega$, and then to one of all families of animals, $\Gamma$.  The following result tells us that if we have two of the possible embeddings between the three specification spaces, we can always find the third embedding in the natural way (see Example~\ref{ex:nested_embeddings}). 

\begin{restatable}[Nested embeddings]{proposition}{propNested}  

\label{prop:nested}
Let $S^{A}$, $S^{AB}$ and $S^{ABC}$ be specification spaces. Given two intensive embeddings 
$\e_{A \to AB}$ and $ \e_{AB \to ABC}$, then
$$\e_{A \to ABC}:=\e_{AB \to ABC } \circ \e_{A \to ABC}$$
is an intensive embedding of $S^{A}$ in $S^{ABC}$.

Conversely, given two intensive embeddings with Galois insertions 
$(\e_{AB \to ABC}, \ \h_{ABC \to  AB})$ and $(\e_{A \to ABC}, \ \h_{ABC \to A})$
such that the respective lumpings in $S^{ABC}$ satisfy $\Lump_{ABC \to AB} \subseteq \Lump_{ABC \to A} $, then 
 $$\e_{A \to AB} := \h_{ABC \to AB} \circ \e_{A \to ABC}$$
is  an intensive embedding of $S^A$ in $S^{AB}$.
The two methods give rise to the same embeddings.
\end{restatable}

\subsection{Relating resource theories}

\begin{bigexample}{Restricted agents within resource theories}{reduced}

\begin{center}
 \includegraphics[height=3cm]{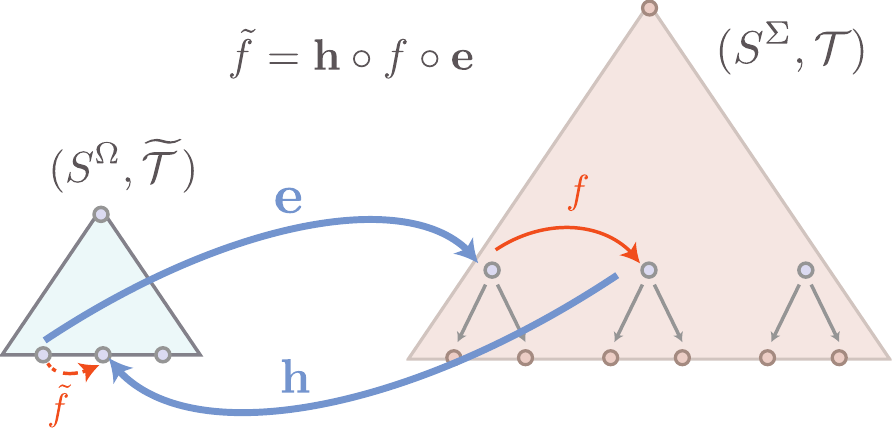}
\end{center}

The simplest example of a restricted agent within a global resource theory $(S^\Sigma, \cT)$ is someone who has limited knowledge, like a macroscopic observer. Their perspective can be represented by a  smaller specification space $S^\Omega$, which is related to $S^\Sigma$ via an intensive embedding $(\e, \h)$. The coarse-grained versions of the transformations in $S^\Omega$ are then built as $\widetilde \cT = \{ \h \circ f \circ \e, \ f \in \cT \}$.
One can use this method to verify whether macroscopic thermodynamics arises from different variations of theories of quantum thermal operations.

\begin{center}
 \includegraphics[height=4cm]{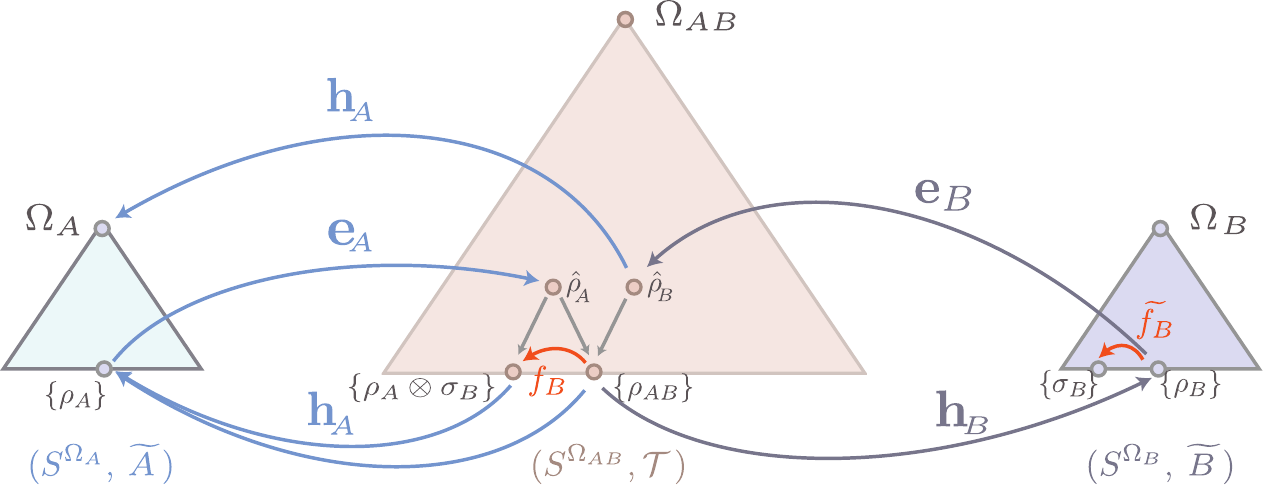}
\end{center}

A more complex example is the case of two local agents (Alice and Bob) acting within a global theory $(S^{\Omega_{AB}}, \cT)$. On the one hand they have access to smaller specification spaces, $S^{\Omega_{A}}$ and $S^{\Omega_{B}}$, conneted via embeddings $(\e_A, \h_A)$ and $(\e_B, \h_B)$ to the global space. On the other, they only have access to subsets of transformations $A, B \in \cT$, corresponding to local operations. Again we need to take the reduced versions of those functions: for example, for functions in $S^{\Omega_{B}}$ we take $\widetilde B = \{ \h_B \circ f_B \circ \e_B, \ f_B \in B \}$. That is, if the original form of a local function in the global space is something like  $f_B =  \id_A \otimes \E_B$, then its reduced version in $S^{\Omega_{B}}$ is $\widetilde f_B =  \E_B$.

The concept of independence between  these two local agents has different aspects to it: for example, it could mean that an action accessible to Bob, like $f_B$, should not impact the reduced specification on Alice's side. It may also mean that any two local specifications are compatible. Finally, it may imply that Alice's  transformations commute with Bob's. We discuss these notions in  Section~\ref{sec:composition}.

\end{bigexample}

Now that we have seen how to relate the languages used by different theories, we can go one step further and consider a particular such agent acting within a global resource theory $(S^\Omega, \cT)$.
This agent, Alice, may be restricted in knowledge, in which case her specification space $S^{\Omega_A}$ would be embedded in $S^\Omega$. She may also be restricted in her actions|for instance, if she is limited to local operations. 
We formalize this notion for intensive embeddings. See Example~\ref{ex:reduced} for illustrations of both restrictions.

\begin{definition}[Restricted agent within a global resource theory] \label{def:restricted}
We say that a resource theory   $(S^{\Omega_A}, \widetilde{A})$ represents a \emph{restricted agent} within an \emph{underlying theory} $(S^\Omega, \cT )$  if:
\begin{enumerate}
\item there is an intensive embedding of $S^{\Omega_A}$ in $S^\Omega$, represented by a Galois insertion $(\e_A, \h_A)$, and
\item  there is  a submonoid  of operations $A \subseteq \cT$ such that
$$ \widetilde{A} =  \{ \h_A \circ f \circ \e_A: \ f \in  A\}. $$
\end{enumerate}
\end{definition}

For example, imagine that the global theory models thermodynamics in the presence of a heat bath (allowing agents to thermalize subsystems and apply reversible, energy-conserving transformations), and Alice is a local agent who sees and controls an isolated system $\hilbert_A$. First we would need to find an embedding of Alice's specification space $S^{\Omega_A}$ in the global one. The second step is to consider only operations of the form  $U_A \otimes \id_{\text{rest}} $, with $[U_A, H_A]=0$, excluding  both  global and thermalizing operations. The set $A$ would be composed of these operations, and $\widetilde A$ would have the \emph{reduced versions} of those operations, from the point of view of Alice: simply $U_A$.
In the next section, we analyze representation of local agents in more detail.

Another way to play with different resource theories is to combine them. As we saw, resource theories are operationally defined|they start by considering a set of operations that are easy to implement. This set might be quite complex, like the set of operations that satisfy a number of constraints. It often helps to study the effect of each of those constraints separately, and build up the final resource theory from individual pieces. This also allows us to vary and adapt the resource theory to the circumstances, for instance if technology evolves to the point where one of the constraints becomes obsolete.

\begin{definition}[Combined resource theories] 
\label{def:combined_rt}
Let $(S^\Omega, \cT)$ and $(S^\Omega, \cF)$ be two resource theories on the same specification space $S^\Omega$. The \emph{combined resource theory}  of $\cT$ and $\cF$ in $S^\Omega$ is defined as $(S^\Omega, \cT \cap \cF)$.
\end{definition}

For example, combining the resource theory of local operations and the theory of thermal operations gives us a theory of thermal operations under locality restrictions.~\cite{Oppenheim02}
A construction $V \to W$ under a combined resource theory $(S^\Omega, \cT \cap \cF)$ is only possible if it is allowed in both  $(S^\Omega, \cT)$ and $(S^\Omega, \cF)$.

\section{Locality and independence}
\label{sec:composition}
In this section we find and characterize the local structure found in a global resource theory $(S^{\Omega}, \cT)$, as per Desideratum~\ref{desi_composition}. From an operational perspective, subsystems correspond to simplified descriptions for local resources and local processes, allowing us to split up global knowledge into independent smaller parts. This means that we are looking for local resource theories of the form $(S^{\Omega_A}, \widetilde{A})$, which we can model by restricted resource theories as in Definition~\ref{def:restricted}.

For two such local resource theories $(S^{\Omega_A}, \widetilde{A})$ and $(S^{\Omega_B}, \widetilde{B})$ (representing local agents Alice and Bob, say), we identify two main operational aspects of independence guaranteed by a traditional subsystem structure: (i)  local descriptions of resources in $S^{\Omega_A}$ and $S^{\Omega_B}$ are independent,
and
(ii)
local actions on $A$  do not affect local descriptions in $S^{\Omega_B}$ (and vice-versa). 
We find relevant consequences of both aspects:  whenever (i) local descriptions do not yield information about other subsystems, local resources can be freely composed, an aspect that is frequently exploited in traditional resource theories. Similarly, in order to split up knowledge, process it locally, and recombine it, it is enough that Alice's operations do not affect Bob's local states (ii). 
Studying the two aspects individually will also be useful to define meaningful concepts of composition, conversion rates or catalysis in general resource theories in which some of the above aspects may fail to be satisfied. They will then yield minimal conditions for such concepts to be well-defined. This will be the subject of Part II of this work|for a brief outline, see Section~\ref{sec:conclusions}.

A related question is how to find such a subsystem structure in the first place, starting only from a global theory. We find that  commutation relations at the level of transformations allow us to construct independent local resource theories. Although this is not necessarily the only way to find an operational subsystem structure, it gives us a complete structure. For example, in quantum theory it recovers the partial trace.

\subsection{Free composition  of local resources}
The first aspect of independence is that 
specifications in the local specification space of Alice yield no information about Bob's local knowledge. For example, in quantum theory $\local{\rho_A}$ carries no information about the local state in $S^{\Omega_B}$, as $\h_B (\local{\rho_A}) = \Omega_B$. 
This leads to the notion of independence between two specification embeddings.

\begin{definition}[Independence of embeddings]  \label{def:independent_embeddings}
Let $S^\Omega$ be a specification space. 
Let there be an intensive embedding $(\e_A, \h_A)$ of a specification space $S^{\Omega_A}$  in $S^{\Omega}$. We say that  a specification $V\in S^\Omega$ is \emph{compatible} with the embedding if $\h_A(V)=\Omega_A$.

Now let there be a second intensive embedding $(\e_B, \h_B)$ of a specification space $S^{\Omega_B}$  in $S^{\Omega}$.
We say that the two embeddings are \emph{compatible} if:
\begin{enumerate}
\item for any specification $V_A\in S^{\Omega_A}$, $\e_A(V_A)$ is compatible with $S^{\Omega_B}$, and 
\item for any specification $W_B\in S^{\Omega_B}$, $\e_B(W_B)$ is compatible with $S^{\Omega_A}$.
\end{enumerate}
\end{definition}

It turns out that embeddings are compatible precisely when two local specifications in the respective specification spaces are freely composable. For example, in quantum theory this would mean that we can combine any $\rho_A$ and $\sigma_B$, that is $\local{\rho_A} \cap \local{ \sigma_B}$ always exists.

\begin{restatable}[Free composition of local resources]{proposition}{propFreeComposition}

\label{prop:free_composition}
Let $S^\Omega$ be a specification space. 
Let there be two intensive embeddings of  specification spaces $S^{\Omega_A}$ and $S^{\Omega_B}$ in $S^{\Omega}$. Then the following are equivalent:
\begin{enumerate}
\item the two embeddings are compatible,
\item for any specifications $V_A\in S^{\Omega_A}$ and $W_B\in S^{\Omega_B}$, there exists a specification $Z\in S^\Omega$ such that $\h_A(Z)=V_A$ and $\h_B(Z)=W_B$,
\item any two specifications $V_A\in S^{\Omega_A}$, $W_B\in S^{\Omega_B}$ can be composed, $\e_A(V_A)\cap \e_B(W_B)\neq \emptyset$, and this combination leaves local information unchanged,  $\h_A(\e_A{V_A}\cap \e_B{W_B})=V_A$.
\end{enumerate}
\end{restatable}

It is worth noting that not all resource theories with a natural notion of subsystems satisfy free composition of local resources. For example, in quantum theory, if $\Omega_{AB}$ consists of all pure bipartite states, and Alice sees a mixed state, we know that Bob cannot have a pure state on his subsystem. 
Similarly, in the genetics of some animal, it could be that not all combinations of genes are viable. The same is true if there is some conserved quantity in the global theory, like a fixed overall energy: then two local states of high energy may not be compatible.  

However, free composition of local resources is at the heart of many concepts in traditional resource theories, such as conversion rates and catalysis. Therefore, if one wishes to find one such resource theory within a physical theory, it is essential  to verify that free composition holds on the systems of interest.

\subsection{Independent processing}

In addition to compatibility of the respective embeddings, the second central aspect of two agents to be able to treat knowledge independently lies in the relation between the embeddings and the transformations that they can apply. That is, Alice's local knowledge should still be valid after Bob performs a local operation. In quantum theory, local specifications in $S^{\Omega_A}$ are independent of local maps on $B$. 

\begin{definition}[Independent agents]
\label{def:independent_agents}
Let $(S^\Omega,\cT)$ be a resource theory and $B \subseteq \cT$ a submonoid of transformations.
Let there be an intensive embedding of a specification space $S^{\Omega_A}$ in $S^{\Omega}$, defined by the Galois insertions $(\e_A, \h_A)$.

We say that a specification $V_A \in \Omega_A$ is independent of the transformations in $B$ if,  for all  $f_B\in B$, 
$$\h_A \circ f_B\circ\e_A(V_A) \subseteq V_A .$$

We that the \emph{embedding is independent} of $B$ if all $V_A\in S^{\Omega_A}$ are independent of $B$.

Two restricted  resource theories $(S^{\Omega_A},\tilde{A})$ and $(S^{\Omega_B},\tilde{B})$ of a global theory $(S^\Omega,\cT)$  are said to be \emph{independent} if
$S^{\Omega_A}$ is independent of $B$, and $S^{\Omega_B}$ is independent of $A$.
\end{definition} 

This independence condition is particularly important in cryptographic applications, where Bob could represent some (unpredictable) adversary, but it is also relevant in the context of locality or isolated degrees of freedom.
When this condition is satisfied,  any knowledge $V_A\in S^{\Omega_A}$ can be ignored for the purpose of implementing a transformation in the monoid $B$.

\begin{restatable}[Independent processing]{proposition}{propIndependentProcessing}
\label{prop:independent_processing}
Let $(S^\Omega,\cT)$ be a resource theory and $B \subseteq \cT$ a submonoid of transformations. 
Let there be an intensive embedding of a specification space $S^{\Omega_A}$ in $S^{\Omega}$, defined by the Galois insertion $(\e_A, \h_A)$. 
If the embedding is independent of $B$, then 
$$f_B (\e_A(V_A) \cap W) = \e_A(V_A) \cap f_B(W) ,$$
for any transformation $f_B \in B$, any local specification $V_A \in S^{\Omega_A}$ and any other specification $W \in S^\Omega$ that is compatible with  $\e_A(V_A)$. 
\end{restatable}

The following theorem tells us  that  when agents are independent, then their respective knowledge can be composed, transformed and treated individually.  In other words, a local treatment of their individual resource theories is compatible with the global picture.

\begin{restatable}[Independent agents can operate independently]{theorem}{thmIndependentAgents}
\label{thm:independent_agents}
Let  $(S^{\Omega_A},\tilde A)$ and $(S^{\Omega_B},\tilde B)$ be independent agents within a global theory $(S^\Omega,\cT)$.
Then for any local functions $f_A\in A$ and $g_B\in B$,  and for any global specification  $V\in S^{\Omega} $,
\begin{align*}
f_A\circ g_B(V)
 &\subseteq  \e_A \circ  \tilde f_A \circ  \h_A (V) \cap \e_B \circ  \tilde g_B \circ  \h_B (V).
\end{align*}
\end{restatable}

Note that this theorem says that we can always lose some information in going to the local pictures (for instance about correlations between the two local descriptions), but the statements made given only local information are still correct. 
In Part II of this work, we will see that we can apply the usual bottom-up approach of traditional resource theories if the independence conditions are satisfied and the local descriptions are compatible.


\subsection{Locality from commutativity of transformations}
\label{section:modularity}

\begin{bigexample}{Finding complete subsystems from commutativity relations}{subsystems_3qubits}

In this quantum example, the resource theory $(S^\Omega, \cT)$ consists of all unitary transformations acting on three qubits. We may use   commutativity relations to find complete subsystems at the level of transformations. 

\begin{center}
 \includegraphics[height=1cm]{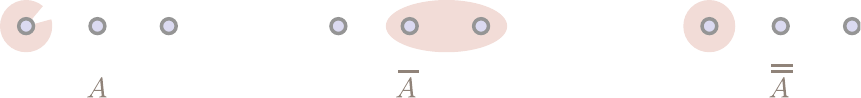}
\end{center}

For example, suppose that we start with a subset of transformations $A \subset \cT$ corresponding to \emph{almost all} transformations on the first qubit.  Taking the  commutant $\com A$  (transformations in $\cT$  that commute with all of those in $A$) gives us the monoid  of all  joint operations on the last two qubits. The bicommutant of $A$, $\bic A$, is then the monoid of all unitaries acting on the first qubit. In this sense, taking the bicommutant corresponds to a form of closure under commutation relations.
 
This is similar in spirit to the von Neumann  bicommutant theorem, which establishes that taking the bicommutant of a von Neumann algebra is equivalent to finding its topological closure. 
Our approach also reminds of the derivation of subsystems on Hilbert spaces through commuting operator algebras (as in Tsirelson's problem \cite{Zanardi2004}).

 \begin{center}
 \includegraphics[height=1cm]{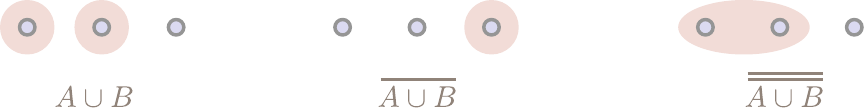}
\end{center}

In the second picture, we see how to find joint systems. Suppose that we want to find the subset of operations on the first two qubits, $A \vee B$. We start from the set $A \cup B$, which corresponds to local operations on the two qubits, but no joint operations. If we take its commutant, we obtain $\com{A\cup B} = C$, the set of operations on the last qubit. The bicommutant is now the set of all operations that commute with $C$, and that is precisely the set of joint operations on $A$ and $B$.

\end{bigexample}

So far we have overlooked another aspect of locality and independence, at the level of action: Alice and Bob's local transformations often commute. For instance, in quantum theory  operations of the form $U_A \otimes \id_B$ commute with those of the form $\id_A \otimes U_B$.
Commutativity of operations is also relevant for a single agent who wants to define subsystems, or identify degrees of freedom, in a larger space that he controls: for example, it can help experimental physicists identify operational qubits within more complex physical systems like trapped ions or superconducting electrodes, or biologists find functional genes within a strand of DNA.  
More generally, here we show  how to use  commutation relations of transformations to find a complete set of independent local resource theories.  
First we formalize commutation relations and independent transformations.

\begin{definition}[Commutant and independent operations]
Let $(S^\Omega, \cT)$ be a resource theory, and let
 $A\subset \cT$ be a subset of operations. The \emph{commutant} of $A$ consists of the transformations that commute with all the elements of $A$, 
\begin{align*}
 \com A := \{g\in \cT: \ &f \circ  g (V)=g \circ  f(V), \\ &\forall \, f\in A, \ \forall \, V \in S^\Omega \}.
\end{align*}
Two submonoids of transformations $A, B \subseteq \cT$ are \emph{independent}, denoted $A\perp B$, if they commute, that is $A \subseteq \com B$ and $B \subseteq \com A$, and $A \cap B = \{1\}$.

Two restricted agents  $(S^{\Omega_A}, \tilde{A})$ and $(S^{\Omega_B}, \tilde{B})$  are said to have access to \emph{independent  operations} if $A \perp B$. 
\end{definition}

The next step is to study the structure induced by these commutation relations, at the level of transformations. 
We  say that a subset of transformations forms a complete subsystem if it is closed under a certain commutation relation (taking the bicommutant).  See Example~\ref{ex:subsystems_3qubits} for an intuitive example.

\begin{definition}[Complete subsystems of transformations]
\label{def:subsystems_transformations}
Let $(S^\Omega, \cT)$ be a resource theory.
The \emph{bicommutant} $\bic A$ of a set of transformations $A\subseteq \cT$, is the commutant of $\com A$. We say that $A$ is \emph{complete}  if $A=\bic A$. (In that case, $A$ forms a submonoid of $\cT$, see Proposition \ref{prop:complete_submonoid}.)

The \emph{space of complete subsystems} of the theory, $\Sys(\cT)$ is the set of all complete submonoids of $\cT$. 
We call each element $A \in \Sys(\cT)$ a \emph{subsystem of transformation}.
Given two subsystems $A, B \in \Sys(\cT)$, we call $A \vee B := \bic{A \cup B}$
the \emph{joint subsystem} of $A$ and $B$. 
\end{definition}

The following proposition tells us that complete subsystems form a well-behaved structure. It further tells us how subsystems can be composed to yield joint systems as well as how to identify common subsystems, and that by this procedure, the smallest subsystem is given by the set of transformations that commute with all others (if two or more operations in the resource theory do not commute, this would be the identity transformation); the largest subsystem is the whole monoid of transformations of the resource theory. Further exploration of these ideas can be found in Appendix \ref{appendix:subsystems}.

\begin{restatable}[Complete subsystems form bounded lattice]{proposition}{thmTransformationsLattice}

\label{prop:complete_lattice}
Let $(S^\Omega, \cT)$ be a resource theory.
The space of complete subsystems $\Sys(\cT)$, together with the operations
\begin{align*}
A \vee B &:= \bic{A \cup B},\\
A \wedge B &:= \bic{A \cap B},
\end{align*}
forms a bounded lattice  $(\Sys(\cT), \vee, \wedge)$. The bottom is formed by the set of transformations that commute with all others, $\com \cT$ and the top is $\cT$.
\end{restatable}

These subsystems represent the most general way in which a global agent could describe subsystems of transformations. However, we should note that they may not necessarily correspond to the operations available to local agents. 
For instance, in the previous thermodynamic example, where Alice controlled an isolated system, if we were to consider $\bic A$, this would contain thermalizing transformations, which we might want to exclude for independent, operational reasons. 

We can now connect commutativity to  notions of locality at the level of specifications:  complete subsystems of transformations can be used to define independent resource theories.  Free composition is achieved with an extra assumption. For constructive proofs, see Appendix~\ref{appendix:subsystems}.

\begin{restatable}[Independent agents from independent transformations]{theorem}{thmTransformationsIndependence}
\label{thm:transformations_independent_agents}

Let $(S^\Omega,\cT)$ be a resource theory and $A, B\subseteq\cT$ two independent complete subsystems of transformations. Then there exist two reduced specification spaces $S^{\Omega_A}$ and $S^{\Omega_B}$ intensively embedded in $S^{\Omega}$ such that the resource theories $(S^{\Omega_A}, \tilde A)$ and $(S^{\Omega_B}, \tilde B)$ represent independent agents.  
\end{restatable}

\subsection{Quantum subsystems}

The tensor product subsystem structure of quantum theory satisfies our generalized notions of subsystems and independence.
Take the resource theory of quantum mechanics in a finite Hilbert space $\hilbert=\bigotimes_{i\in I} \hilbert_i$, where the allowed operations correspond to TPCP maps. The first thing to note is that the set of TPCP maps acting on each individual subspace is each own bicommutant.
Hence, these sets form complete subsystems of transformations, and can be found in the theory through our top-down approach. From Theorem \ref{thm:transformations_independent_agents} it follows that they induce intensive embeddings, that in this case correspond precisely to the partial trace.

\begin{restatable}[Partial trace as an induced embedding]{proposition}{propPartialTrace}
\label{prop:partial_trace}

Let $(S^\Omega, \cM)$ be quantum theory in a finite Hilbert space $\hilbert=\bigotimes_{i\in I} \hilbert_i$. 
Let $J \subseteq I$ denote any collection of subspaces and $\com J$ its complement. 
The set of TPCPMs $\cT_J \subseteq \cT$ acting only on subsystems $J$ (that is maps of the form $\E_J \otimes \mathcal I_{\com J}$) forms a complete subsystem of transformations, $\bic{\cT_J} = \cT_J$.  This set induces the intensive embedding  $(\e_J, \h_J)$, with 
\begin{align*}
 \e_J : \quad  S^{\Omega_J} \to& \ S^{\Omega}\\
  W_J \mapsto& \quad \{\rho \in \Omega: \ \tr_{\com J}(\rho) \in W_J  \} \\
  &= \{ \local{\rho_J}: \ \rho_J \in W_J  \}, \\
 \h_J : \quad  S^{\Omega} \to& \ S^{\Omega_J}\\
  W \mapsto& \quad 
   \{\tr_{\com J} (\rho): \ \rho \in W \} \\
  &=\{\rho_J \in \Omega_J: \  \widehat{\rho_J} \cap W \neq \emptyset \} . 
\end{align*}
In the above, $\Omega_J$ denotes the space of density matrices in the subspace $\bigotimes_{i\in J} \hilbert_i$. 

For non-intersecting collections $J$ and $L$, we have that $S^{\Omega_J}$ and $S^{\Omega_L}$ are freely composable. 
Furthermore,  $(S^{\Omega_J}, \tilde \cT_J)$ and $(S^{\Omega_L}, \tilde \cT_L)$ represent independent agents with commuting operations. 

\end{restatable}

The above insight is particularly interesting in the cases where Hilbert space decomposition is not unique, or the quantum resource theory is more restrictive in the set of allowed operations. For example, due to technical constraints, in certain situations it might be easier to work with an entangled basis as opposed to the bases of individual atoms or spins, e.g.\ in the case of quantum dots to encode logical qubits in fault-tolerant quantum computing \cite{Taylor2005} or of virtual qubits in nano heat engines \cite{Brunner2012}. 

One can find further interesting embeddings and restricted resource theories within quantum mechanics.
Indeed, the set of complete subsystems would include not only transformations on individual subspaces, but also transformations on more specific degrees of freedom. For example, in unitary quantum mechanics rotations along the $X$ axis of a qubit 
also correspond to a complete subsystem, because they form their own bicommutant. The corresponding embedding would be of a classical particle with only this degree of freedom.


\subsection{Effective resource theories}

Restricted resource theories arise from more general theories through the limitations in either knowledge  or available transformations of a given agent. 
One may now ask how the agent's possibilities could be enhanced by having access to a particular additional resource, like a heat bath, a laser beam or a catalyst. In traditional resource theories, we model such an additional resource by appending a new subsystem, such as a catalyst, to a given state on a target system. However, from a global point of view, the additional resource corresponds merely to additional knowledge outside the target. For example, suppose that the agent is  restricted in knowledge to $S^{\Omega_A}$, which is embedded in the global space $S^\Omega$. The additional resource can be represented by a global specification $K \in S^\Omega$  that is compatible with the embedding $S^{\Omega_A}$ (for example, information about the state of an independent subsystem).  The resulting induced transformations on the smaller specification space give rise to a new resource theory in $S^{\Omega_A}$, which we call an \emph{effective resource theory}. This concept is similar in spirit to the regularized resource theories of \cite{Fritz2015}.

\begin{definition}[Effective resource theory]
\label{def:effective_resource_theory}
Let $(S^\Omega,\cT)$ be a resource theory, and 
$(S^{\Omega_A}, \widetilde A)$ describe a restricted agent within the theory. 
Now let $K\in S^\Omega$ be a specification compatible with $S^{\Omega_A}$. 

We call the resource theory $(S^{\Omega_A}, \widetilde A^K)$ the  \emph{effective resource theory}  of  $(S^{\Omega_A}, \widetilde A)$ induced by $K$ if the functions in $\widetilde A^K$ have the form
\begin{align*}
 \widetilde f^K_A: \ S^{\Omega_A} &\to S^{\Omega_A} \\
 V &\mapsto \h_A \circ f_A (\e_A(V) \cap K) , 
\end{align*}
where $f_A \in A$ and $(\e_A, \h_A)$ is the Galois insertion defining the embedding.
\end{definition}

A  prominent example is given by the resource theory of thermal operations, which can be seen as an effective resource theory resulting from the resource theory of unitary, energy-conserving operations through access to additional heat baths at a fixed temperature. In the upcoming Part~II of this work we also discuss the relation between catalysis and effective resource theories. 

\section{Approximations and robustness}
\label{sec:approximations}

\begin{bigexample}{Approximation structures in the animal kingdom}{zoo_approximation}
A natural notion of proximity between animals emerges from their biological classification:
take $\E$ to be the ordered set $ \{$ Species $<$ Genus $<$ Subfamily $<$ Family $<$ Suborder $<$ Order $<$ Class $<$ Phylum $<$ Kingdom $\}$. 
We define the approximation structure such that the $\eps$-ball of an animal contains all the animals that belong to the same $\eps$-classification.
For example, 
\begin{align*}
\{\text{King penguin}\}^{\text{Species}}
&= \{\text{King penguin}\}\\
\subseteq \quad
\{\text{King penguin}\}^{\text{Class}}
&= \text{Aves} \\
\subseteq \quad
\{\text{King penguin}\}^{\text{Kingdom}}
&= \text{Animal} = \Omega.
\end{align*}
But maybe there is more than one way to characterize animals; we might want to classify them according to their geographic distribution. For instance, we can extend $\E$ with a chain for habitat characterization,
\begin{align*}
\E'= \E \cup \{0< \text{Ecoregion} < \text{Biome} < \text{Ecozone} < \text{Planet} \}. 
\end{align*}
For example, the approximation structure $A_{\E'}$ tells us that fur seals and penguins are geographically, if not biologically, close,
\begin{align*}
\{\text{Antarctic fur seal}\} &\nsubseteq \{\text{King penguin}\}^{\text{Class}}, \\
\{\text{Antarctic fur seal}\} &\subseteq \{\text{King penguin}\}^{\text{Phylum}},\\
\{\text{Antarctic fur seal} \}&\subseteq \{\text{King penguin}\}^{\text{Ecoregion}}.
\end{align*}
Note that $W^\eps \subseteq W^{\eps'}$ only when  $\eps \leq \eps'$, and that there is no operation $+$ defined on $\E'$, so we cannot talk of properties such as the triangle inequality for this approximation structure.
\end{bigexample}

Some state spaces have a notion of proximity between states. For such notions to be meaningful, they should express, at some level, how easy it is to distinguish the states. They may be very diverse: they could be defined relative to an agent that is constrained in her observations, or they could describe an arbitrary coarse-graining of knowledge according to some criteria. Examples of such notions would be the trace distance in quantum theory, or the distinguishability based on a particular observable feature or criterion.

In the simple case of the trace distance, there is a real number $\eps \in [0,1]$ that parametrizes neighbourhoods of states, like $\ball^\eps(\rho)$, according to precision. 
The neighbourhoods are sets of states, or specifications|special specifications that characterize the metric of the space induced by a given distinguishability criterion. This is in line with Desideratum~\ref{desi_knowledge}: specifications, the same objects that already describe the resources in our framework, can be used to model such coarse-grained descriptions.

We call the structure of such neighbourhoods or coarse-grainings the \emph{approximation structure} of the specification space. Then we generalize this notion by dropping the restriction that  neighbourhoods be parametrized by a real number; indeed, we only require that the parameters be partially ordered. 
For example, the parameters could reflect more than one criterion of precision, and therefore  not have a total order (see Example~\ref{ex:zoo_approximation}). For simplicity of notation, we use $W^\eps = \ball^\eps (W)$.

\begin{definition}[Approximation structures]
\label{def:approximations}
Let $(S^\Omega, \cT)$ be a resource theory and let $(\E, \leq)$ be a set with a partial order.
An approximation structure is a family  $A_{\E}(\Omega) = \{\ball^\eps\}_{\eps \in \E}$ of inflating endomorphisms,
\begin{align*}
\ball^\eps: \ S^\Omega &\to S^\Omega \\
W &\mapsto W^\eps \ \supseteq W,
\end{align*}
such that
\begin{enumerate}
\item for all $ W \in S^\Omega$ and $\eps, \eps' \in \E$,\quad  $\eps \leq \eps' \implies W^\eps \subseteq W^{\eps'}$, and
\item  there exists a \emph{saturating element} $ \eps_{\max} \in \E$ such that  for all $W \in S^\Omega$, \quad  $W^{\eps_{\max}}= \Omega$.
\end{enumerate}

We call the parameter $\eps$ the \emph{precision} of the approximation, and the specification $W^\eps$ the \emph{$\eps$-ball}, \emph{$\eps$-neighbourhood} or \emph{$\eps$-approximation} of $W$.

If there exists an element $0\in \E$ such that $W^0 = W$ for all specifications, we say that the structure is \emph{attainable}.
\end{definition} 

When applicable, we may also demand triangle inequalities (e.g.~for the trace distance). 
For an example of an unattainable approximation structure, take $\Omega = \mathbb R$ and $x^\eps$ being the $10^{- \eps}$-digit approximation of a real number $x$. No finite number of digits would allow us to reach $x$ in this way. This and other examples,  definitions and results are explored in  Appendix~\ref{appendix:approximations}.

Now we may characterize resources and resource theories with respect to a given approximation structure. For example, quantum mechanics is stable under the trace distance, because TPCP maps, being linear, cannot increase the distance between two states. In classical mechanics, on the other hand, there are many examples of chaotic theories under specific metrics.

\begin{definition}[Stability and robustness]
We say that  a resource theory $(S^\Omega, \cT)$ is \emph{stable} according to an approximation structure  $A_{\E}(\Omega)$ if 
$$
V \to W \ \implies V^\eps \to W^\eps,
$$
for all $V,W \in S^\Omega, \eps \in \E$.

We say that a resource $V \in S^\Omega$ is at least \emph{$\eps$-robust} if $V^\eps$ is not a free resource.
\end{definition}

One natural direction is to quantify how stable or  chaotic a theory is, by bounding the maximum divergence achievable in the theory. This is similar to the usage of  Lyapunov exponents, but it may be generalized to the case where  $\E$ is not a commutative monoid.
The robustness of a resource can be related to the fine-tuning of theories. 
The following proposition ensures that, for stable theories, one cannot obtain robust resources by processing resources that are not robust.

\begin{restatable}[Stability and robustness]{proposition}{propRobustness}

\label{prop:robustness}

Let  $(S^\Omega, \cT)$ be a stable resource theory according to an approximation structure $A_{\E}(\Omega)$. Then, for any $V,W\in S^\Omega$ such that  $V \to W$, if $W$ is $\eps$-robust then $V$ is also $\eps$-robust. 
\end{restatable}

Approximation structures can be carried over through embeddings. The following proposition guarantees that given an approximation structure on a specification space $S^\Omega$, we are able to construct an approximation structure on any smaller specification space $S^\Sigma$ that is related to $S^\Omega$ via an intensive embedding.

\begin{restatable}[Reduced approximation structures]{proposition}{propReducedApproximations} 
\label{prop:reduced_approximations}
Let $S^\Sigma$ be a specification space equipped with an approximation structure $A_\E(\Sigma)$, which is parametrized by a poset $\E$. Let there be $S^\Omega$ embedded in $S^\Sigma$ via $(\e, \h)$. 
Then there exists a  approximation structure $A_ \E$ in $S^\Omega$ also parametrized by $\E$, such that $ \e(V^\eps) \supseteq (\e(V))^\eps $.

Any saturating element $\epsilon_{\max} \in\E$ is still saturating for the reduced approximation structure.  Similarly, if $A_\E(\Sigma)$ is attainable, so is $A_\E(\Omega)$. 
\end{restatable}

\section{Probability and convexity}
\label{sec:convexity}

An issue that we have not discussed so far is that of probabilistic knowledge and convexity. The reader may have wondered how an agent may express probabilistic knowledge of reality. At first sight, specifications only allow for statements of the sort `I know that the state is either $\rho$ or $\sigma$', and not `I know that it is $\rho$ with probability $p$ or $\sigma$ with probability $1-p$'. 
Indeed, our formalism does not impose a notion of probability; instead, it assumes that if the agent has such a notion, it is already incorporated in an underlying convex state space $\Omega$. For example, in quantum theory,  the space of density matrices $\Omega$ is convex: for any two $\rho, \sigma \in \Omega$, there is a density matrix  $p\ \rho + (1-p)\ \sigma \in \Omega$  that expresses the probabilistic knowledge discussed above. Therefore, there is no need to impose another subjective notion of probability at the level of specifications. 

However, the state space might be convex for other reasons than to express probabilistic knowledge: for example, in chemistry we might want to describe mixtures of substances through convex combinations, such as a solution of $1$ gram of salt in $50$ grams of water. The structure of such mixtures will be much the same as the structure of probabilistic mixtures, and so we treat both types of convexity in this section. 
In either case, we may define the convex mixture\footnote{For more complete results, see Appendix~\ref{appendix:convexity}.  In particular, we use a generalized notion of convexity for sets that may not have all the structure of a real vector space.} of two specifications as 
$$ p\ V + (1-p) \ W
  := \bigcup_{\nu \in V} \bigcup_{ \omega \in W }
   \{p\ \nu + (1-p)\ \omega \}.$$

Let us have a closer look at probabilistic mixtures.
Consider the simple specification $V= \{ \omega, \nu\}$,  interpreted as `I might have state $\omega$ or state $\nu$'. An agent with some notion of probability might also read it as `I have either state $\omega$ or $\nu$, and I do not know the probability distribution over those two possibilities'. For him, an equivalent description of $V$ would be 
$$V^\P = \bigcup_{0\leq p \leq 1} \{p \ \omega + (1-p) \ \nu \}.$$
We call this new specification the \emph{probabilistic version} of $V$, and it corresponds to the convex hull of $V$. 
It is important to note that, unless a specification $V$ is already convex, $V$ and $V^\P$ are not strictly equal; their equivalence is given from an interpretation of probability that we take as natural.\footnote{A Bayesian might interpret this as `the specification is compatible with all possible priors.'} More precisely, this \emph{probabilistic equivalence} is given by an equivalence relation $\sim_\P$ defined by $V\sim_\P V^\P$. Given such an interpretation, however, we may choose to work in the probabilistic quotient space, that is, we can always identify specifications that are equivalent under $\sim_\P$. In Appendix  \ref{appendix:convexity} we  derive consistency conditions to work in the quotient space.

For example, if $\Omega$ is the set of possible outcomes of an exam, and my two specifications are $V = \{\text{Pass}\}$ and $W= \{$Pass, Fail$\}$, then the convex mixture 
\begin{align*}
0.95\, V + 0.05\,  W 
&= \{\text{Pass}\} \cup \{ 0.95 \ \text{Pass} + 0.05 \ \text{Fail} \}  \\
&\sim_\P \bigcup_{p \geq 0.95} \{p \text{ Pass} + (1-p)\ \text{Fail}\}
\end{align*}
corresponds to the specification `Pass with probability at least $95\%$'.  In other words, our framework allows us to formalize fuzzy knowledge of an agent that may not have a total prior distribution over a set of states or events, but only a partial prior of the sort `I know that the probability is between $p$ and $p+\delta$.'

\begin{bigexample}{Quantum convexity}{mixtures_specifications}
Let $\Omega$ be the set of all density operators for a qubit. This set is convex|the extremal points are the set of all pure states (the surface of the Bloch sphere).
 Let $V= \{\pure0, \pure1 \} \in S^\Omega$. Then 
\begin{align*}
 V^\P = \bigcup_{0\leq p \leq 1} \{p\ \pure0 + (1-p)\ \pure1 \}.
\end{align*}
We can denote the elements of $V^\P$ by $\{v_p\}_{0\leq p  \leq 1}$. For instance, $v_{0.7} = \{0.7 \ \pure0 + 0.3\ \pure1 \} $. 
In particular,  $ v_{1/2}  = \{\frac12 (\pure0+\pure1 )\}$ is the fully mixed state, which corresponds to taking the uniform distribution. 
Note that this is still more precise knowledge than not knowing the distribution: $\{v_{1/2}\} \subset V^\P$.

\end{bigexample}

Now we may study whether a set of transformations preserves convexity. 
For convex mixtures in general, this is not always desirable. For instance, in the example of a solution of salt in water the properties of the resulting substance, such as the electrical conductivity, might not directly reflect the properties of the individual substances. However, in the case of probabilistic mixtures, the transformations should 
respect the character of subjective probabilistic knowledge of mixtures. In this case, the requirement of convexity-preserving transformations is analogous in spirit to asking transformations to be specification homomorphisms.

\begin{definition}[Convex resource theories] \label{def:convex_rt_main}
A resource theory $(S^\Omega, \cT)$ is \emph{convex} if the state space $\Omega$ is convex  and all transformations in $\cT$ preserve convexity, that is, 
    \begin{align*}
    f (p \ V + (1-p) \ W) = p\ f(V) + (1-p)\ f(W),
    \end{align*} 
for all $f \in \cT$, all $V, W \in S^\Omega$,  and 
$ p \in [0,1]$.

The resource theory  is \emph{doubly convex} if, in addition, the set of transformations $\cT$ is also convex, and convex combination of transformations behave naturally (see Appendix~\ref{appendix:convexity} for a full definition), in particular 
$$ (p \ f + (1-p) \ g) (V) 
\sim_\P  p \ f(V) + (1-p) \ g(V).$$
\end{definition}

For example, the set of density operations over a fixed Hilbert space is convex, and quantum transformations (TPCPMs) are linear, and so convexity-preserving,  $\E(p \ \rho + (1-p) \ \sigma) = p\ \E(\rho) + (1-p)\  \E(\sigma)$. Therefore all quantum resource theories are convex. Furthermore, quantum theory  is doubly convex, and so is LOCC. The theory of unitary operations is not doubly convex (because we are not allowed to mix unitaries). 
In Appendix~\ref{appendix:convexity} we characterize doubly convex resource theories. The following result is of particular interest.
\begin{restatable}[Convexity of free resources]{theorem}{thmFreeConvex}
\label{thm:free_convex}
Let $(S^\Omega, \cT)$ be a doubly convex resource theory. Then, the set of free resources is convex (under probabilistic equivalence), that is, for any two free resources  $V$ and $W$,  $(p \ V + (1-p) \ W)^\P$ is also a free resource. 

In particular, the set of free states is convex: if $\{\nu\}$ and $\{\omega\}$ are free, then $\{p \ \nu + (1-p) \ \omega\}$ is also free.

\end{restatable}  

If a state space is 
convex, then we may ask whether a restricted agent also describes his knowledge with a convex state space. It turns out that whether or not this is true only depends on the lumping map induced by the embedding, which should satisfy a weaker version of convexity preservation. 
In quantum theory, this condition applies to the embeddings given by the partial trace. 

\begin{restatable}[Convex embeddings]{theorem}{thmConvexReduced}
\label{thm:embedding_convex}

Let $\Omega$ be a convex state space, and let $\Lump$ be a lumping in  $S^\Omega$ inducing a reduced specification space $S^T$. If $\Lump$ satisfies  
\begin{align*}
\Lump (p \ V + &(1-p) \ W)\  \supseteq\\
&
p\  \Lump (V) + (1-p)\  \Lump (W),
\end{align*}
then the reduced state space $T$ is also convex.

\end{restatable}

In particular, if $(S^\Omega, \cT)$ is a doubly convex resource theory, and general consistency conditions are satisfied between convex combinations and the embedding, then the induced restricted resource theory  $(S^T, \widetilde \cT)$ is also doubly convex. For details of this, see Appendix~\ref{appendix:convexity}.

\section{Discussion}
\label{sec:conclusions}

Guided by operational principles and explicit desiderata, we have developed a generalized framework of resource theories that extends the range of applicability of resource theories by making the subjectivity of the agent explicit. We wrap up by returning to our original desiderata, and seeing how they are satisfied by our results. 
Then we discuss relations to other approaches and finally we map future directions of research.

\subsection{Meeting the desiderata}

\desiknowledge*

Specification spaces are a simple yet powerful tool to model descriptions of resources. A state of knowledge or specification is formalized as a set of possible physical states admitted by an agent (Definition~\ref{def:specification_space}).
Combining knowledge is achieved by intersecting those sets, and forgetting by order-inflating functions (Definition~\ref{def:specification_space}). 
Approximation structures formalize the coarse-graining of knowledge through arbitrary criteria, reflecting the distinguishability of resources according to particular observable features (Definition~\ref{def:approximations}).
Specification spaces also accommodate the notion of probabilistic knowledge via convexity of the state space, and probabilistic equivalence relations (Definition~\ref{def:convex_rt_main}). 

Physical transformations should respect the nature and structure of knowledge: for this we require them to be order homomorphisms and, in the case of probabilistic knowledge, convexity-preserving (Definitions~\ref{def:resource_theory} and \ref{def:convex_rt_main}). This is also taken into account in the pre-order induced in a resource theory by
specification transformations: forgetting information always comes for free (Definition~\ref{def_construction}). 
We see how a resource theory behaves under approximation structures and probabilistic knowledge by studying robustness and convexity of free resources (Proposition~\ref{prop:robustness} and Theorem~\ref{thm:free_convex} respectively).

\desirelating*

Theories that differ in their language to describe resources can be related via specification embeddings (Definition~\ref{def:specification_embeddings}).
Theorem~\ref{prop:general_embeddings} tells us that a relation between different resources can always be established combining intensive embeddings (which link coarse- and fine-grained descriptions of the same resources) and extensive embeddings  (which add entirely new resources). We show how to 
build intensive embeddings from arbitrary coarse-grainings of information  and how to nest successive embeddings (Propositions~\ref{thm:lumping} and~\ref{prop:nested} respectively). Finally, we formalized the notion of local states of knowledge|or equivalently specifications that correspond to the knowledge of a limited agent (Definition~\ref{def:local_specification}).

More generally, agents that differ in both language and action, such as local agents acting within a global theory, can be modelled via restricted resource theories (Definition~\ref{def:restricted}).  In particular we can combine the restrictions of different resource theories (Definition~\ref{def:combined_rt}). 
We characterize embeddings that respect approximation structures and convexity, such that these aspects of knowledge are carried over between related specification spaces (Proposition~\ref{prop:reduced_approximations} and Theorem~\ref{thm:embedding_convex}).

\desicomposition*

We find an operational subsystem structure  in a global resource theory by building local resource theories based only on  commutation relations between physical transformations (Theorem~\ref{thm:transformations_independent_agents}). 
In particular, quantum subsystems and the partial trace can be derived in this way (Theorem~\ref{prop:partial_trace}).
We further explore the subsystem structure  induced by commutation relations in Definitions~\ref{def:effective_resource_theory} and~\ref{def:subsystems_transformations}  and in  Proposition~\ref{prop:complete_lattice}.

This is not necessarily the unique way of finding a subsystem structure; more generally, we find two aspects of locality that are used in traditional building-block approaches to resource theories. The first is that local descriptions are independent of each other, which ensures that  they can  always be combined (Definition~\ref{def:independent_embeddings} and Proposition~\ref{prop:free_composition}). This condition does not hold in the presence of global constraints. The second aspect of locality is that the descriptions of one local agent are independent of the actions of another (Definition~\ref{def:independent_agents}). This ensures that we can consistently process the actions of independent agents locally  (Theorem~\ref{thm:transformations_independent_agents}).

These two notions in their strongest sense are taken for granted in traditional approaches; however, we will see in Part II of this work that they are not strictly necessary in order to define relevant concepts such as copies of resources, correlations, memories or catalysis. To pave the road to that work we explore weaker and asymmetric versions of the two notions of independence in Definitions~\ref{def:independent_embeddings} and~\ref{def:independent_agents} and in Proposition~\ref{prop:independent_processing}. These allow us to introduce the notion of effective resource theory  (induced for example by a catalyst) in Definition~\ref{def:effective_resource_theory}.

\subsection{Relation to existing work}
\label{sec:literature}

\subsubsection{Theories of knowledge}
\label{sec:theories_knowledge}

There are many ways how one could describe an agent's subjective knowledge, and there is a huge body of work within a variety of fields dedicated to this question, ranging from philosophy to economics, engineering and computer science. The present  work bridges those disciplines with quantum information theory and resource theories, and in the future, more sophisticated ideas from those fields could be incorporated in our framework.

\paragraph{Modal logic.}
The idea to represent knowledge in terms of sets of possible states is not new. In fact, the analogous concept of \emph{possible worlds} dates back to Leibniz~\cite{Leibniz1739}, and was most prominently brought forward by Kripke~\cite{Kripke1963} and Lewis~\cite{Lewis2013}. Kripke  developed the semantics for \emph{modal logic}, which quantifies over possible worlds to express knowledge and inference (for a historical survey on modal logic, see~\cite{Goldblatt2003}). 

\paragraph{Epistemic logic.}
Building up on these notions, \emph{epistemic logic}~\cite{Fagin1995, Halpern1995} pursues the task of characterizing knowledge, based on the formalism of Kripke structures. Kripke structures are a way of modelling the knowledge of different agents, in which one starts with a set of possible states (in our case $\Omega$), and models an agent's knowledge by means of so-called \emph{possibility relations}\footnote{The exact formalism   is different but  isomorphic to this.} $R_i(\nu,\omega) \in \{0,1\}$  that express whether or not agent $i$ can exclude the possibility $\omega\in\Omega$ if the true state is $\nu\in\Omega$. Note in particular that $R_i$ is not necessarily symmetric. 
Given a particular true state $\nu$, the knowledge of agent $i$ essentially corresponds to a subset of possible states  $V_{\nu, i}= \{ \omega  \in \Omega: R_i(\nu, \omega) =1\}$, and is thus just a specification. In our language, we can understand the possibility relations as inducing an approximation structure, with $\E =\{i\}_i$, and $\{\nu\}^i = V_{\nu, i}$.
Epistemic logic then analyses the propositions known by a particular agent: agent $i$ knows a proposition $\phi$ if and only if $\phi(\omega)$ is true in all states $\omega \in V_{\nu, i}$ that are indistinguishable for the agent $i$ from the true state $\nu$. In our framework, this could be modeled by means of homomorphisms $\phi: S^\Omega \to S^{\{ 0,1\}}$. 
Within epistemic logic, further statements about an agent's knowledge follow from studying the semantics of this model of Kripke structures. A set of complete axioms for this semantics have been characterized by Hintikka~\cite{Hintikka1962}.

\paragraph{Linear and dynamic epistemic logic}
It would  be interesting to explore in more detail how our framework connects to modern developments in epistemic logic. In particular, one could draw ideas from \emph{dynamic epistemic logic}, which studies knowledge update in such models (see e.g.\ \cite{Ditmarsch2007}, or for a recent algebraic analysis \cite{Baltag2007}). 
Related to this, \emph{linear logic}~\cite{Girard1987, Girard1995}  tries to extend classical logic to situations where inferences can change the underlying premise, similar in spirit to resource theories. A good review of the connection between linear logic and resource theories has been given in \cite{Fritz2015}. 
In particular, notions of \emph{common knowledge} shared by several agents, and the way it evolves, are fundamental to communication, cryptography and information security settings.  In this context, it would be promising to look at some of the puzzles that are studied in dynamic epistemic logic such as the muddy children's puzzle  ~\cite{Fagin1995,Halpern1995} and understand whether our framework presents an advantage in treating them.

\subsubsection{Notions of probability}

\paragraph{Bayesian formalism.}
Our framework includes the possibility of  a convex state space to encode subjective probabilities, taking the specification space as a structure on top of that. In this way, specifications can be used to model Bayesian agents. Moreover, we may explore interesting cases, like a bounded Bayesian agent that is limited in her memory and information processing capacities.
For example, even if the agent has good reason to model local knowledge with a probability distribution or a density matrix $\rho_A$, it might be computationally hard (and artificial) to pick  and store a prior over compatible global states.
By means of specifications like $\local {\rho_A}$, such an agent may still predict the effect of global actions on her local state of knowledge, consistent with any possible prior, without over-modelling her knowledge. 
In the introduction we argued that a Bayesian approach might not always be suitable to represent a realistic agent's knowledge. In the following we will briefly review theories that weaken the axioms of probability theory towards a more operational approach.

\paragraph{Dempster-Shafer theory, subjective logic and plausibility.}
Dempster-Shafer theory studies an agent's degree of belief on propositions (which are essentially specifications).
To do so, it introduces a new layer on top of the basic structure of epistemic logic:  a set of \emph{belief functions} which follow similar but weaker axioms as probabilities in a Bayesian framework~\cite{Dempster1967, Shafer1976}. 
Based on these belief functions, an agent can determine lower and upper probabilities (or degrees of confidence) for particular propositions. Dempster-Shafer theory also guides the agent on how to update these beliefs  as he learns more, combines knowledge and forgets or coarse-grains the knowledge, much like our  intensive specification embeddings. For an extensive review, see~\cite{Yager2008};
this approach has been recently taken up by the quantum logic community~\cite{Fritz2015a}.
Similarly, \emph{subjective logic}~\cite{Josang2001} quantifies agents' knowledge about certain propositions by means of a degree of belief $b$, disbelief $d$ and a remaining level of uncertainty $u=1-p-d$. Interestingly, it also studies the interaction of many agents and resulting trust networks~\cite{Josang2006}.
As future work, it would be interesting to apply these ideas to our framework. For example, we could investigate the impact of   our notions of independence on the rules for updating belief.

\paragraph{Fuzzy logic.}
Finally, one could draw links to fuzzy logic~\cite{Novak1999,Zadeh1965,Zimmermann2001}, where the premise is that propositions can take truth values anywhere between $0$ and $1$, or elements can belong to a set to a degree between $0$ and $1$ respectively.  
In our framework, this could be used construct a Bayesian model on top of specifications, in which different elements in the specification would receive different degrees of certainty. However, such a model could already be simulated by a convex state space, and so it is not clear if more insight can be gained from this approach.

\subsubsection{Approximation structures}

In order to quantify proximity between states and specifications, we have introduced a definition of approximation structures, which parametrize neighbourhoods of states in a  very general way. For instance, we do not even  require an underlying symmetric distance measure, so notions such as the work distance from \cite{Brandao2013b} are covered. Indeed, its generality stands out among other approaches to proximity in the literature.

\paragraph{Topology.} By defining approximations through inflating homomorphisms, our definition is analogous to the notion of \emph{neighbourhoods} in topology (see e.g.\ \cite{Munkres2000}). However, the general structure of our approximation parameters $(\E, \leq)$, characterized by a partial order, allows us to formalize operational notions of proximity which do not constrain the underlying structure of the topology (such as a metric) but nonetheless provide structure on the neighborhoods. 
In the context of topology, one could also compare our approach to the concept of \emph{proximity spaces}, which give a way of characterizing the nearness of sets~\cite{Naimpally1971}. However, this approach relies on a set of fairly restrictive axioms (such as symmetry), and as such seems less general than similar relations that follow from our approximation structures.

\paragraph{Rough sets.} The approach that is probably most interesting to compare to our notion of approximation is given by the \emph{rough sets} introduced by Pawlak~\cite{Pawlak1973, Pawlak1982}. In this approach, sets are approximated by means of \emph{lower} and \emph{upper approximations}, which result from indistinguishability relations according to particular \emph{attributes}. In our language, these attributes can be understood as lumpings, which correspond to a special case of our approximation homomorphisms, so that they induce equivalence classes of indiscernible objects according to the respective criteria. It would be interesting to generalize the ideas of lower and upper approximations for our definitions of approximation structures.

\subsubsection{Resource theories}

We have referred to previous work on resource theories throughout this paper. However, some recent generalizations of resource theories \cite{Brandao2015, Coecke2014, Fritz2015} deserve a more detailed analysis. 

\paragraph{Asymptotic reversibility.}
Ref.~\cite{Brandao2015}  addresses the question of the asymptotic reversibility of a resource theory, showing that such reversibility follows from a maximal set of allowed transformations. It would be interesting to connect these ideas to our framework, and see if their result can be generalized beyond quantum resource theories. One could draw inspiration from their work to define conversion rates in our framework, and see how their expressions can be adapted to fit generalized resource theories.  
Another notable connection of \cite{Brandao2015} to our paper lies in the set of free states: while in their work, convexity of the free resources is assumed, we derive it in Theorem \ref{thm:free_convex}. 

\paragraph{Categorical approach.} 
Recent works~\cite{Coecke2014, Fritz2015} present an approach connecting resource theories to category theory. These papers in general provide a very good introduction to the structure of resource theories as they have been studied to date, and the interested reader can find many references to related work therein. In particular, Section $10$ of \cite{Fritz2015} gives an excellent overview of previous research in the various areas connected to resource theories including linear logic and constructor theory. In fact, this paper discusses the various aspects of resource theories in detail and points out valuable links to related concepts within mathematics and wider research areas throughout, referring to von Neumann and Morgenstern in decision theory~\cite{vonNeumann1947} and Lieb and Yngvason in the foundations of thermodynamics~\cite{Lieb1999}. As we have also mentioned in the introduction, the approach in \cite{Fritz2015} differs from ours in that it does not treat knowledge and the subjectivity of an agent explicitly, and takes a bottom-up approach to subsystems and composition.

\paragraph{Abstract cryptography.}
In cryptography resource theories emerge naturally: constraints and subjectivity of agents are the name of the game, and always treated explicitly. 
The present work is heavily influenced by the ideas of abstract  cryptography  found in~\cite{Maurer2011}.
In this work, cryptographic constructions are phrased in terms of resource theories: for example, the one-time-pad encryption scheme can be seen as a transformation of an initial resource (an authenticated channel together with a secret key) into a target resource (a private channel). 
Furthermore,  resources may be formulated by means of specifications. 
For example, an authenticated channel between two parties, Alice and Bob, may leak information to a third party, Eve. It is implicit in this notion that
 Alice and Bob do not know how much is leaked to Eve (for example, they cannot use the channel as a means to communicate with Eve). 
The idea of authenticated channel may therefore be modeled as the set of all channels between Alice and Bob with different levels of leakage to Eve|a specification.

\subsection{Directions}

With the basic framework set up, there is much to explore.
Concepts in traditional resource theories like memories, catalysis, currencies or conversion rates should be generalized to our framework, where in particular one uses generalized notions of subsystems|this is the main subject of the upcoming Part II of this work. On the other hand, it would also be very interesting to explore monotones on the pre-order on the specification space, and see whether results from traditional resource theories (like the uniqueness of monotones~\cite{Lieb1999, HorodeckiOppenheim2013}) can be carried over to more general settings.
Also, the framework itself leaves interesting directions open to explore further, such as to introduce a specification space of transformations. Finally, one could look at particular applications of our framework in more detail, such as the connection between microscopic and macroscopic thermodynamics or cryptographic settings. In the following we detail some of these directions.

\subsubsection{Memories and catalysis}

The concept of memories and catalysis in traditional resource theories relies heavily on the strict notion of subsystems. In traditional resource theories, a memory corresponds to a system that stores information, which can then be used to perform certain tasks more efficiently, like the erasure of another system that is correlated with the memory~\cite{delRio2011}. Since this notion uses the idea of subsystems, it will be interesting to reformulate results concerning memories in our framework | in particular, one can highlight which aspects of subsystems are sufficient for defining a memory.

Similarly, catalysis in traditional resource theories is the phenomenon that additional resources can help speed up, or make possible a process which otherwise would have been more costly or impossible. This concept is well-known in chemistry (for example, the reaction of hydrogen peroxide into water is facilitated by manganese dioxide), but has also received attention recently in quantum resource theories~\cite{Turgut2007, Aberg2013, Brandao2013b,Ng2014}. Again, catalysis in traditional resource theories is defined in terms of an additional subsystem that is appended to the system in consideration; in our framework we can study which aspect of generalized subsystems is necessary to capture the concept of catalysis. In particular, catalysis induces an effective resource theory on the remaining systems (similarly to \cite{Fritz2015}, where catalysis gives rise to a regularization); we make this idea precise in Part II of this work.

\subsubsection{Composition and copies of local resources}
Another concept central to traditional resource theories is the idea of composition and copies of local resources, which can then be used to define conversion rates between resources (see e.g.\ \cite{Brandao2015}). These concepts  make heavy use of the traditional subsystem structure, where local resources can be combined with the tensor product to produce uncorrelated composite resources. In Part II of this work, we
take a more operational approach that goes beyond the tensor product. For example, to compose two heat baths we could allow them to be slightly correlated, as long as these correlations are weak enough that they do not affect their joint behaviour. 
We are still able to find meaningful concepts of copies and conversion rates. One can then check if for familiar resource theories, asymptotic results like those of \cite{Brandao2015} could be recovered.

\subsubsection{Cost and yield of resources}

Once one has defined the notions of composition and copies of local resources, one can study the cost and yield of resources and resource transformations in terms of a ``standard'' resource, a currency. Such a currency allows to facilitate any resource transformation, given enough of it, and provides a useful means of quantifying the value of resources. This does not exist in all resource theories; however, many traditional resource theories are equipped with a currency. For example, in bi-partite LOCC maximally entangled pairs of qubits form a currency, in noisy operations it would be pure qubits, and in thermodynamics we may quantify resources by means of energy stored in an explicit battery system. 
In Part II of this work, we  formalize the notion of a currency in terms of generalized subsystems. We may then compare currencies that arise from copies of resources (such as pure or maximally entangled qubits) to those defined on one system (such as the energy of a battery). Finally, one can draw connections to monotones, conversion rates and reversibility, as well as to existing concepts such as the entropy meters in \cite{Lieb2014}.

\subsubsection{Monotones}
Another main direction for future work is to classify more explicitly the pre-order on specification space that is induced from a resource theory. The general question of whether or not an arbitrary transformation $V\to W$ is possible can be very  hard.
Hence, we should look for ways to answer this question in special cases, for instance by finding monotones. There are two steps to approach this issue: the first would be to exploit that some specifications have a concise description, and such descriptions can be reached from any specification by means of coarse-graining. In this context, it would be interesting to explore the ideas of lower and upper approximations introduced in the works on rough sets~\cite{Pawlak1973, Pawlak1982}. The second important step is that monotones along the pre-order that are known for traditional resource theories should be extended to specifications. If this is not possible directly, still lower and upper bounds to monotones could be found in terms of state space monotones. In some explicit cases, where the pre-order on state space is fully known and characterized, for example through majorization, a corresponding analogue for specification spaces could be identified.

\subsubsection{Specifications of transformations}
Another straightforward extension of our framework would be to consider specification spaces of allowed transformations, as well as formalizing approximation structures and embeddings on the level of transformations. A specification space of transformations $S^\cT$ would correspond to coarse-graining, or operationally, to uncertainty over which transformation exactly has been applied, for example after applying process tomography.  Mathematically, a specification of transformations would itself act as a homomorphism, that is, for $F=\{f,g\}, f,g\in\cT$ we shall demand that it acts as 
$F(V)= f(V)\cup g(V)$. 
Embeddings of transformations have been implicitly  treated in the context of restricted and local resource theories. More generally, this could be done via embeddings on $S^\cT$. Lastly, we may set up approximation structures on $S^\cT$  that reflect approximation structures on $S^\Omega$, such that for instance $f^\epsilon(V)=f(V^\epsilon)$, or $f(V)=W\implies f^\epsilon(V)=W^\epsilon$. When exploring these concepts, one should look for general duality between resources and transformations, that is conditions under which statements about specifications of transformations might be reduced to statements about resources.

\subsubsection{Concrete applications}
Finally, it remains to look at particular applications of our generalized framework where our results yields further insight. One example would be to study in more detail the connection between macroscopic thermodynamics and microscopic models such as classical statistical mechanics or noisy and thermal operations through embeddings, and see if this link can be made formal. In this context, one should formalize the insight that specification spaces make the equal a priori probability postulate for the microcanonical ensemble in statistical mechanics obsolete.
Another option is to explore applications to cryptography and information security in general, where our framework could be used to relate the knowledge of  adversarial agents.

\subsubsection{Foundations}

While resource theories may be purely operational frameworks, they can reveal fundamental aspects of the underlying physical theory. 
For example, the task of deriving quantum theory from a set of axioms~\cite{Chiribella2016, Masanes2013, Henson2015} is naturally cast in terms of a resource theory formalism. This may be done by starting from a general framework for probabilistic theories, where resources are sets of black boxes characterized by inputs,  outputs and probability distributions (expressing things like state preparation, a physical transformation or measurement  statistics), and transformations correspond to different ways of connecting inputs and outputs of those boxes (for example by preparing two local states and then  performing a joint measurement). 
One then imposes  subsequent operational constraints on both resources and allowed transformations (based on things like non-signalling, local tomography, no-cloning, or familiar macroscopic behaviour) until the resulting resource theory singles out quantum theory. 
It would be interesting to  extend this approach beyond probabilistic theories, by using specifications to describe arbitrary resources.

\newpage

\begin{acknowledgements}
We thank Fernando Brand\~ao, Bob Coecke, Philippe Faist, Tobias Fritz, Patrick Hayden, Philipp Kammerlander, Ueli Maurer, Christopher Portmann and Gilles P\"utz for insightful discussions and feedback, and Bruno Montalto, who years ago ranted `people say random a lot when they mean uniform', thus planting the seeds for this work. This work is partly based on LdR's PhD thesis.
We acknowledge support from the European Research Council (ERC) via grant No.~258932, from the Swiss
National Science Foundation through the National Centre of
Competence in Research \emph{Quantum Science and Technology}
(QSIT),  by  the  European  Commission  via  the  project \emph{RAQUEL}, and by the COST Action MP1209.
LdR further acknowledges support  from the ERC Advanced Grant NLST. 
\end{acknowledgements}

\onecolumngrid
\newpage
\appendix

\addcontentsline{toc}{section}{\sc{Appendix}}

\section{Basics of algebra and order theory, general statements}
\label{appendix:algebra_order}

In this appendix we review standard notions from algebra and order theory, and introduce a few general concepts and results of our own.

\subsection{Algebraic structures}

\begin{definition}[Basic algebraic structures]
  A pair $(S,\cdot )$ composed of a set $S$ and a binary operation $\cdot \ : S \times S \to S$ is:
  \begin{itemize}
  \item a \emph{magma} if $S$ is closed under $\cdot$, $x\cdot y \in S$,  for all $ x,y \in S$;
  \item a \emph{semigroup} if it is an associative magma, $x \cdot (y \cdot z) = (x \cdot y) \cdot z$,  for all $ x,y,z \in S$;
  \item a \emph{band} if it is an \emph{idempotent} semigroup, $x \cdot x =x$, for all $x \in S$;
  \item a \emph{monoid} if it is a semigroup with an identity element $1\in S$,  $1\cdot x = x \cdot 1 =x$, 
  for all $ x,y,z \in S$; 
  \item a \emph{group} if it is a monoid in which every element  $x \in S$ has an inverse $x^{-1} \in S$, such that $x^{-1}\cdot x = x \cdot x^{-1} =1$.
    \end{itemize}
\end{definition}

\begin{definition}[Submonoids]
Let $(S, \cdot)$ be a monoid. A \emph{submonoid} of $S$ is a subset $A \subseteq S$ such that $(A, \cdot)$ is a monoid. If $ \{1\} \subset A \subset S$ we say that it is a \emph{proper} submonoid; otherwise it is a \emph{trivial} submonoid.
\end{definition}

The above definition applies analogously to more complex algebraic structures, like groups.

\begin{lemma}
\label{lemma:monoid_intersection}
Let $(S, \cdot)$ be a monoid, and $(A,\cdot)$ and $(B,\cdot)$ be two submonoids of $(S, \cdot)$. Then $(A\cap B,\cdot)$ is also a submonoid of $(S, \cdot)$.
\end{lemma}
\begin{proof}
As the identity is element of both $A$ and $B$, it is also in $A\cap B$. It is left to show that if $a,b \in A\cap B$, then also $a\cdot b\in A\cap B$. But since $a,b$ are elements of both $A$ and $B$, and $(A,\cdot)$ and $(B,\cdot)$ are monoids, also $a\cdot b\in A$ and $a\cdot b\in B$. Hence $a\cdot b\in A\cap B$.
\end{proof}

\begin{definition}[Monoid homomorphisms]
Let $(S, \cdot)$ and $(T, \star)$ be two monoids with identities $1_S$ and $1_T$ respectively. A map $f$ from $S$ to $T$ is a \emph{monoid homomorphism} if $f(a\cdot b)= f(a)\star f(b)$ and $f(1_S) = 1_T$.
If in addition $f$ is bijective, then $f$ is called a \emph{monoid isomorphism}, and the two monoids are said to be \emph{isomorphic}.
\end{definition}

\subsection{Order structures}

\begin{definition}[Pre and partial orders]
A \emph{preorder} $\leq$ on a set $S$ is a is a binary relation between elements of $S$ that is reflexive ($x \leq x$) and transitive ($x\leq y \wedge y\leq z \implies x\leq z$).

A \emph{partial order} $\leq$ is a an antisymmetric  preorder ($x \leq y \wedge y \leq x \implies x=y$). We call the structure $(S, \leq)$ a partially ordered set, or \emph{poset}. 

A \emph{strict partial} order $<$ is a binary relation that is irreflexive, antisymmetric (if $x<y$ holds, then $y<x$ doesn't) and transitive. It may always be built from any partial order $\leq$ by excluding reflexivity (we define it as $x<y: \ x\leq y \ $ and $ \  x\neq y$). 

The \emph{dual relations} $\geq$ and $>$ are built from $\leq$ and $<$ respectively analogously ($x \geq y \Leftrightarrow y \leq x$).

\end{definition}

Now we introduce terminology for relations between elements of a poset, as well as special elements and subsets. For intuitive examples, see Figs.~\ref{fig:meet_join} and~\ref{fig:ideal_filter}.

\begin{definition}[Relations between elements of a poset]
Let $(S, \leq)$ be a poset.
Two elements $x, y \in S$ are said \emph{comparable} if either $x \leq y$ or $y \leq x$.

If all elements of $S$ are comparable, then $\leq$ is a \emph{total order}. A \emph{chain} is a totally ordered subset of $S$.

Let $X \subseteq S$ be a subset of $S$. We say that an element $y \in S$ is:
\begin{itemize}
\item an \emph{upper bound} of $X$,  $y \geq X$, if $ x\leq y$, for all $x \in X $;
\item a \emph{lower bound} of $X$, $y \leq X$, if  $ y\leq x$, for all $x \in X $;
\item the \emph{join} of $X$, $y = \vee X$, if it is the least upper bound of $X$, that is, $y \geq X$ and $y\leq z$, for all $ z \geq X$;
\item the \emph{meet} of $X$, $y = \wedge X$, if it is the greatest lower bound of $X$, that is,  $y \leq X$ and $y\geq z$, for all $ z \leq X$;
\item the \emph{top} or greatest element of $S$, $y=1$, if $y\geq S$;
\item the \emph{bottom} or least element of $S$, $y=0$, if $y\leq  S$.
\end{itemize}

\end{definition}

\begin{figure}[t]

\begin{center}
\includegraphics[width=0.4\textwidth]{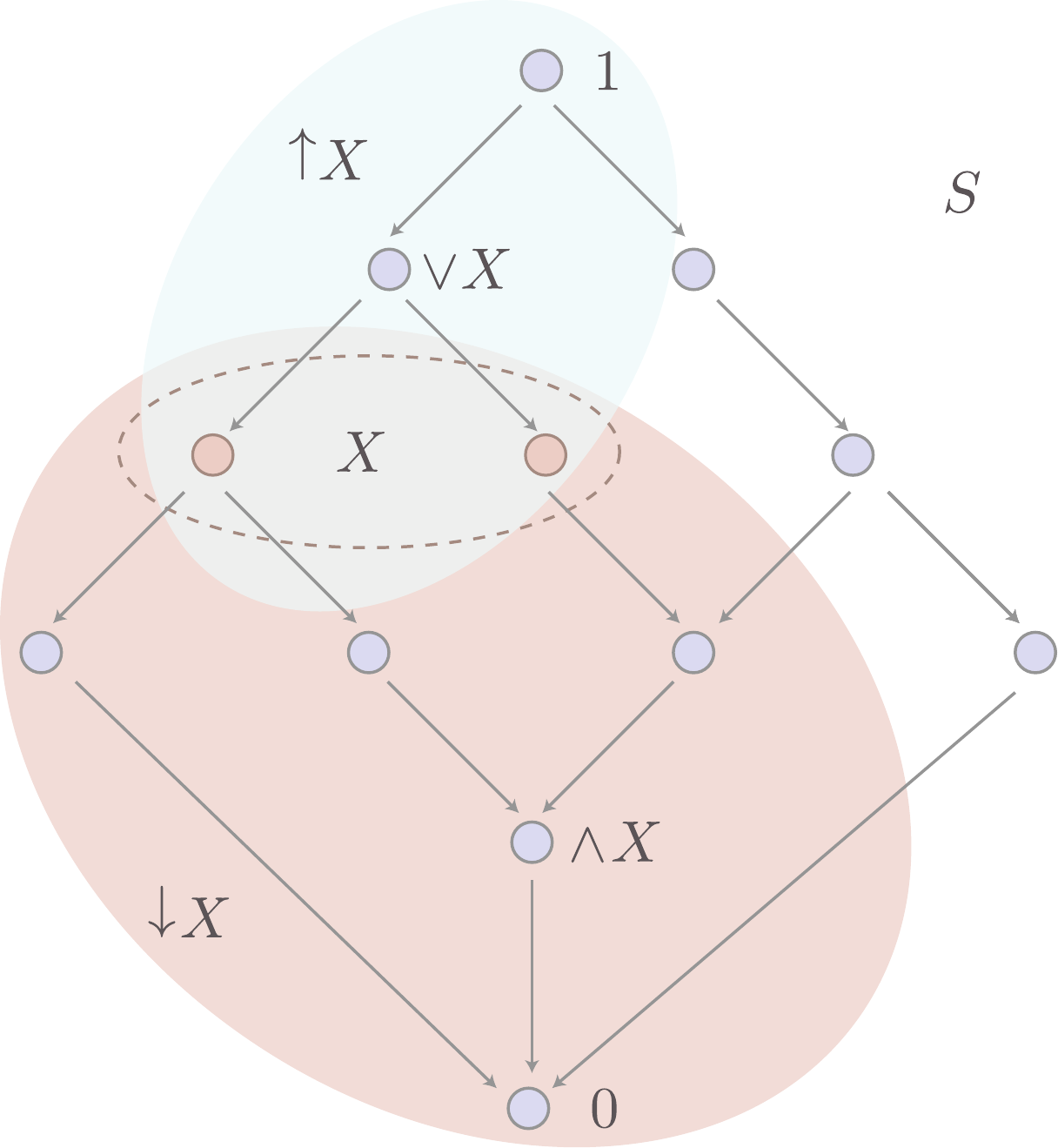}
\end{center}

\caption{
	{\bf Join, meet, top and bottom.}
	In this diagram of a poset $S$, arrows represent the partial order: $x \to y \Leftrightarrow x \geq y$. 
	We start with the two-element set $X$ (dashed line), and we identify the sets $\uparrow X$ (in blue) and $\downarrow X$ (in orange). We can also identify the join and meet of $X$: $\vee X$ and $\wedge X$. In this example, none of them is contained in $X$, but that is not always the case. Finally, we can see the top and bottom of $S$, $1$ and $0$. 
}

\label{fig:meet_join}
\end{figure}

\begin{figure}[t]

\begin{center}
\includegraphics[width=0.9\textwidth]{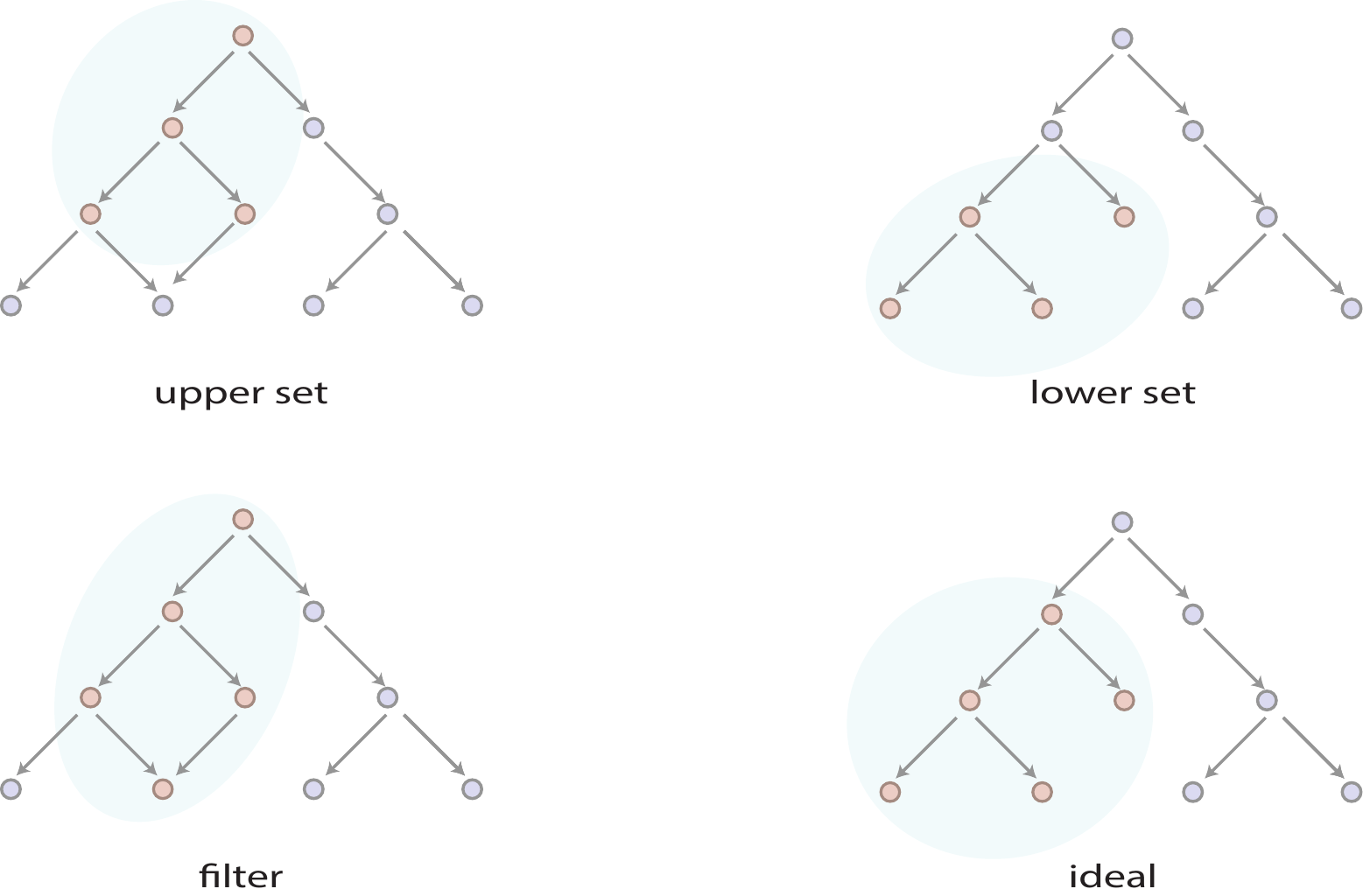}
\end{center}

\label{fig:ideal_filter}

\caption{
	{\bf Upper and lower sets, ideals and filters.}
	In each example, the subset $X$ of a poset is indicated by orange elements and blue background. 
	On the top left, $X$ is an upper set, given that $X = \uparrow\!X$. Note however that $X$ is not filtered, as not every subset of $X$ has a lower-bound in $X$. To obtain a filter (bottom right), we add an element to $X$. Now $X$ is a principal filter, because it has a maximum element (in this case, the top of the poset).
	Similarly, on the top right we have a lower set ($X=\downarrow\! X$) which is not directed, because not all of is subsets have an upper bound in $X$. This is fixed in the bottom right diagram, with the addition of an element. Note that the resulting ideal is not principal, because it lacks a minimum element.
}

\end{figure}

\begin{definition}[Labeling subsets of a poset]
Let $S$ be a poset and  $X \subseteq S$  a non-empty subset. We define:
\begin{itemize}
\item $\uparrow\! X = \bigcup_{x \in X} \{y\in S: y \geq x \}$;
\item $\downarrow\! X = \bigcup_{x \in X} \{y\in S: y \leq x \}$.
\end{itemize}
We say that $X$ is 
\begin{itemize}
\item an \emph{upper set} if $X = \uparrow\! X$;
\item a \emph{lower set} if $X = \downarrow\! X$;

\item \emph{directed} if every finite subset of $X$ has an upper bound in $X$;
\item \emph{filtered} if every finite subset of $X$ has a lower bound in $X$;

\item an \emph{ideal} if $X$ is a directed lower set;
\item a \emph{filter} if $X$ is a filtered upper set;

\item a \emph{principal} ideal if $X$ is an ideal and has a maximum element;
\item a \emph{principal} filter if $X$ is a filter and has a minimum element.
\end{itemize}

\end{definition}

\begin{remark}
Let $S$ be a poset and  $X \subseteq S$  a non-empty subset.

The set of upper bounds of $X$ is given by $\bigcap_{x\in X} \uparrow\! \{x\}$. If the join of $X$ exists, then this equals $\uparrow(\vee X)$.

Analogously, the set of lower bounds of $X$ is $\bigcap_{x\in X} \downarrow\! \{x\}$. If the meet of $X$ exists, this equals $\downarrow(\wedge X)$.
\end{remark}

\begin{definition}[Lattices and semilattices]
A poset $(S, \leq)$ is
\begin{itemize}
  \item a \emph{join-semilattice}  if  all finite non-empty subsets $X\subseteq S$ have a join in $S$;
  \item a \emph{meet-semilattice}  if  all finite non-empty subsets $X\subseteq S$ have a meet in $S$;
  \item a \emph{lattice} if is it both a join- and a meet-semilattice;
  \item a \emph{bounded lattice} if it is a lattice with both a top and a bottom;
  \item a \emph{complete} semilattice or lattice if the above applies to all non-empty subsets of $S$, not only to finite ones.
\end{itemize}

\end{definition}

\begin{remark}
  Lattices and semilattices define algebraic structures: 
\begin{itemize}

 \item A join-semilattice $(S, \leq)$ defines a commutative band $(S, \vee)$ with the binary join operation, $x \vee y := \vee \{x,y\}$.  The two definitions are related as 
 $$x \leq y \iff x \vee y = y. $$
 
  \item A meet-semilattice $(S, \leq)$ defines a commutative band $(S, \wedge)$ with the binary meet operation, $x \wedge y := \wedge \{x,y\}$. 
The two definitions are related as 
 $$x \leq y \iff x \wedge y = x. $$
  
  \item A lattice $(S, \leq)$ defines a structure $(S, \vee, \wedge)$ with two associative, commutative, idempotent operations, linked by the absorption law, $x \wedge (x \vee y ) = x \vee (x \wedge y) =x$. We call such structures\dots lattices. 
  
  \item A bounded lattice $(S, \leq)$ with top $1$ and bottom $0$ forms an algebraic lattice $(S, \vee, \wedge, 1, 0)$ where $1$ is the identity element for $\wedge$ (and absorbing for $\vee$), and $0$ is the identity element for $\vee$ (and absorbing for $\wedge$).
\end{itemize}
\end{remark}

\begin{example}
 The power set $2^\Omega$ of any set $\Omega$ is the set of all subsets of $\Omega$. If we order it by inclusion, $(2^\Omega, \subseteq)$, it forms a complete lattice, $(2^\Omega, \cup, \cap)$.  The join is given by set union, and the meet by set intersection. The top is $\Omega$ (which acts like an identity towards $\cap$ and like the absorbing element towards $\cup$) and the bottom is the empty set (vice-versa).
\end{example}

\subsection{Functions in order theory}

\begin{definition}[Functions between posets]
Let $(S,\leq)$, $(T,\leq)$ be two posets. A function $f: S \to T$ is said to be
\begin{itemize}
\item \emph{order-preserving} if $x \leq y \implies f(x) \leq f(y)$;  
\item \emph{order-reflecting} if $f(x) \leq f(y)\implies  x \leq y$;
\item \emph{order-embedding} if it is both order-preserving and order-reflecting, $x \leq y \Leftrightarrow f(x) \leq f(y)$;
\item an \emph{order isomorphism} if it is a bijective order-embedding function, in which case we say that the two posets are \emph{isomorphic}, and finally
\item if $S=T$, we say that $f$ is \emph{inflating} if $x \leq f(x)$ and \emph{deflating} if $f(x) \leq x$, for all $x$.
\end{itemize}
\end{definition}

\begin{remark} \label{lemma:injective_embeddings}
Order-reflecting functions, and in particular embeddings, are injective, given that $f(x)= f(y)  \implies x \leq y \wedge y \leq x \Leftrightarrow x=y$.
\end{remark}

\begin{definition}[Semilattice homomorphisms and quasi-homomorphisms]
Let $(S, \vee)$, $(T, \vee)$ be two join-semilattices. A \emph{homomorphism} between the two is a function $f: S\to T$ that preserves binary joins,
\begin{align*}
      f(x \vee y) = f(x) \vee f(y),
\end{align*}
for all $x, y \in S$.
The same goes for meet-semilattices, with $\wedge$ replacing $\vee$.
For complete semilattices, we demand that 
\begin{align*}
      f( \vee X) = \bigvee_{x \in X} f(x),
\end{align*}
for any subset $X \subseteq S$.
An \emph{semilattice endomorphism} is an semilattice homomorphism from a semilattice to itself. 
An \emph{semilattice isomorphism} is a surjective semilattice homomorphism.

A join-semilattice \emph{quasi-homomorphism} is a function $f: S\to T$ such that 
$$f(V\vee W)\geq f(V)\vee f(W).$$ 
A join-semilattice \emph{quasi-endomorphism} is  a quasi-homomorphism from a join-semilattice $S$ to  itself.
\end{definition}

\begin{remark}
 Homomorphisms of join-semilattices are in particular quasi-homomorphisms. 
\end{remark}

\begin{lemma}
Quasi-homomorphisms of join-semilattices are order-preserving functions. 
\end{lemma}

\begin{proof}
Let $(S, \leq )$, $(T, \preceq)$ be two join-semilattices, and $f: S\to T$ be a quasi-homomorphism. Let $x \leq y \in S$.
Recall that according to the definition of algebraic join-semilattices, $x \leq y \iff x \vee y = y$.
We have
\begin{align*}
f(x) \preceq f(x)\vee f(y) \preceq f(x\vee y) = f(y).
\end{align*}
\end{proof}

\begin{corollary}
Let $(S, \leq )$, $(T, \preceq)$ be two join-semilattices, and $f: S\to T$ be a quasi-homomorphism.  Let $x, y \in S$. If $x\leq y$, then $f(x) \vee f(y) = f(x \vee y)$.
\end{corollary}

\begin{proof}
Since $f$ is order-preserving, 
$$x \leq y \implies f(x) \preceq f(y) \iff f(x) \vee f(y) = f(y) = f(x \vee y).$$
\end{proof}

\begin{definition}[Lumping]
Let ($S, \leq$) be a join-semilattice. A map  $\Lump: S \to S$ is called a \emph{lumping} if it  is an idempotent inflating endomorphism. If it is an idempotent inflating quasi-endomorphism, we call it a \emph{quasi-lumping}.
\end{definition}

\subsection{Galois connections}
\label{sec:Galois}

\begin{definition}[Galois connection]
Let $(S,\leq)$, $(T,\leq)$ be two posets. A (antitone) \emph{Galois connection} between them is a set of two  functions, $g: S \to T$ and $h: T \to S$ such that
\begin{align*}
  y\leq g(x)  \ \Leftrightarrow \ h(y) \leq x,
\end{align*}
for all $x \in S, y \in T$.
The functions $g$ an $h$  are   \emph{adjoints}.
The composition $h \circ g: S\to S$ is called the associated \emph{closure operator}, while  $g \circ h: T\to T$ is the associated \emph{kernel operator}.

If $h \circ g= \id_S$, we say that the connection is a \emph{Galois insertion} of $S$ in $T$. In that case, we denote by $\Lump$ the kernel operator,  $\Lump :=g \circ h: T\to T$.
\end{definition}

\begin{remark}
 If $(g, h)$ is a Galois connection between $(S, \leq)$ and $(T, \leq)$, then 
 \begin{itemize}
 \item $h , g, h \circ g$ and $g \circ h$ are order-preserving;
 \item $h \circ g$ and $g \circ h$ are idempotent;
 \item $h \circ g(x) \leq x$, and $g \circ h(y) \geq y$, for all $x \in S$ and $y \in T$.
 \end{itemize}
\end{remark}

\begin{lemma}
In a Galois insertion, $h$ is surjective.
\end{lemma}
\begin{proof}
In a Galois insertion of $T$ in $S$, $h(g(x))=x$ for all $x\in T$. Hence for any $x\in T$, there exists an element $y$ in the domain $S$ of $h$ such that $h(y)=x$, namely $y=g(x)$.
\end{proof}

\begin{remark}[Composition of Galois connections]
If $(g_1, h_1)$ is a Galois connection between $(R, \leq)$ and $(S, \leq)$  and 
 $(g_2, h_2)$ is a Galois connection between $(S, \leq)$ and $(T, \leq)$, then the composition 
$(g_2 \circ g_1, h_1 \circ h_2)$ is a  Galois connection between $(R, \leq)$ and $(T, \leq)$. If they were both Galois insertions, then the composed connection is also a Galois insertion.
\end{remark}

\subsection{Quotient spaces}
\label{sec:quotient_spaces}

\begin{definition}[Quotient spaces]
Let $S$ be a set. An \emph{equivalence relation} $\sim$ is a binary relation that is reflexive, transitive and symmetric. 
Equivalence classes induced by an equivalence relation are the sets $[x] := \{y \in S: x \sim y \}$. 
The \emph{quotient space} of $S$ under the equivalence relation $\sim$ is defined as the set of all equivalence classes,
$$S/\!\sim := \{ [x], x \in S\} .$$
\end{definition}

\begin{definition}[Induced equivalent relations]
Let $(S, \leq)$ and $(T, \leq)$ be two  join-semilattices and $f: S \to T$ a semilattice order-homomorphism. 
We define the binary \emph{relation induced} by $f$ as
$$ x \sim y : f(x) = f(y). $$
\end{definition}

\begin{lemma}
\label{lemma:properties_induced_relation}
Let $(S, \leq)$ be a complete join-semilattice and $\Lump$ a quasi-lumping in $S$.  The relation $\sim$ induced by $\Lump$ has the following properties:
\begin{enumerate}

\item $\sim$ is an equivalence relation.

\item $[\Lump(x)]=[x]$.

\item If $x \sim y$, then $\Lump(x) \vee y = \Lump(x)$; more generally
$\bigvee_{y\in[x]} y = \Lump(x)$.

\end{enumerate}

If, in addition, $\Lump$ is an endomorphism (and therefore a lumping), then the following properties also hold:

\begin{enumerate}
\item $\sim$ is \emph{cumulative}: $x \sim y \implies x \sim  x \vee y$.

\item If $x \sim y$ and $w \sim z$, then $x \vee w \sim y \vee z$.
\end{enumerate}

\end{lemma}

\begin{proof} 
For quasi-endomorphisms, we have:
\begin{enumerate}
\item $\sim$ is an equivalence relation (reflexive, symmetric and transitive), because $=$ is itself an equivalence relation.
\item Follows from idempotence of $\Lump$: $\Lump(\Lump(x)) = \Lump(x)$.

\item Follows from the facts that $\Lump(x)\in [x]$  and  $\Lump$ is inflating, $\Lump(x) = \Lump(y) \geq y$. 
\end{enumerate}

For endomorphisms, we have in addition:

\begin{enumerate}

\item $\sim$ is cumulative, because $\Lump$ is a homomorphism. Indeed, let $x \sim y \Leftrightarrow \Lump(x) = \Lump(y)$. Then $\Lump(x \vee y) = \Lump(x) \vee \Lump(y) = \Lump(x) \vee \Lump(x) = \Lump(x)$. In other words, $x \sim x \vee y$.

\item Since $\Lump$ is a homomorphism and $x\sim y$ as well as $w\sim z$, $\Lump(x\vee w)=\Lump(x)\vee \Lump(w)=\Lump(y)\vee \Lump(z)=\Lump(y\vee z)$.
\end{enumerate}

\end{proof}

As we will see, these induced equivalence classes are useful to create Galois insertions. The following lemma tell us that we can always build such classes starting from any set of homomorphisms.

\begin{lemma}[Lumping generated by homomorphisms] \label{lemma:lumping_generated_homomorphisms}

Let $(S, \leq)$ and $(T, \leq)$ be two join semilatices and $f: S \to T$ a semilattice homomorphism.  Then the map
\begin{align*}
\Lump_f: S &\to S \\
x &\mapsto \bigvee \{ y: f(x) = f(y) \}
\end{align*}
is a lumping. 
More generally, if  $A$ is a set of homomorphisms $\{ f_A: S \to T\}_f$, then
\begin{align*}
\Lump_A: S &\to S \\
x&\mapsto \bigvee_{f_A\in A} \Lump_{f_A} (x)
\end{align*}
is  a lumping. 
\end{lemma}
\begin{proof}

To see that $\Lump_A$ is inflating, note that $x \in \{ y: f_A(x) = f_A(y) \}$, so $x \leq \Lump_{f_A} \leq \Lump_{A}$. 
To see that $\Lump_{f_A}$ is an endomorphism, note that
\begin{align*}
\Lump_f(x \vee z) 
&= \bigvee \{ y: f(x \vee z) = f(y) \} \\
&= \bigvee \{ y: f(x) \vee f(z) = f(y) \} \\
&= \left( \bigvee \{ y: f(x) = f(y) \} \right) \vee \left( \bigvee \{ y: f(z) = f(y) \}  \right) \\
&= \Lump_f(x) \vee \Lump_f(z),
\end{align*}
and from this it follows that $\Lump_A$ is also an endomorphism, because
\begin{align*}
\Lump_A(x \vee z) &= \bigvee_{f_A\in A} \Lump_{f_A} (x \vee z) \\
&= \bigvee_{f_A\in A} \Lump_{f_A} (x) \vee \Lump_{f_A}  (z) \\
&= \left( \bigvee_{f_A\in A} \Lump_{f_A} (x)  \right) \vee \left( \bigvee_{f_A\in A} \Lump_{f_A} (z)  \right)  \\
&= \Lump_A(x) \vee \Lump_A (x).
\end{align*}
Finally, to show that $\Lump$ is idempotent, we use the fact that $\Lump_f$ is an endomorphism,
\begin{align*}
\Lump_f \circ \Lump_f(x)
&=   \Lump_f \left( \bigvee \{ y: f(x ) = f(y) \}\right) \\
&= \bigvee_{y: f(x ) = f(y) } \Lump_f(y)  \\
&= \bigvee_{y: f(x ) = f(y) } \bigvee \{ {z: f(y) = f(z) }  \} \\
&= \bigvee \{ {z: f(x) = f(z) }  \} 
= \Lump_f (x).
\end{align*}

\end{proof}

The following proposition tell us that every lumping induces a Galois insertion. We will use this idea to define intensive embeddings and probabilistic equivalence.  

\begin{proposition}[Galois insertion induced by a lumping]
\label{thm:induced_Galois}
Let $(S, \leq)$ be a complete join-semilattice, let $\Lump$ be a lumping in $S$, and $\sim$ the  equivalence relation induced by $\Lump$. 
Then there is a Galois insertion of $S/\!\sim$ in $S$ defined by the pair $(g, h)$, where
\begin{align*}
g: \quad S/\!\sim &\to S\\
[x] &\mapsto \bigvee_{y \in [x]} y,\\
h: \quad S &\to S/\!\sim \\
x &\mapsto [x].
\end{align*}
Also, $(S/\!\sim, \preceq)$ is a join-semilattice with the induced partial order
\begin{align*}
[y] \preceq [x] : \quad 
g([y]) \leq g([x]),
\end{align*} 
and the join operation
$$[x] \vee [y] := h ( g([x]) \vee g([y])).  $$
Furthermore, $\Lump = g \circ h$, both $g$ and $h$ are semilattice homomorphisms, and $g$ is an order embedding of $S/\!\sim$ in $S$. 
\end{proposition}

\begin{proof}
Before we begin, observe that
\begin{align*}
g([x]) = g([y]) 
 \Leftrightarrow 
\bigvee_{x' \in [x]} x'
= \bigvee_{y' \in [y]} y' 
\Leftrightarrow f(x) = f(y)
\Leftrightarrow 
[x] = [y].
\end{align*}
Also note that by definition, any element of $S/\!\sim$ can be written as $[x]$ for some $x \in S$.

First we show that $h \circ g$ is the identity in $S/\!\sim$ and that $g \circ h = \Lump$. 
\begin{enumerate}
\item To show that $g \circ h =\Lump$, we use the fact that $\Lump$ is inflating,
$g \circ h(x) = g([x]) = \bigvee_{y \in [x]} y = \Lump(x)$  (Lemma~\ref{lemma:properties_induced_relation}).
\item We now have $h \circ g([x])  = h(\Lump(x)) = [x]$.
\end{enumerate}

Now we show that $(S/\!\sim, \preceq)$ is a poset:
\begin{enumerate}
\item $\preceq$ is reflexive by definition: 
$g([x])\leq g([x]) \Leftrightarrow [x] \preceq [x]$.
\item Transitivity of $\preceq$ follows from transitivity of $\leq$,
\begin{align*}
g([x]) \leq g([y])  
\quad  \text{and} \quad 
g([y]) \leq g([z])  
&\implies 
g([x]) \leq g([z])   
\quad \Leftrightarrow \\ 
\Leftrightarrow \quad
[x] \preceq [y]
\quad  \text{and} \quad 
[y] \preceq [z] 
&\implies 
[x] \preceq [z] .
\end{align*}
\item  $\preceq$ is antisymmetric:
\begin{align*}
[x] \preceq [y]
\quad  &\text{and} \quad 
[y] \preceq [x]  
\quad \Leftrightarrow \\ \Leftrightarrow \quad
g([x]) \leq g([y])
\quad  &\text{and} \quad 
g([y]) \leq g([x]) 
\quad \Leftrightarrow \\ \Leftrightarrow \quad
g([x]) = g([y])
\quad & \Leftrightarrow \quad
[x] = [y].
\end{align*}
\end{enumerate}

Note that $g$ is an order embedding by definition of the partial order in the reduced space. 
Now we show that $\preceq$ and $\vee$ are related in the usual way. Before we do it, note that $[x] \vee [y]= h ( g([x]) \vee g([y]))$ is well-defined, because $S$ is a join-semilattice.
\begin{enumerate}
\item First we show that $[x] \preceq [y] \implies [x] \vee [y] = [y]$. 
We have 
\begin{align*}
[x]\preceq [y]
&\iff g([x])\leq g([y]) \\
\flag{S \text{ join-semilattice}}
&\iff g([x])\vee g([y]) = g([y]) \\
&\implies h(g([x])\vee g([y]))=h(g([y]))\\
\flag{h \circ g \text{ identity, definition of } \vee} 
&\iff [x] \vee [y] = [y].
\end{align*}

\item Now we show that $ [x] \vee [y] = [y] \implies [x] \preceq [y]$. 
We have 
\begin{align*}
[x]\vee [y]=[y]
&\iff 
h(g([x])\vee g([y]))=[y] \\
&\implies 
g\circ h(g([x])\vee g([y]))=g([y]) \\
&\iff 
\Lump (g([x])\vee g([y]))=g([y]) \\
\flag{\Lump \text{ inflating}}
&\implies 
g([x])\vee g([y]) \leq g([y]) \\
&\implies 
g([x]) \leq g([y]) \\
&\iff 
[x] \preceq [y].
\end{align*}
\end{enumerate}

Now we show that $g$ and $h$ are homomorphisms. 
\begin{enumerate}
\item First we show that $h$ is a homomorphism.
We have 
\begin{align*}
h(x)\vee h(y)
&=h(g(h(x))\vee g(h(y))) \\
\flag{\Lump = g \circ h} &= h(\Lump(x) \vee \Lump(y)) \\
\flag{\Lump \text{ homomorphism}} &= h(\Lump(x \vee y)) \\
&= h\circ g \circ h (x \vee y)\\
\flag{h \circ g \text{ identity}} &= h(x \vee y).
\end{align*}

\item Now we prove that $g$ is a homomorphism. We have
\begin{align*}
g([x])\vee g([y])
&= \Lump(x) \vee \Lump(y)\\
\flag{\Lump \text{ homomorphism}} &= \Lump(x \vee y)\\
&= g \circ h(x \vee y) \\
\flag{h \text{ homomorphism}}
&= g([x] \vee [y]).
\end{align*}

\end{enumerate}

To prove that $(g, h)$ is a Galois insertion, it remains to show that it is a Galois connection:
\begin{enumerate}

\item First we show that $[y]\preceq [x] \implies  y\leq g([x])$.
We have that $ [y]\preceq [x] 
\iff g([y]) \leq g([x])$, so we just need to show that $y \leq g([y])$. This comes directly from the definition of join in the semilattice $S$.

\item Now we show $ y\leq g([x]) \implies [y]\preceq [x]$. We have
\begin{align*}
y\leq g([x])  
&\iff y\vee g([x])=g([x]) \\
&\implies h(y\vee g([x]))=h(g([x]))\\
\flag{h \text{ homomorphism}}
&\iff
h(y)\vee h(g([x]))=h(g([x]))\\
\flag{h \circ g \text{ identity}}
&\iff 
[y]\vee [x]=[x] \\
&\iff [y]\preceq [x].
\end{align*}

\end{enumerate}

Note also that, since the two posets are related by a Galois insertion, if $[x] \vee [y]$ were not a least upper bound for $[x]$ and $[y]$, then $g([x]) \vee g([y])$ would not be a least upper bound to $g([x])$ and $g([y])$, leading to a contradiction. 
\end{proof}

We obtain an analogous result if instead of a homomorphism we have a quasi-homomorphism. The proof is identical, and we include it here for completeness.

\begin{proposition}[Galois connection induced by a quasi-lumping]
\label{thm:quasi_reduced_induced_Galois}
Let $(S, \leq)$ be a complete join-semilattice, let $\Lump$ be a quasi-lumping in $S$, which induces an equivalence relation $\sim$ that is cumulative. Then there is a Galois insertion of $S/\!\sim$ in $S$ defined by the pair $(g, h)$, where
\begin{align*}
g: \quad S/\!\sim &\to S\\
[x] &\mapsto \bigvee_{y \in [x]} y,\\
h: \quad S &\to S/\!\sim \\
x &\mapsto [x].
\end{align*}
Also, $(S/\!\sim, \preceq)$ is a join-semilattice with the induced partial order
\begin{align*}
[y] \preceq [x] : \quad 
g([y]) \leq g([x]),
\end{align*} 
and the join operation
$$[x] \vee [y] := h ( g([x]) \vee g([y])).  $$
Furthermore, $\Lump = g \circ h$, $g$ is an order-embedding semilattice quasi-homomorphism, and $h$ is an (order-preserving) quasi-homomorphism.
\end{proposition}

\begin{proof}
Before we begin, observe that
\begin{align*}
g([x]) = g([y]) 
 \Leftrightarrow 
\bigvee_{x' \in [x]} x'
= \bigvee_{y' \in [y]} y' 
\Leftrightarrow \Lump(x) = \Lump(y)
\Leftrightarrow 
[x] = [y]
\end{align*}
because we assumed that $\sim$ is cumulative.
Also note that by definition, any element of $S/\!\sim$ can be written as $[x]$ for some $x \in S$.

First we show that $h \circ g$ is the identity in $S/\!\sim$ and that $g \circ h = \Lump$. 
\begin{enumerate}
\item To show that $g \circ h =\Lump$, we use the fact that $\Lump$ is inflating,
$g \circ h(x) = g([x]) = \bigvee_{y \in [x]} y = \Lump(x)$  (Lemma~\ref{lemma:properties_induced_relation}).
\item We now have $h \circ g([x])  = h(\Lump(x)) = [x]$.
\end{enumerate}

Now we show that $(S/\!\sim, \preceq)$ is a poset:
\begin{enumerate}
\item $\preceq$ is reflexive by definition: 
$g([x])\leq g([x]) \Leftrightarrow [x] \preceq [x]$.
\item Transitivity of $\preceq$ follows from transitivity of $\leq$,
\begin{align*}
g([x]) \leq g([y])  
\quad  \text{and} \quad 
g([y]) \leq g([z])  
&\implies 
g([x]) \leq g([z])   
\quad \Leftrightarrow \\ 
\Leftrightarrow \quad
[x] \preceq [y]
\quad  \text{and} \quad 
[y] \preceq [z] 
&\implies 
[x] \preceq [z] .
\end{align*}
\item  $\preceq$ is antisymmetric:
\begin{align*}
[x] \preceq [y]
\quad  &\text{and} \quad 
[y] \preceq [x]  
\quad \Leftrightarrow \\ \Leftrightarrow \quad
g([x]) \leq g([y])
\quad  &\text{and} \quad 
g([y]) \leq g([x]) 
\quad \Leftrightarrow \\ \Leftrightarrow \quad
g([x]) = g([y])
\quad & \Leftrightarrow \quad
[x] = [y].
\end{align*}
\end{enumerate}

Note that $g$ is an order embedding by definition of the partial order in the reduced space. 
Now we show that $\preceq$ and $\vee$ are related in the usual way. Before we do it, note that $[x] \vee [y]= h ( g([x]) \vee g([y]))$ is well-defined, because $S$ is a join-semilattice.
\begin{enumerate}
\item First we show that $[x] \preceq [y] \implies [x] \vee [y] = [y]$. 
We have 
\begin{align*}
[x]\preceq [y]
&\iff g([x])\leq g([y]) \\
\flag{S \text{ join-semilattice}}
&\iff g([x])\vee g([y]) = g([y]) \\
&\implies h(g([x])\vee g([y]))=h(g([y]))\\
\flag{h \circ g \text{ identity, definition of } \vee} 
&\iff [x] \vee [y] = [y].
\end{align*}

\item Now we show that $ [x] \vee [y] = [y] \implies [x] \preceq [y]$. 
We have 
\begin{align*}
[x]\vee [y]=[y]
&\iff 
h(g([x])\vee g([y]))=[y] \\
&\implies 
g\circ h(g([x])\vee g([y]))=g([y]) \\
&\iff 
\Lump (g([x])\vee g([y]))=g([y]) \\
\flag{\Lump \text{ inflating}}
&\implies 
g([x])\vee g([y]) \leq g([y]) \\
&\implies 
g([x]) \leq g([y]) \\
&\iff 
[x] \preceq [y].
\end{align*}
\end{enumerate}

Now we show that $h$ is an order-preserving quasi-homomorphism and that $g$ is a quasi-homomorphism (which as we noticed before is actually order-embedding). 
\begin{enumerate}
\item First we show that $h$ is order-preserving. To see that this has to be the case, note that $\Lump$ is order-preserving because it is a semilattice quasi-homomorphism. Since $g$ is an order-embedding, now also $h$ has to be order-preserving since $\Lump=g\circ h$.
\item Then we show that since $h$ is order-preserving, $h$ is a quasi-homomorphism. We have 
\begin{align*}
h(x)\vee h(y)
&=h(g(h(x))\vee g(h(y))) \\
\flag{\Lump = g \circ h} &= h(\Lump(x) \vee \Lump(y)) \\
\flag{\Lump \text{ quasi-homomorphism}, h \text{ order-preserving}} &\preceq h(\Lump(x \vee y)) \\
&= h\circ g \circ h (x \vee y)\\
\flag{h \circ g \text{ identity}} &= h(x \vee y).
\end{align*}

\item Now we prove that $g$ is a quasi-homomorphism. We have
\begin{align*}
g([x])\vee g([y])
&= \Lump(g([x])) \vee \Lump(g([y]))\\
\flag{\Lump \text{ quasi-homomorphism}} &\leq \Lump(g([x]) \vee g([y]))\\
&= g (h(g([x])\vee g([y]))) \\
&= g([x] \vee [y]).
\end{align*}

\end{enumerate}

To prove that $(g, h)$ is a Galois insertion, it remains to show that it is a Galois connection:
\begin{enumerate}

\item First we show that $[y]\preceq [x] \implies  y\leq g([x])$.
We have that $ [y]\preceq [x] 
\iff g([y]) \leq g([x])$, so we just need to show that $y \leq g([y])$. This comes directly from the definition of join in the semilattice $S$.

\item Now we show $ y\leq g([x]) \implies [y]\preceq [x]$. We have
\begin{align*}
y\leq g([x])  
&\iff y\vee g([x])=g([x]) \\
&\implies h(y\vee g([x]))=h(g([x]))\\
\flag{h \text{ quasi-homomorphism}}
&\implies
h(y)\vee h(g([x]))\leq h(g([x]))\\
\flag{h \circ g \text{ identity}}
&\iff 
[y]\vee [x]\preceq [x] \\
\flag{=\text{ follows as $\succeq$ trivial}}&\iff [y]\preceq [x].
\end{align*}

\end{enumerate}

Note also that, since the two posets are related by a Galois insertion, if $[x] \vee [y]$ were not a least upper bound for $[x]$ and $[y]$, then $g([x]) \vee g([y])$ would not be a least upper bound to $g([x])$ and $g([y])$, leading to a contradiction. 
\end{proof}

\subsection{Commutant and bicommutant}

The notions in this section  will be used to find  modularity of transformations in terms of commutativity relations in Appendix \ref{appendix:subsystems}. The semigroup considered will then be the monoid of transformations.

\begin{definition}[Commutant and bicommutant]
Let $(S, \cdot)$ be a semigroup. 
We say that $a$ and $b \in S$ \emph{commute} if $a \cdot b = b \cdot a$.
We say that two subsets $A$ and $B$ of $S$ \emph{commute} if every two elements $a \in A, b \in B$ commute.

The \emph{commutant} of a subset $A \subseteq S$ is the set of all elements that commute with $A$; we denote it $\com A$. The \emph{bicommutant}  or \emph{completion} of $A$, denoted $\bic A$, is the commutant of the commutant of $A$. 
If $\bic{A} = A$, we say that $A$ is \emph{complete}.
\end{definition}

\begin{lemma} \label{lemma:commutant_properties}
Let $(S, \cdot)$ be a semigroup. Then the following properties hold for all subsets $A, B \subseteq S$:
\begin{enumerate}
\item  $A \subseteq \bic{A}$;
\item if $A\subseteq B$, then $\com B \subseteq \com A$;
\item if $A=\com A$, then $A$ is complete;
\item $\com A$ is complete; as a corollary, $\bic A$ is complete; \label{lemma:complete_commutant}
\item $\com A \cap \com B = \com{A \cup B}$; 
\item  $ \com A \cup \com B \subseteq \com{A \cap B} $;
\item if $A$ and $B$ are complete, then $A \cap B$ is complete.
\end{enumerate}

\end{lemma}

\begin{proof}
\begin{enumerate}
\item  Every $a \in A$ commutes with $ \com A$, so $a \in \bic A$.
\item  Every element  $b \in \com B$ commutes with $B$, and in particular with $A$, therefore $b \in \com A$.
\item $\com A = A \implies \bic A = \com A = A$;
\item First we use  property 1, with  $A= \com C$. We obtain  $\com C \subseteq \bic{\com C}$. Then we use  property 2, taking  $A=C$ and $B= \bic C$. We obtain $\com{\bic C} \subseteq \com C$.
\item 
$\com A \cap \com B 
= \{y: ay =ya, \ \forall a \in A\} \cap \{y: by=yb, \ \forall b \in B\} 
=  \{y: ay =ya, \ \forall a \in A \cup B \}
=\com{A \cup B}$ ;
\item Using the second property,
$A, B \supseteq A \cap B  \implies \com A, \com B \subseteq \com{A \cap B} \implies \com A \cup \com B \subseteq \com{A \cap B} $.
\item Using property 5, we have $A \cap B = \bic A \cap \bic B = \com{\com A \cap \com B}$, which is complete by property 4.
\end{enumerate}

\end{proof}

\newpage
\section{Resource theories}
\label{appendix:resource_theories}

\subsection{General definitions and remarks}

\begin{definition}[Equality of maps]
Let $f$ and $g$ be two maps defined in a specification space $S^\Omega$.  We say that $f=g$ if $f(W) = g(W)$, for all $W \in S^\Omega$.  
\end{definition}

\begin{remark}
If both $f$ and $g$ are order homomorphisms,  it  suffices to demand 
$f(\omega) = g(\omega)$, for all $\omega \in \Omega$.
\end{remark}

\subsection{Proofs of claims from the main text}

Given a set of allowed transformations, we can analyse the structure that they induce in our specification space. We show that $\to$ is a pre-order: $V \to V$ and, if $V \to W$ and $W\to Z$, then $V \to Z$.

\begin{proposition}[Pre-order in resource theories]
Let $(S^\Omega,\cT) $ be a resource theory.
 The relation $V \to W$ is a pre-order in $S^\Omega$.
\end{proposition}
 
\begin{proof}
We have to show that $\to$ is reflexive and transitive. These properties hold because $\cT$ is a monoid. Firstly, $\cT$ has an identity element $1 \in \cT: \ 1(W) =W, \ \forall \, W\in S^\Omega$, and therefore $W \to W$.
Secondly, $\cT$ is closed under composition of transformations (for any $f, g \in \cT, f\circ g \in \cT$). We use this to prove transitivity, that is,  $V\to W $ and $W \to X$ implies $V\to X$. Let $f$ be the function that achieves $V\to W: f(V) \subseteq W$ and $g$ the function that achieves $W\to X: g(W) \subseteq X$. Then
\begin{align*}
 g\circ f(V) 
 &= \bigcup_{\omega \in f(V)}  g (\{\omega\}) 
 \subseteq \bigcup_{\omega \in W}  g (\{\omega\})
 = g(W) \subseteq X.
\end{align*} 
\end{proof}

\subsection{Additional Results}

\subsubsection{Properties of the pre-order}

\begin{lemma}[Knowing more cannot hurt]
Let $(S^\Omega,\cT )$ be a resource theory in a power set $S^\Omega$. Let $V, W \in S^\Omega$ be two compatible specifications, and let $Z\in S^\Omega$. If $V \to Z$ then $V \cap W \to Z$. 
\end{lemma}

\begin{proof}
Let $f$ be the transformation that achieves $f(V) \subseteq Z$. Then
\begin{align*}
  f(V \cap W) 
  &= \bigcup_{\omega \in V \cap W} f(\{\omega \}) 
  \subseteq \bigcup_{\omega \in V} f(\{\omega \}) 
  = f(V)
  \subseteq Z.
\end{align*}
\end{proof}

The following lemma simply tells us that in order to reach a specification $W$ from $V$ we need to apply the same transformation to all elements of $W$.

\begin{lemma}
Let $(S^\Omega,\cT )$ be a resource theory. Then, for any $V, W \in S^\Omega$, 
\begin{align*}
  V \to W 
  \quad \Leftrightarrow \quad
  \exists\, f \in \cT \
  \forall \, \omega \in V:
  \ f(\{\omega\}) \subseteq W.
\end{align*}
\end{lemma}

\begin{proof}
  We use the fact that all the endomorphisms in $\cT$ are element-wise functions (Lemma~\ref{lemma:homomorphisms}).
  To prove $\implies$ we take the transformation $f \in \cT$ that achieves $V\to W$ on the left-hand side. We have
  \begin{align*}
   W &\supseteq f(V) \\
     &= \bigcup_{\omega \in V} f(\{\omega \})\\
     &\supseteq f(\{\omega \}), \ \forall \, \omega \in V.
  \end{align*}
  
  To prove that $\impliedby$ holds,  we take the function $f$ that achieves the right-hand side. We have
  \begin{align*}
    f(V) 
    &= \bigcup_{\omega \in V} \underbrace{f(\{\omega \})}_{\subseteq W} \\
    &\subseteq \bigcup_{\omega \in V} W \\
    &=W ,
  \end{align*}
  which means that $V \to W$.
\end{proof}

\subsubsection{Partial order and quotient space}
Now we can define sets of equivalent resources: we say that two resources $V, W \in S^\Omega$ are equivalent if they are inter-convertible, that is, both $V \to W$ and $W \to V$. 

\begin{definition}[Quotient space]
Let $(S^\Omega, \cT)$ be a resource theory. We define the equivalence relation of mutual convertibility as 
\begin{align*}
V \sim_{\cT} W
\quad \Leftrightarrow \quad
V\to W \text{ and } W \to V,
\end{align*}
the corresponding equivalence classes
\begin{align*}
[W]_{\cT}:= \{V \in S^\Omega:\ V\sim_{\cT} W\},
\end{align*}
and the quotient space
\begin{align*}
S^\Omega/\cT:= S^\Omega/\!\sim_{\cT} = \{[W]_{\cT}: W \in S^\Omega \}.
\end{align*}
\end{definition}

In the quotient space, $\to$ is a partial order. 

\begin{remark} 
Let $(S^\Omega, \cT)$ be a resource theory.
The quotient space $S^\Omega/\cT$, together with the partial order induced by $\to$,
\begin{align*}
 [V]_{\cT}  \leq [W]_{\cT} 
  \ \Leftrightarrow  \ V\to W,
\end{align*}
is bounded; the top element is $[\Omega]_{\cT}$.  
\end{remark}

We may now define equivalent resource theories. Note that the following definition does not require the two state spaces $\Omega$ and $\Sigma$ to be equal, only that we can find the same structure under the respective allowed transformations.

\begin{definition}[Equivalent resource theories]
Two resource theories $(S^\Omega, \cT)$ and $(S^\Sigma, \cM)$ are  \emph{equivalent} if $S^\Omega / \cT$ is isomorphic to $S^\Sigma / \cM$. 
\end{definition}

\subsubsection{Free resources and conservation laws}

\begin{remark}
Let $(S^\Omega, \cT)$ be a resource theory. Then $V \in S^\Omega$ is a \emph{free resource} if and only iff $V \in [\Omega]_{\cT}$. 
\end{remark}

\begin{proof}
Note that $V \to \Omega$ is trivially true for any $V \in S^\Omega$, because $V\subseteq \Omega$.  This means that $\Omega \to V \Leftrightarrow V \in [\Omega]_{\cT}$.
\end{proof}

Some resource theories satisfy conservation laws: for instance, applying a thermal operation to a Gibbs state results in a Gibbs state, or applying local operations to bipartite states conserves the mutual information between the two systems. The following definition formalizes this intuition.

\begin{definition}[Conserved resources]
Let $(S^\Omega, \cT)$ be a resource theory. A resource $V \in S^\Omega$ is \emph{conserved} by the theory if $f(V) = V$, for all transformations $f\in \cT$.

Let $(S^\Omega, \cT)$ be a resource theory. We say that a  restricted theory $(S^\Omega, \cM)$ is obtained via \emph{conservation} of a set of specifications $\mathcal V \subseteq S^\Omega$ if $
  \cM= \{f \in \cT: f(V) = V, \ \forall \ V \in \mathcal V  \}$.

\end{definition}

\begin{definition}[Resource-independent transformations]
  Let  $(S^\Omega, \cT)$ be a resource theory. A transformation $f_V \in \cT$ is said \emph{resource-independent} if $f_V(W)= f_V(\Omega) = V$, for all specifications $W\in S^\Omega$.
\end{definition}

In the context of quantum resource theories,  resource-independent transformations correspond to replacing the global state with a fixed specification. For instance, in the case of thermal operations or Gibbs-preserving maps, this could be a global Gibbs state.

\begin{remark}
 Let $(S^\Omega, \cT)$ be a resource theory. If there is a resource-independent $f_V \in \cT$, then $V$ is a free resource.

If a resource theory preserves one specification $V$, there can be at most one resource-independent transformation, $f_V: W \to V, \forall \ W \in S^\Omega$. If the theory conserves more than one specification, then there can no resource-independent transformations.
\end{remark}

\newpage
\section{Specification embeddings}
\label{appendix:embeddings}
In this appendix we explore specification embeddings, and provide rigorous proofs to the results presented in the article.

\subsection{General definitions and remarks}

\begin{definition}[Maps preserved by a lumping or embedding]

Let $S^\Omega$ be a specification space, $\Lump$ an inflating, idempotent quasi-endomorphism in $S^\Omega$ and $F$ a set of maps $f: S^\Omega \to S^\Omega$. We say that $\Lump$ \emph{preserves} $F$ if  $\Lump \circ f = \Lump \circ f \circ \Lump$, for all $f\in F$.

If  $\Lump$ induces an intensive embedding $(\e, \h)$, we may also say that the embedding preserves $F$.

\end{definition}

\begin{remark}
For intensive embeddings, it suffices to demand $\h \circ f = \h \circ f \circ \Lump$.
\end{remark}

\subsection{Proofs of claims from the main text}

\propGeneralEmbeddings*

\begin{proof}
For the first equality,
$\e=\e_\text{ext}\circ\e_\text{int},$
first we define a new state space $$\Gamma:=\{\nu\in\Sigma':\nu\in i\circ\e(\{\omega\})\text{ for some }\omega\in\Omega\},$$
where $\Sigma'$ is a copy of $\Sigma$ and $i$ the isomorphism between $\Sigma$ and $\Sigma'$.

As $\Gamma\subseteq\Sigma'$, there is trivially an extensive embedding $\e_\text{ext}$ that connects $S^\Gamma$ to $S^\Sigma$. Now, define the intensive embedding $\e_\text{int}:\ S^\Omega\to S^\Gamma$ through the corresponding Galois adjoint $$\h_\text{int}(\{\gamma\})=\omega\iff\gamma\in i\circ\e(\{\omega\}).$$
Here, $\h_\text{int}$ is well-defined because $\e$ is an order-embedding, and gives a surjective homomorphism. The corresponding $\e_\text{int}$ is defined in the usual way: $$\e_\text{int}(V):=\{\gamma\in\Gamma: \h_\text{int}(\{\gamma\})\subseteq V\}.$$

Now, it is easy to verify that $\e=\e_\text{ext}\circ\e_\text{int}$.

To show that $\e$ can be also written as $\e=\widetilde\e_\text{int}\circ\widetilde\e_\text{ext}$, define a space $\Gamma$ through $$\Gamma:=\Omega'\cup\{\gamma\in \Sigma': \gamma\nsubseteq i\circ\e(\{\omega\})\text{ for any }\omega\in\Omega\}$$
with $\Omega'$ a copy of $\Omega$ and $\Sigma'$ a copy of $\Sigma$ with the isomorphism $i$, and take $\widetilde\e_\text{ext}$ to be the trivial extensive embedding between $S^\Omega$ and $S^\Gamma$. Now, define the intensive embedding $\widetilde\e_\text{int}: S^\Gamma\to S^\Sigma$ analogue to before in $\e_\text{int}$ through
$$\widetilde\h_\text{int}(\{\sigma\})=\omega\iff\gamma\in i\circ\e(\{\omega\}),$$
for any $\sigma\in\Sigma$ for which $i(\sigma)\notin\Gamma$. The elements $\gamma\in\Gamma$ that are already element of $\Sigma$ get mapped to themselves, that is
$$\widetilde\e_\text{int}(\{\gamma\})=i^{-1}(\{\gamma\})\iff\gamma\in\Sigma'.$$
Then, we find as required that
$\e=\widetilde\e_\text{int}\circ\widetilde\e_\text{ext}$.
\end{proof}

\propLumpingEmbedding*
\begin{proof}
The direct statement follows directly from Lemma \ref{lemma:induced_reduced_space} below. The converse follows by simply noting that $\e\circ\h$ is an idempotent and inflating endomorphism for any intensive embedding $\e$ with Galois adjoint $\h$.
\end{proof}

\begin{lemma}
\label{lemma:induced_reduced_space}
Let $S^\Omega$ be a specification space, let $\Lump$ be an idempotent, inflating specification endomorphism in $S$, and $\sim$ the  equivalence relation induced by $\Lump$.
Let $(g, h)$ be the induced Galois insertion of the quotient space $S^\Omega/\!\sim$ in $S^\Omega$ (see Proposition \ref{thm:induced_Galois}). 

Under these conditions, there exists an isomorphism $i$ between $S^\Omega/\!\sim$ and a specification space $S^\Sigma$, such that $(g \circ i^{-1}, i \circ h)$ is an intensive embedding of $S^\Sigma$ in $S^\Omega$. We call $S^\Sigma$ the \emph{reduced specification space induced} by $\Lump$.
\end{lemma}

\begin{proof}
First we identify the set of images of singletons,
$$\Sigma'= \{h(\{\omega\}), \  \omega \in \Omega \} = \{[\{\omega\}],  \omega \in \Omega \} \subseteq S^\Omega/\!\sim . $$
Our new state-space $\Sigma$ is simply a  copy of $\Sigma'$. We define this copy  via a constructive isomorphism,
\begin{align*}
\widetilde i: \quad 
\Sigma' &\to \Sigma \\
[\{\omega\}] &\mapsto \widetilde i([\{\omega\}] ) =: \sigma_{[\omega]}.
\end{align*}
We build $S^\Sigma$ from $\Sigma$ as usual. 
The second step is to posit that every element of $S^\Omega/\!\sim $ can be written as 
$$ \bigvee_{\omega \in W}  [\{\omega\}] ,$$
for some specification $W \in S^\Omega$. This holds because $(g,h)$ is a Galois insertion and $h$ is a semilattice homomorphism: the image of lower sets (like singletons) are lower sets.
Moreover, we know that for any specification $W \in S^\Omega$, the element  $\bigvee_{\omega \in W}  [\{\omega\}]$ exists, because the quotient space is a semilattice. 

Finally, we define the general isomorphism between $S^\Omega/\!\sim $ and $S^\Sigma$ as
\begin{align*}
i : \quad
S^\Omega/\!\sim &\to S^\Sigma \\
\bigvee_{\omega \in W}  [\{\omega\}]  &\mapsto \bigcup_{\omega \in W} \{\ \widetilde i(\ [\{\omega\}]\ )\  \} = \bigcup_{\omega \in W} \{ \sigma_{[\omega]} \}.
\end{align*}
This is by definition a semilattice homomorphism, which is simply mapping one order to the other. It is also bijective, for the reasons above-mentioned, so it is an isomorphism. 

To show that $(g \circ i^{-1}, i \circ h)$ is an intensive embedding of $S^\Sigma$ in $S^\Omega$, note that the pair of isomorphisms $(i^{-1},i)$ is in particular a  (trivial)  Galois insertion. The composition of two Galois insertions is a Galois insertion, and the composition of two homomorphisms is a homomorphism, so our pair is an intensive embedding.

\end{proof}

\propNested*

\begin{proof}
Let us prove the first statement. Thanks to Proposition \ref{thm:lumping} all we have to show is that $\Lump_{ABC \to A}:= \e_{A\to ABC} \circ \h_{ABC \to A}$ is an idempotent,
inflating specification endomorphism in $S^{ABC}$ (with $\h_{ABC \to A} = \h_{AB \to A} \circ \h_{ABC \to AB}$). This then ensures that $\e_{A\to ABC}$ forms an intensive embedding since $\Lump_{ABC \to A}$ would induce an embedding isomorphic to it.

To show that $\Lump_{ABC \to A} $ is inflating, we have
\begin{align*}
\Lump_{ABC \to A} (V) 
  &= \e_{AB \to ABC} \circ \underbrace{\e_{A \to AB} \circ  \h_{AB \to A}}_{\flag{\Lump_{AB \to A} \text{ inflating}}} \circ \, \h_{ABC \to AB} (V) 
  \supseteq  \underbrace{\e_{AB \to ABC}  \circ \h_{ABC \to AB} }_{\flag{\Lump_{ABC \to AB} \text{ inflating}}} (V)
  \supseteq V.
\end{align*}
For idempotence, we have
\begin{align*}
&\Lump_{ABC \to A} \circ \Lump_{ABC \to A} \\
  &\quad = \e_{AB \to ABC} \circ \e_{A \to AB} \circ  \h_{AB \to A} \circ \, 
    \underbrace{\h_{ABC \to AB}  \circ 
  \e_{AB \to ABC} }_{\flag{\text identity}}
  \circ\, \e_{A \to AB} \circ  \h_{AB \to A} \circ \,
  \h_{ABC \to AB} \\
  & \quad= \e_{AB \to ABC} \circ \e_{A \to AB} \circ  \underbrace{\h_{AB \to A}  \circ \e_{A \to AB} }_{\flag{\text identity}} \circ \, \h_{AB \to A} \circ \, \h_{ABC \to AB} 
  = \Lump_{ABC \to A}.
\end{align*}
Finally, $\Lump_{ABC \to A} $ is a composition of homomorphisms, so it is a specification endomorphism (Lemma~\ref{lemma:composing_homomorphisms}).

For the converse, we will prove the same properties for $\Lump_{AB \to A}:= \e_{A \to AB} \circ \, \h_{AB \to A}$  (with $\h_{AB \to A} = \h_{ABC \to A} \circ \, \e_{AB\to ABC}$).
To show that the lumping is inflating, we have 
\begin{align*}
\Lump_{AB \to A} (V) 
  &= \h_{ABC \to AB} \circ \underbrace{\e_{A \to ABC} \circ  \h_{ABC \to A}}_{\flag{\Lump_{ABC \to A} \text{ inflating}}} \circ \, \e_{AB \to ABC}\, (V) 
  \supseteq  \underbrace{\e_{AB \to ABC}  \circ \h_{ABC \to AB} }_{\flag{\text{identity}}}\, (V)
  = V.
\end{align*}
To show that it is idempotent, we write 
\begin{align*}
&\Lump_{AB \to A} \circ \Lump_{AB \to A} \\
  &\quad = \h_{ABC \to AB} \circ
  \Lump_{ABC \to A} \circ  
  \underbrace{\e_{AB \to ABC}\, \circ \h_{ABC \to AB}}_{\flag{\text{inflating }\Lump_{ABC \to AB} \subseteq \Lump_{ABC \to A}}}
  \circ \Lump_{ABC \to A} \circ \, \e_{AB \to ABC}  \\
  &\quad = \h_{ABC \to AB} \circ
    \Lump_{ABC \to A} \  
  \circ \underbrace{\Lump_{ABC \to A}}_{\flag{\text{idempotent}}} \circ \, \e_{AB \to ABC}  \\
  &\quad =\h_{ABC \to AB} \circ     \Lump_{ABC \to A} \circ \, \e_{AB \to ABC}
  =\Lump_{AB \to A}.
\end{align*}
 Finally, this lumping is again trivially a specification endomorphism.

To see that the definitions in the two parts coincide, note that taking $\e'_{A \to AB}:=\h_{ABC \to AB}\circ\e_{A \to ABC}$ yields
\begin{align*}
\e'_{A \to ABC} :=& \e_{AB \to ABC} \circ e'_{A \to AB} \\
=& \underbrace{\e_{AB \to ABC} \circ \h_{ABC \to AB}}_{\flag{\Lump_{ABC \to AB}}}\circ\e_{A \to ABC}  \circ \underbrace{\h_{ABC \to A} \circ \e_{A \to ABC}}_{\flag{\text{identity}}}  \\
=& \underbrace{\Lump_{ABC \to AB}}_{\flag{\subseteq \ \Lump_{ABC \to A}, \text{ inflating}}} \circ \Lump_{ABC \to A} \circ\ \e_{A \to ABC} \\
=& \Lump_{ABC \to A} \circ \ \e_{A \to ABC} 
= \e_{A \to ABC} .
\end{align*}

\end{proof}

\subsection{Additional results}

\subsubsection{Homomorphisms}

In the following we use the term `homomorphism' to refer to specification homomorphisms, that is order homomorphisms in specification spaces.
For simplicity, we may denote $f(\{\omega\})$ by $f(\omega)$.

\begin{remark}
 \label{rem:build_homomorphisms}
 Any function between two state-spaces, $\com{f}: \Omega \to \Sigma$ can be used to build a homomorphism,
  \begin{align*}
    f: \quad S^\Omega &\to S^\Sigma \\
    W &\mapsto \bigcup_{\omega \in W} \{\com f(\omega)\} .
  \end{align*} 
\end{remark}

\begin{lemma} \label{lemma:composing_homomorphisms}
 The composition of any two  element-wise functions in $S^\Omega$ is also an element-wise function.
\end{lemma}
\begin{proof}
 Take two element-wise functions, $f(W) =  \bigcup_{\omega \in W} \widetilde f(\omega) $ and $g(W) =  \bigcup_{\omega \in W} \widetilde g(\omega)$. We have
  \begin{align*}
  g \circ f (W) 
   &= g \left( \bigcup_{\omega \in W} \widetilde f(\omega) \right) \\
   &= g \left( \bigcup_{\omega \in W} \bigcup_{\omega' \in \widetilde f(\omega)} \ \{\omega'\} \right) \\
   &= \bigcup_{\omega \in W} \bigcup_{\omega' \in \widetilde f(\omega)} \ \widetilde g (\omega')\\
   &=: \bigcup_{\omega \in W}  \widetilde{g\circ f}(\omega) ,
  \end{align*}
  where we defined the function
   \begin{align*}
       \widetilde{g\circ f}: \quad \Omega 
       &\to S^\Omega \\
     \omega &\mapsto  
     \bigcup_{\omega' \in \widetilde f(\omega)} \ \widetilde g (\omega') .
   \end{align*}
\end{proof}

The following lemma concerns knowledge combination and homomorphisms: it says that, given two specifications, we always get a more precise description by combining them first and applying an homomorphism second than the other way around. 

\begin{lemma}\label{lemma:intersection_homomorphisms}
Let $S^\Omega$ and $S^\Sigma$ be two specification spaces,  and let $\mathcal W \subseteq S^\Omega$ be any subset of compatible specifications (that is, $\cap \mathcal W \neq \emptyset$). 
Then, for any homomorphism $f: S^\Omega \to S^\Sigma$,  
$$f(\cap \mathcal W )\subseteq \bigcap_{W \in \mathcal W } f(W) \in S^\Sigma.$$

\end{lemma}
\begin{proof}
Using the fact that  $f$ is a homomorphism, and therefore we can write $f(W) = \bigcup_{\omega \in W} \widetilde f(\omega)$, we have
  \begin{align*}
    \bigcap_{W \in \mathcal W }  f(W)
    &= \bigcap_{W \in \mathcal W } 
    \left(\bigcup_{\omega \in W } f(\omega)  \right) \\
    &= \bigcap_{W \in \mathcal W } 
      \left(     
        \left( 
          \bigcup_{\omega \in \, \cap \mathcal W } f (\omega)  
       \right)
       \cup
        \left( 
          \bigcup_{\omega \in\,  W \backslash (\cap \mathcal W) } f (\omega)  
        \right) 
      \right) \\
    &\supseteq  
       \bigcap_{W \in \mathcal W}    
            \left( 
              \bigcup_{\omega \in \, \cap \mathcal W } f (\omega)  
           \right) \\
    &=\bigcup_{\omega \in \, \cap \mathcal W } f (\omega)  \\
    &= f(\cap \mathcal W).
  \end{align*}
\end{proof}

\begin{lemma} \label{lemma:quasi_endomorphism_ss}
Let $S^\Omega$ be a specification space and $\Lump$ an idempotent, inflating quasi-endomorphism in $S^\Omega$. Let $V, W \in S^\Omega$. Then
\begin{enumerate}
\item  $\Lump (\Lump(V) \cup \Lump(W)) = \Lump(V \cup W)$,
\item if $V$ and $W$ are compatible, $\Lump(V) \cap \Lump(W) \supseteq \Lump(V \cap W)$,
\item if $V$ and $W$ are compatible, $\Lump (\Lump(V) \cap \Lump(W)) = \Lump(V) \cap \Lump(W)$.
\end{enumerate}
\end{lemma}

\begin{proof}

\begin{enumerate}
\item First observe that, since $\Lump$ is a quasi-endomorphism, order-preserving and idempotent,
\begin{align*}
\Lump(V)\, \cup\, \Lump(W) &\subseteq \Lump(V \cup W)  \\
\implies  \Lump (\Lump(V) \cup \Lump(W)) &= \Lump \circ \Lump(V \cup W)  \\
&=\Lump(V \cup W) . 
\end{align*}

To show the other direction, we use the fact that $\Lump$ is inflating and order preserving,
\begin{align*}
V \cup W &\subseteq \Lump(V) \cup \Lump(W) \\
\implies \Lump(V \cup W)  &\subseteq \Lump( \Lump(V) \cup \Lump(W) ).
\end{align*}

\item  We can write $V=(V\cap W)\cup (V\backslash W)$ and $W=(V\cap W)\cup (W\backslash V)$. Now we use the fact that $\Lump$ is a quasi-endomorphism to show
\begin{align*}
\Lump(V)\cap \Lump(W) &= \Lump((V\cap W)\cup (V\backslash W)) \ \cap \ \Lump((V\cap W)\cup (W\backslash V))\\
&\supseteq (\Lump(V\cap W)\cup \Lump(V\backslash W))\ \cap \ (\Lump(V\cap W)\cup \Lump(W\backslash V))\\
&\supseteq \Lump(V\cap W).
\end{align*}

\item Similarly to 1., direction $\supseteq$ follows from the fact that $\Lump$ is inflating. The other direction follows from 2.\ together with the fact that $\Lump$ is idempotent.
\end{enumerate}

\end{proof}

\subsubsection{Specification embeddings}

Here stand some properties of specification embeddings. Order embeddings, Galois connections and related notions are defined in Appendix~\ref{appendix:algebra_order}. 
In the following, $S^\Omega$ and $S^\Sigma$ are two specification spaces, and there is a specification embedding $\e: S^\Omega \to S^\Sigma$.

\begin{lemma}
Embeddings are directed, that is $\e(S^\Omega)$ is a directed set. Furthermore, it has a maximum element, $\e(\Omega)$. 
\end{lemma}

\begin{proof}
Let $V, W \in S^\Omega$. We have to show that $\e(V) \cup \e(W) \in \e(S^\Omega)$. Since embeddings are homomorphisms, $\e(V) \cup \e(W)  = \e(V \cup W)$, and $V \cup W \in S^\Omega$. 
Similarly, because $S^\Omega$ has a top, $\Omega$, the maximum of $\e(S^\Omega)$ is $\e(\Omega)$.
\end{proof}

\begin{lemma}
Embeddings are filtered, up to the empty set. That is, let $V, W \in S^\Omega$. If $V \cap W \neq \emptyset$, then $\e(V) \cap \e(W) \in \e(S^{\Omega})$.
\end{lemma}

\begin{proof}
If $V \cap W \neq \emptyset$  that is  $V\cap W \in S^\Omega$, then $ \e(V) \cap \e(W)= \e(V \cap W) \in \e(S^\Omega)$, because $\e$ is an embedding. 
\end{proof}

\begin{lemma}
 The composition of two specification embeddings is a specification embedding. In particular, the composition of two extensive embeddings is an extensive embedding, and the composition of two intensive embeddings is an intensive embedding.
\end{lemma}

\begin{proof}
The general claim follows from the facts that the composition of two order embeddings is an order embedding, and the composition of two homomorphisms is an homomorphism. For intensive embeddings, we have in addition that the composition of two Galois insertions is a Galois insertion. For extensive embeddings, $|\e_2 \circ \e_1(\{\omega\})| = 1$, since $|\e_1(\{\omega\})| = 1$ and $\e_2$ is also an extensive embedding. 
\end{proof}

\begin{lemma}
  For extensive embeddings, $\{\emb(\{\omega\})\}_{\omega \in \Omega} \subseteq \Sigma$. In other words, we could relabel the elements of $\Sigma$ such that $\emb = \id$ and $\Omega \subseteq \Sigma$.
\end{lemma}

\begin{proof}
For any $\omega \in \Omega$, we have $|\e(\{\omega\}) | = 1 \implies \e(\{\omega\}) \in \Sigma$. Furthermore, $\e$ is an embedding, and therefore injective (see Remark~\ref{lemma:injective_embeddings}). This means that $\Omega$ gets mapped to an isomorphic subset of $\Sigma$.
\end{proof}

\begin{lemma}
  For extensive embeddings, $\emb(S^\Omega)$ is a principal ideal of $S^\Sigma$.
\end{lemma}

\begin{proof}
We already know that $\emb(S^\Omega)$ is directed and has a maximum element. It remains to show that it is a lower set. 
Let $V \in S^\Omega$ and $X \in S^\Sigma$. We have to show that if $X \subseteq \e(V)$, then there exists $W \in S^\Omega$ such that $X= \e(W)$. Because embeddings are homomorphisms, we have $\e(V)= \bigcup_{\omega \in V} \e(\omega)$. Now we use the fact that $|\e(\omega)|=1$, 
\begin{align*}
  X \subseteq \bigcup_{\omega \in V} \e(\omega) 
  \quad \text{and} \quad
   |\e(\omega)|=1, \forall \omega \in V
   \quad \implies \quad
   \exists W \subseteq V:
  X &\subseteq \bigcup_{\omega \in W} \e(\omega)
   = \e(W).
\end{align*}

\end{proof}

\begin{lemma}
For intensive embeddings, $|\h(\{\sigma\})|=1$, for all $\sigma \in \Sigma$. 
\end{lemma}

\begin{proof}
Suppose that $|\h(\{\sigma\})|>1$. This implies that we can write $h(\{\sigma\})=\{\widetilde a\} \cup \widetilde X$ for some $\widetilde a\in\Omega, \widetilde X \in S^\Omega, \widetilde a \notin \widetilde X$. For a Galois insertion, $\e(\h(\{\sigma\}))\supseteq \{\sigma\}$. But since $\e$ is a homomorphism, we know that $\e(\{\widetilde a\}\cup\widetilde X)=\e(\{\widetilde a\})\cup \e(\widetilde X)$, and so this implies that either $\e(\{\widetilde a\})\supseteq \{\sigma\}$ or $\e(\widetilde X)\supseteq \{\sigma\}$. But as $\h$ is also a homomorphism and in particular order-preserving, and $\h\circ\e$ is the identity on $S^\Omega$, this implies either $\{\widetilde a\}=\h(\e(\{\widetilde a\}))\supseteq \h(\{\sigma\})$ or $\widetilde X=\h(\e(\widetilde X))\supseteq \h(\{\sigma\})$. But this contradicts our assumption that $\h(\{\sigma\})=\{\widetilde a\}\cup \widetilde X$ with $\widetilde a\notin \widetilde X$.
\end{proof}

\begin{remark}
If a specification embedding $\e: S^\Omega \to S^\Sigma$ is both intensive and extensive, then it is trivial, that is $\Omega $ is isomorphic to $ \Sigma$. 
\end{remark}

\newpage
\section{Locality}
\label{appendix:subsystems}

\subsection{General definitions and remarks}

The following definition will help us create a structure of embeddings corresponding to local descriptions, based only on commutativity relations between transformations. That is, we go from a subsystem structure at the level of transformations to one at the level of specifications (Theorem \ref{thm:transformations_independent_agents}).

We start with any submonoid of transformations $A$, and lump over them. This will give us reduced descriptions that are invariant under those transformations. For example, If our theory is quantum mechanics on three qubits and $A$ consists of all quantum maps acting on the first qubit, then $\Lump_{\cancel A}$ will generate local descriptions of the other two qubits.

\begin{definition}[Embedding generated by a monoid]  \label{def:lumping_transformations}

Let $(S^\Omega,\cT)$ be a resource theory and $A\subseteq\cT$ a submonoid of transformations. 
We say that the monoid $A$ \emph{generates} the lumping $\Lump_{\cancel A}$, as in Lemma 
\ref{lemma:lumping_generated_homomorphisms},
\begin{align*}
\Lump_{\cancel A}: S^\Omega &\to S^\Omega \\
V&\mapsto \bigcup_{f_A\in A}  \{W \in S^\Omega: f_A(V) = f_A(W)  \}.
\end{align*}
 Similarly, we say that an intensive embedding $(\e_{\cancel A}, \h_{\cancel A})$ that arises from $\Lump_{\cancel A}$ (as in Prop.\ \ref{thm:induced_Galois}) is \emph{generated} by the monoid $A$.
\end{definition}

\subsection{Proofs of claims from the main text}

\subsubsection{Independence of local descriptions}

\propFreeComposition*

In order to prove this proposition,  we need the following lemma.

\begin{lemma}
\label{lemma:ignorance_composability}
Let $S^A$ be a specification space intensively embedded in $S^\Omega$, and $V\in S^\Omega$ be a specification. The following are equivalent:
\begin{enumerate}
\item $\Lump_A(V)=\Omega$,
\item For all specifications $\local{W_A}$ local in $A$, it holds that $V\cap \local{W_A}\neq \emptyset$ and $\Lump_A(V\cap \local{W_A})=\local{W_A}$.
\end{enumerate}
\end{lemma}
\begin{proof}
It is easy to show that $(2)\implies (1)$ by taking $\local{W_A}=\Omega$. For the other direction, we rewrite  $(1)$ as
\begin{align*}
\e_A \circ \h_A (V) &=\Omega 
\implies \underbrace{\h_A \circ\e_A}_{\flag{\text{identity}}} \circ\ \h_A (V)=\h_A (\Omega)  
\iff \h_A (V)=\Omega_A 
\iff \bigcup_{\nu \in V} \h_A (\{\nu \}) = \Omega_A \supseteq W_A.
\end{align*}
This implies that for any $\omega_A\in W_A$ there exists $\nu \in V$ such that $\h_A (\{ \nu\}) = \{\omega_A\}$. Now, $$\local{W_A}=\{\omega\in \Omega: \h_A(\{\omega\})\subseteq W_A\}
\ni \nu ,
$$
which shows that $V \cap \local{W_A}\neq \emptyset$, and also that
\begin{align*}
W_A \subseteq \h_A (V \cap \local{W_A}) \implies \e_A(W_A) \subseteq\e_A \circ  \h_A (V \cap \local{W_A})  \iff \local{W_A} \subseteq \Lump_A  (V \cap \local{W_A}) .
\end{align*}
However, trivially also $\Lump_A(\local{W_A}\cap V)\subseteq\local{W_A}$ and so $\Lump_A(\local{W_A}\cap V)=\local{W_A}$.
\end{proof}

We may now proceed.

\begin{proof}[Proof of Proposition \ref{prop:free_composition}]
The equivalence between $1$ and $3$ follows from Lemma \ref{lemma:ignorance_composability}.
Next we prove $2 \implies 3$. That the intersection is non-empty follows from $2$ together with the fact that the lumping is inflating:
$Z\subseteq \local{Z_A}\cap \local{Z_B} = \local{V_A} \cap \local{W_B}.$
The second property also holds since  $\Lump_{A}:=\e_A\circ\h_A$ is idempotent and inflating, so that $$\Lump_{A}(\local{V_A} \cap \local{W_B})=
\Lump_{A}(\Lump_{A}(Z)\cap \Lump_{B}(Z))
\supseteq \Lump_{A} (\Lump_{A}(Z) \cap Z)
=\Lump_{A}(Z) 
= \local{V_A},$$
and the other direction ($\subseteq$) holds trivially. Applying $\h_A$ on both sides yields the desired equality. 
To show $3 \implies 2$, set $Z:= \local{V_A} \cap \local{W_B} \neq \emptyset$. 
Now trivially
\begin{align*}
\h_{A} (Z) 
&= V_A,\\
\h_{B} (Z) 
&= W_B.
\end{align*}
\end{proof}

\propIndependentProcessing*

\begin{proof}
Since $\e(V)$ is compatible with $W$, we can write $W=(W\cap \e(V)) \cup Z$ for some $Z\in S^\Omega$ such that $\h(Z)\cap V=\emptyset$. This is because if it were the case that $\h(Z)\cap V\neq \emptyset$, there would be an element $z\in Z$ such that $\h(\{z\})\subseteq V$, and so $z\in\e(V)$ by definition of $\e(V)$ and we can take it out of $Z$. Now
\begin{align*}
\e(V)\cap f(W)&=\e(V)\cap f((W\cap \e(V))\cup Z)\\
&=(\e(V)\cap f(W\cap \e(V)))\cup (\e(V)\cap f(Z))\\
\flag{\h\circ f(Z) \subseteq \h(Z)}&=(\e(V)\cap f(W\cap \e(V)))\cup \emptyset\\
\flag{\h(W\cap \e(V))\subseteq V}&=f(W\cap \e(V))
\end{align*}
where we have again used that $\e(V)=\bigcup\{\omega\in \Omega: \h(\{\omega\})\subseteq V\}$.
\end{proof}

\thmIndependentAgents*

\begin{proof}
With Proposition \ref{prop:independent_processing}, using the independence of the embeddings $\e_A$ and $\e_B$ from the transformations in $B$ and $A$ respectively,

\begin{align*}
f_A\circ g_B(V)&\subseteq f_A\circ g_B(\e_A\circ\h_A(V)\cap\e_B\circ\h_B(V))\\
&=f_A(\e_A\circ\h_A(V)\cap g_B\circ\e_B\circ\h_B(V))\\
&\subseteq f_A\circ \e_A\circ\h_A(V) \cap f_A\circ g_B\circ\e_B\circ\h_B(V)\\
&\subseteq f_A\circ \e_A\circ\h_A(V) \cap \e_B\circ\h_B\circ f_A\circ g_B\circ\e_B\circ\h_B(V)\\
&\subseteq \e_A \circ  \widetilde f_A \circ  \h_A (V) \cap \e_B \circ  \widetilde g_B \circ  \h_B (V)
\end{align*}
\end{proof}

\subsubsection{Modularity of transformations}

\begin{proposition}[Complete subsets of monoids are submonoids]
\label{prop:complete_submonoid}
Let $(S, \cdot)$ be a monoid and let $A$ be a complete subset of $S$. Then $(A, \cdot)$ is a submonoid.
\end{proposition}

\begin{proof}
We need to show that $(A, \cdot)$ is associative, that $A$ contains the identity $1$, and that it is closed under the operation~$\cdot$.
Associativity is inherited from $(S, \cdot)$, and $1$ commutes with all elements (in particular with all those in $\com A$), so $1 \in \bic A = A$. 
Now let $a,\widetilde a \in A$. Then, for any $b \in \com A$ we have
\begin{align*}
b \cdot ( a \cdot \widetilde a  )
&=( b \cdot  a ) \cdot \widetilde a  \\
&= ( a \cdot  b ) \cdot  \widetilde a  \\
&= a \cdot ( b \cdot \widetilde a)  \\
&= a \cdot ( \widetilde a  \cdot b )  \\
&= (a \cdot  \widetilde a ) \cdot b   ,
\end{align*}
where we used the associativity of $S$, and the fact that $b$ commutes with both $a$ and $\widetilde a$. This implies that $a\cdot \widetilde a \in \bic A = A$.
\end{proof}

\thmTransformationsLattice*

\begin{proof}
This proof relies on properties of the commutant, expressed in Lemma~\ref{lemma:commutant_properties}.

Let $A, B, C \in \Sys(S)$. We have to show that $\Sys(S)$ is closed under the two operations, that they are associative, commutative, obey the identities, and are linked by the absorption law. 
\begin{enumerate}
\item $A\vee B \in  \Sys(S)$ and $A\wedge B \in   \Sys(S)$ :

$A \vee B$ is complete by definition, therefore forms a submonoid. The same applies to $A \wedge B$, by property 7 of Lemma~\ref{lemma:commutant_properties}. Note that this set is never empty because both submonoids contain the identity $1$. 

\item $A \vee (B \vee C) = (A \vee B) \vee C $:

We have
\begin{align*}
A \vee (B \vee C) 
&= \bic{A \cup \bic{B \cup C}}\\
 &= \com{\com A \cap \com{\bic{B \cup C}}}  \quad \flag 1  \\
 &= \com{\com A \cap \com{B \cup C}} \quad \flag 2 \\
&= \com{\com A \cap \left( \com B \cap \com C\right)} \quad \flag 1 \\
&= \com{\com A \cap \com B \cap \com C},
\end{align*}
where we used the properties of the commutant from Lemma~\ref{lemma:commutant_properties}:  \flag{1} $\com{F \cup G} = \com F \cap \com G$ and \flag{2} $\com F$ is complete; the last step comes from associativity of $\cap$. The same argument shows that $(A \vee B) \vee C =  \com{\com A \cap \com B \cap \com C}$.

\item $A \wedge (B \wedge C) = (A \wedge B) \wedge C$:

Follows directly from associativity of $\cap$. 

\item $A\vee B = B \vee A$ and  $A\wedge B = B \wedge A$ ;

Follow from commutativity of $\cup$ and $\cap$ respectively. 

\item $A\vee \com \cT = A$:

We have $A\vee \cT = \bic{A \cup \cT } = \bic A = A$, because every complete subsystem contains $\com \cT$ (Lemma \ref{lemma:subsystem_contains_centre}).

\item $A\wedge S = A$:

We have $A\wedge S = A\cap S =A $ .

\item $A \wedge (A \vee B ) = A  $:

We start from
$ A \wedge (A \vee B ) 
= A \cap \bic{A \cup B} $.
Then, on the one hand we have $ A \cap \bic{A \cup B} \subseteq A$, and on the other 
$A \cap \bic{A \cup B} \supseteq A \cap (A \cup B) = A$. 

\item $A\vee (A \wedge B) = A$:

We have 
$
A \vee (A \wedge B ) 
= \bic{A \cup (A\cap B)} 
= \bic A =A.
$

\end{enumerate}

\end{proof}

\subsubsection{Independent agents from independent transformations}

\thmTransformationsIndependence*

\begin{proof}
We take $\Lump_A := \Lump_{\cancel {\com A}} $ and $\Lump_B := \Lump_{\cancel {\com B}} $, as in Definition \ref{def:lumping_transformations}.  Then, by Proposition \ref{prop:generated_embeddings_independence} below,  $\Lump_A$ is independent of $B$ (because $B \subseteq \com A$), and vice-versa. 
\end{proof}

\begin{proposition}[Properties of embeddings generated by transformations]\label{prop:generated_embeddings_independence}
Let $(S^\Omega,\cT)$ be a resource theory and $A\subseteq\cT$ a submonoid of transformations.  Then $A$ generates an intensive embedding $(\e_{\cancel A}, \h_{\cancel A})$ that is independent of $A$, that is
$$f_A \circ \Lump_{\cancel A}(V)\subseteq\Lump_{\cancel A}(V),$$ 
for all $f_A \in A$ and $V\in S^\Omega$.
Furthermore, the embedding preserves $A$ and any functions that commute with $A$, that is 
$$ \Lump_{\cancel A}\circ f (V) =  \Lump_{\cancel A} \circ f \circ  \Lump_{\cancel A} (V),$$
for all $f \in A \cup \com A$ and $V\in S^\Omega$.
\end{proposition}

\begin{proof}
The first statement  follows trivially from the definition of the embedding through the equivalence class induced by $\Lump_{\cancel A}$ generated by $A$, which satisfies $f_A(\Lump_{\cancel A}(V))\subseteq\Lump_{\cancel A}(V)$ for all $f_A\in A$ and $V\in S^\Omega$. Thus the embedding is independent of $A$.

To prove the second statement, we have $\h_{\cancel A}\circ f_A(V)=\h{\cancel A}(V)$ for all $f_A\in A$ and $V\in S^\Omega$. Now we may plug in $\Lump_{\cancel A}$ on the right to obtain 
$\h_{\cancel A}\circ\Lump_{\cancel A}(V)=\h_{\cancel A}\circ f_A\circ\Lump_{\cancel A}(V)$, which means that the embedding preserves $A$.
To see that it also preserves functions that commute with $A$, let $f_B \in \com A$. We have
\begin{align*}
\Lump_{\cancel A} \circ f_B\circ \Lump_{\cancel A}(V)
&=\Lump_{\cancel A} \circ f_B \left( \bigcup_{f_A\in A}\{ W:f_A(W)=f_A(V)\} \right)\\
&=\Lump_{\cancel A} \left( \bigcup_{f_A\in A}\{ f_B(W):f_A(W)=f_A(V)\} \right)\\
&=\bigcup_{g_A\in A}\bigcup_{f_A\in A}\{W':g_A(W')=g_A\circ f_B(W),\ f_A(W)=f_A(V)\}\\
&\subseteq \bigcup_{g_A\in A}\bigcup_{f_A\in A}\{W':g_A(W')=g_A\circ f_B(W),\ f_B\circ f_A(W)=f_B\circ f_A(V)\}\\
\flag{f_B\circ f_A=f_A\circ f_B}
&=\bigcup_{g_A\in A}\bigcup_{f_A\in A}\{W':g_A(W')=g_A\circ f_B(W),\ f_A\circ f_B(W)=f_A\circ f_B(V)\}\\
\flag{g_A\circ f_A= \ell_A \in A}
&\subseteq \bigcup_{\ell_A\in A}\{W':\ell_A(W')=\ell_A\circ f_B(V)\}\\
&=\Lump_{\cancel A}\circ f_B(V),
\end{align*}
while the other direction holds because the lumping is inflating. 
\end{proof}

In fact, as the transformations on $A$ and $B$ commute, we can obtain a stronger version of the Theorem above. This is shown in Corollary~\ref{corollary:commuting_independent_agents} later.

In addition, the following proposition states that with an extra assumption on the state space and the transformations in two complete subsystems $A$ and $B$, independent transformations yield independent agents with freely composable resources.

\begin{proposition}[Freely composable resource theories from independent transformations]
Let $(S^{\Omega}, \cT)$ be a resource theory, and let $A,B\in\Sys(\cT)$ be two independent subsystems. If for all $V,W\in S^{\Omega}$ 
$$\exists X\in S^\Omega,f_{\com A}\in \com A,f_{\com B}\in \com B \text{ s.t. } f_{\com A}(V)=f_{\com A}(X), f_{\com B}(W)=f_{\com B}(X)$$
then there exist two reduced specification spaces $S^{\Omega_A}$ and $S^{\Omega_B}$ intensively embedded in $S^\Omega$ which are freely composable and such that the resource theories $(S^{\Omega_A},\tilde A)$ and $(S^{\Omega_B},\tilde B)$ represent independent agents. 
\end{proposition}

\begin{proof}
As for the proof of Theorem~\ref{thm:transformations_independent_agents}, we construct the embedding through 
$\Lump_A=\Lump_{\cancel{\com A}}$
and 
$\Lump_B=\Lump_{\cancel{\com B}}$. 
From the Theorem, it follows that the two resulting resource theories 
$(S^{\Omega_A},\tilde A)$ and $(S^{\Omega_B},\tilde B)$ 
are independent. It 
is left to show that
$\Lump_{A}\circ\Lump_{B}(V)=\Lump_B\circ\Lump_A(V)=\Omega$ 
for all $V\in S^\Omega$
(since by Proposition~\ref{prop:free_composition} compatibility of embeddings is equivalent to free composition of local resources).
We have
\begin{align*}
\Lump_A\circ\Lump_B(V)&=\Lump_{\cancel {\com A}}\circ\Lump_{\cancel {\com B}}(V)\\
&=\Lump_{\cancel{\com A}}\left(\bigcup_{X\in S^\Omega,f_{\com B}\in \com B}(X: f_{\com B}(X)=f_{\com B}(V))\right)\\
&=\bigcup_{Y\in S^\Omega}\bigcup_{X\in S^\Omega}\bigcup_{f_{\com A}\in \com A,f_{\com B}\in \com B}(Y: f_{\com A}(Y)=f_{\com A}(X),f_{\com B}(X)=f_{\com B}(V))\\
&=\Omega
\end{align*}
and vice versa for $\Lump_B\circ\Lump_A(V)$, and so the local descriptions are freely composable.
\end{proof}

\propPartialTrace*

\begin{proof}
In order to show that $(\e,\h)$ induces an embedding, we have to show that the pair forms a Galois insertion and that $\e,\h$ are specification homomorphisms. The latter is easy to see by definition of $\e$ and $\h$, which can be understood as element-wise functions. To see that they form a Galois insertion, we note that
$\h \circ \e$ is trivially the identity, and that by construction 
for all $V_J\in S^{\Omega_J}, W \in S^\Omega$
$$ W \subseteq \e(V_J) 
\iff \tr_{\com J} (\rho) \in V_J \ \forall \rho \in W
\iff \h(W) \subseteq V_J $$
as required.
\end{proof}

\subsection{Additional Results}

\begin{remark}
If two agents are independent, then any transformation in $A \cap B$ acts trivially in the local specification spaces $S^{\Omega_A}$ and $S^{\Omega_B}$. That is, those operations only change degrees of freedom that are out of reach for the two agents (like some common environment or microscopic parameters).
\end{remark}

\subsubsection{Independent agents with commuting transformations}

The following corollary corresponds to Theorem~\ref{thm:independent_agents}, with the additional feature that the transformations of the two agents commute. This is the case, for example, for the independent agents constructed in the proof of Theorem~\ref{thm:transformations_independent_agents}.

\begin{corollary}
\label{corollary:commuting_independent_agents}
Let $(S^{\Omega_A},\widetilde A)$ and $(S^{\Omega_B},\widetilde B)$ represent two independent agents within a global theory  $(S^\Omega,\cT)$, such that their transformations $A,B\in\cT$ commute.
Then the two agents can operate independently and the resulting knowledge can be recombined, that is, for any $f_A\in A$ and $f_B\in B$ and $V\in S^{\Omega}$,
$$f_A\circ g_B(V)\subseteq f_A\circ \Lump_A(V)\cap g_B\circ\Lump_B(V).$$
\end{corollary}
\begin{proof}
The proof follows from the independence of the embeddings $\e_A$ and $\e_B$ from the transformations in $B$ and $A$ respectively, as well as from the fact that transformations in $A$ and $B$ commute.

\begin{align*}
f_A\circ g_B(V)&\subseteq f_A\circ g_B(\e_A\circ\h_A(V)\cap \e_B\circ\h_B(V))\\
\flag{\text{Proposition \ref{prop:independent_processing}}}&=f_A(\e_A\circ\h_A(V)\cap g_B\circ\e_B\circ\h_B(V))\\
&\subseteq f_A\circ \e_A\circ\h_A(V) \cap f_A\circ g_B\circ\e_B\circ\h_B(V)\\
&= f_A\circ \e_A\circ\h_A(V) \cap g_B\circ f_A\circ\e_B\circ\h_B(V)\\
&\subseteq f_A\circ \e_A\circ\h_A(V) \cap g_B\circ\e_B\circ\h_B(V).
\end{align*}
\end{proof}

\subsubsection{Inherited subsystems}

Sometimes, a resource theory $(S^\Omega,\cT)$ arises from particular (perhaps very limiting) constraints on an agent within a world of more general physical processes: think for example of the set of thermal operations within quantum theory. In this case, the subsystems induced by $\cT$ might not reflect the actual subsystem structure that an agent would like to assign based on his understanding of locality. In such cases, it can be natural to define subsystems at the level of the mother theory|a resource theory that allows a broader set of allowed transformations. Such a theory would allow all transformations an agent regards as physically possible, while at the same time a particular resource theory in consideration could be recovered from the mother theory by restricting to a subset of the allowed operations.

\begin{definition}[Inherited subsystems]
Let $(S^\Omega, \cT)$ be a restricted resource theory within $(S^\Omega, \cM)$. We define the \emph{inherited subsystem structure} of the former as
$$\Sys(\cT \langle \cM) = \{A \cap \cT, \ A \in \Sys(\cM)  \}. $$
\end{definition}

\begin{remark}
Although the reduced subsystems are now not necessarily complete, Lemma~\ref{lemma:monoid_intersection} ensures that they are still monoids. Similarly, independent subsystems in the mother theory induce independent subsystems in the daughter theory.
\end{remark}

In the end, it depends on the particular implementation of this framework which resource theory to consider for an operational definition of subsystems.  We introduce an example in the next section.

\subsubsection{Abelian sets and centre}

Suppose that a restricted resource theory $(S^\Omega, \cT)$ contains a subset of transformations that commute with all  others: for example the preparation of a conserved resource (such as a global Gibbs state for thermal operations). Those transformations form the Abelian \emph{centre} of $\cT$, and will be part of every bicommutant, and therefore every subsystem of transformations, just like the identity operation. 
For this reason, the centre of $\cT$, if non-trivial, can be problematic: having such global operations in local subsystems might not have the operational meaning that we would look for in a subsystem structure. 
We can overcome this issue by inheriting the subsystem structure from a higher theory with a trivial centre, such as general TPCPMs instead of thermal operations.  
In this section we define the notion of centre and show that it is a subsystem itself, contained in every other subsystem.

\begin{definition}[Commutative subsets]
Let $S$ be a semigroup. We say that $A\subseteq S$ is \emph{abelian} or \emph{commutative} if $A \subseteq \com A$.

We say that two subsets $A, B\subseteq S$ \emph{commute} if $A \subseteq \com B$ and $B \subseteq \com A$.

\end{definition}

\begin{definition}[Centre]
Let $S$ be a semigroup and $A\subseteq S$. We define the \emph{centre} of $A$ as
$$Z(A):= A \cap \com A .$$

In case $S$ is a monoid with identity $1$, we say that $A$ is \emph{centreless} if $Z(A) =\{1\}$.
\end{definition}

\begin{lemma}
Let $S$ be a semigroup and $A\subseteq S$ a  subset.  Then $Z(A)$ is  abelian.
\end{lemma}

\begin{proof}
We have to show that $Z(A) \subseteq \com{Z(A)}$. From Lemma~\ref{lemma:commutant_properties} we have
\begin{align*}
\com{Z(A)} 
= \com{A \cap \com A} \
\supseteq \com A \cup \bic A\
\supseteq \com A \cap \bic A\
\supseteq \com A \cap  A\ = Z(A).
\end{align*}
\end{proof}

\begin{lemma}
Let $S$ be a semigroup and $A\subseteq S$ a complete subset.  Then $Z(A)$ is complete.
\end{lemma}
\begin{proof}
From Lemma~\ref{lemma:commutant_properties}, we have that 
\begin{align*}
A \cap \com A 
 &= \bic A \cap \com A \\
 &= \com{\com A \cup A},
\end{align*}
which is the commutant of a set, and therefore complete.
\end{proof}

\begin{lemma}
Let $S$ be a monoid, and let $A$ and $B$ be two commuting complete subsystems. Then $A \cap B \subseteq Z(A) \cap Z(B)$.
\end{lemma}

\begin{proof}
Because $A$ and $B$ commute, we have $B \subseteq \com A $, therefore $A \cap B \subseteq A \cap \com A = Z(A)$. Analogously, $A \subseteq \com B$, so $A\cap B \subseteq Z(B)$. Therefore $A \cap B \subseteq Z(A)\cap Z(B)$.
\end{proof}

\begin{corollary}
If $A$ and $B$ are two commuting complete subsystems, then $A\cap B$ is abelian.
\end{corollary}

\begin{corollary}\label{cor:centreless_independent}
If $A$ and $B$ are two commuting complete subsystems and either of them is centreless, then $A \cap B = \{1\}$.
\end{corollary}

\begin{corollary}
If $S$ is a monoid and $A\subseteq S$ a complete submonoid, then $Z(A)$ is a complete submonoid.
\end{corollary}

\begin{lemma} \label{lemma:subsystem_contains_centre}
Let $S$ be a monoid and $A\subseteq S$ a  subset.  Then $Z(S) \subseteq Z(A) \subseteq \bic A$.
\end{lemma}

\newpage
\section{Approximation structures}
\label{appendix:approximations}
\subsection{General definitions and remarks}

In the main text, we have defined approximation structures through the partially ordered set $(\E, \leq)$. Often, $(\E, \leq)$ is also a monoid (for instance, $ (\mathbb R^+, +)$) that induces a distance in $\Omega$. In that case, an approximation structure might  satisfy a triangle inequality. The following definition generalizes that idea to the case where there might be more than one monoid in $\E$.

\begin{definition}[Triangle inequality]
\label{def:triangle_inequality}
Let $S^\Omega$ be a specification space equipped with an approximation structure $A_{\E}$. 
If there are chains  $\{\E_i \subseteq \E\}_{i \in \I }$ which are also commutative monoids $\{(\E_i, +_i)\}_{i \in \I }$, then we say that the
approximation structure satisfies a \emph{triangle inequality} if 
\begin{align*}
     \forall \, W \in S^\Omega,\
     \forall i \in \I ,\
      \forall \eps,\eps' \in \E_i:\quad
    (W^{\eps'})^{\eps'}
    \subseteq  W^{\eps+_i \eps'}.
\end{align*}
\end{definition}

\begin{remark}\label{rem:triangle}
Let $S^\Omega$ be a specification space equipped with an approximation structure  $A_{\E}$ that satisfies the triangle inequality. Then, for every commutative monoid $\E_i\subseteq \E$, and every $\eps, \eps'\in \E_i$,
\begin{align*}
 \text{if } \quad V\subseteq W^\eps \quad \text{ and } \quad
  \widetilde V \subseteq V^{\eps'}, \quad \text{ then } \quad
  \widetilde V \subseteq  W^{\eps+ \eps'},
\end{align*}
for all $ W, V , \widetilde V \in S^\Omega$.
\end{remark}

An approximation structure induces a special way to organize knowledge in terms of $\eps$-balls around elements of the state-space (see the examples at the end of this appendix). This is captured by the following definition.

\begin{definition}[Approximation space]
Let $\Omega$ be a state space equipped with an approximation structure $A_\E(\Omega)$. The \emph{approximation space} induced by $A_\E$ on $\Omega$ is the set of all specifications that are neighbourhoods of states, $\appspace{\Omega}\subseteq S^\Omega$.
\end{definition}

\subsection{Proofs of claims from the main text}

\propRobustness*

\begin{proof}
 Since the theory is stable, $V^\eps \to  W ^\eps$. If $V^\eps \in [\Omega]_{\cT}$, then $\Omega \to V^\eps \to W^\eps$. But $W^\eps \notin [\Omega]_{\cT}$, so we reach a contradiction.
\end{proof}

\propReducedApproximations*

\begin{proof} 

We define the new family of functions $A_\E(\Omega)=\{\cdot^\epsilon\}_{\epsilon\in\E}$ on the reduced specification space through
$$ W^\epsilon 
;= \h((\e(W))^\epsilon)$$
for any $W\in S^\Omega$, where $\e$ and $\h$ are the functions defining the Galois insertion between the two spaces. By definition, 
$$\e(W^\epsilon)\supseteq (\e(W))^\eps.$$
We need to show the following things:
\begin{enumerate}
\item The new functions are inflating specification endomorphisms on $S^\Omega$.

This follows from  being a composition of homomorphisms. 

\item For all $W \in S^\Omega$ and $\eps, \eps' \in \E$,\quad  $\eps \leq \eps' \implies W^\eps \subseteq W^{\eps'}$.

Follows from the facts that $\e$ and $\h$ are order-preserving and the approximations and  $A^\eps(\Sigma)$  satisfy the desired property.

\item There exists a \emph{saturating element} $ \eps_{\max} \in \E$ such that  for all $W \in S^\Omega$, \quad  $W^{\eps_{\max}}= \Omega$.

Since $\h(\Sigma)=\Omega$, any saturating element $\epsilon_{\max}$ stays saturating.

\item If the original approximation structure is attainable for some $0\in\E$, so is the new approximation structure.

Since $V^0=V$ for any $V\in S^\Sigma$, we have that $W^0=\h((\e(W))^0)=\h\circ\e(W)=W$, for any $W\in S^\Omega$.

\end{enumerate}

\end{proof}

\begin{remark}
Note that in the reduced approximation structure, it can happen that $\cdot^\epsilon=\cdot^{\epsilon'}$ although $\cdot^\epsilon\neq\cdot^{\epsilon'}$ in the original approximation structure.
\end{remark}

\begin{remark}
If an approximation structure $(\E,\leq)$ satisfies a triangle inequality on $S^\Sigma$, this property does not necessarily carry over to a reduced approximation structure on $S^\Omega$ embedded in $S^\Sigma$ as above. When this is actually the case is studied in Proposition \ref{prop:triangle_reduced} below.
\end{remark}

\subsection{Additional Results}

\subsubsection{General properties of approximation structures}

\begin{corollary}\label{lemma:intersection_approximations}
Let $S^\Omega$ be a specification space  equipped with an approximation structure $A_\E(\Omega)$. 
It is more precise to combine two states of knowledge  $V$ and $W$ and then approximate the result than to approximate first and then combine the approximations, 
$$(V\cap W)^\eps \subseteq V^\eps \cap W^\eps,$$
for all $V, W \in S^\Omega$ and all $\eps \in \E$.
\end{corollary}

\begin{proof}
Follows directly from Lemma~\ref{lemma:intersection_homomorphisms}.
\end{proof}

We have seen that intensive embeddings induce reduced approximation structures. The following propositions concern further properties of such induced approximation structures that are carried over from the initial approximation structure.

\begin{proposition}[Robustness in restricted resource theories] \label{prop:robust_restricted}
Let  $(S^\Sigma ,\cT)$ be a resource theory equipped with an approximation structure  $A_\E(\Sigma)$. Let there be an intensive embedding  $(\e, \h)$ resulting in restricted theory $(S^\Omega,\widetilde \cT)$ and approximation structure $A_\E(\Omega)$. 

Then, if $V\in S^\Sigma$ is not $\epsilon$-robust according to $A_\E(\Sigma)$, neither is $\h(V)$ $\epsilon$-robust according to $A_\E(\Omega)$.
\end{proposition}

\begin{proof}
If $V$ is not $\epsilon$-robust, this means that $\Sigma\to V^\epsilon$, that is, there exists a function $f\in\cT$ such that $f(\Sigma)\subseteq V^\epsilon$.
Since $\h$ is order-preserving, we have
\begin{align*}
\h\circ f(\Sigma)&\subseteq\h(V^\epsilon) \\
\iff 
\h \circ f \circ \e(\Omega)  &\subseteq\h(V^\epsilon) \\
\iff  \widetilde f (\Omega)  &\subseteq\h(V^\epsilon)
\subseteq \h((\e\circ\h(V))^\epsilon) =(\h(V))^\epsilon.
\end{align*}
Hence $\Omega\to(\h(V))^\epsilon$, and so $\h(V)$ is also not $\epsilon$-robust.
\end{proof}

The next proposition shows that if an embedding preserves the approximation structure, then an induced approximation structure inherits any triangle inequality.

\begin{proposition}[Triangle inequality in reduced approximation structures]
\label{prop:triangle_reduced}
Let $S^\Omega$ be a specification space equipped with an approximation structure $A_\E(\Omega)$ that respects a triangle inequality, and $S^T$ be a specification space embedded in $S^\Omega$ with an intensive embedding $\e$ that preserves the endomorphisms of the approximation structure. Then the reduced approximation structure $A_\E(T)$ on $S^T$ induced by $\e$ and $A(\Omega)$ also respects the triangle inequality.
\end{proposition}
\begin{proof}
Since the embedding preserves the approximation endomorphisms, we find that for $\epsilon,\epsilon'\in\E_i$ for all chains $\E_i\subseteq\E$ that are commutative monoids,
\begin{align*}
(V^\epsilon)^{\epsilon'}&:=(\h((\e(V))^\epsilon)^{\epsilon'}\\
&=\h((\e\circ\h((\e(V))^\epsilon))^{\epsilon'})\\
\flag{\e\text{ preserves }\cdot^\epsilon}&=\h(((\e(V))^\epsilon)^{\epsilon'})\\
&\subseteq \h((\e(V))^{\epsilon+_i\epsilon'})\\
&=V^{\epsilon+_i\epsilon'}.
\end{align*}
\end{proof}

\subsubsection{Examples of approximation structures}

\begin{example}[Approximations in quantum theory]
Approximations arise very naturally in the setting of quantum theory.
We start by taking a distance measure in state-space, that is, a distance between density matrices. For this we pick the purified distance, a generalization of the trace distance that is invariant under purifications and extensions, and can only decrease under physical operations and projections \cite{Tomamichel2010}. The purified distance is originally defined for subnormalized states; here we present a simpler version for normalized states, given that we are working in $\Omega$. 

\begin{definition}[Purified distance]
The \emph{purified distance}~\cite{Tomamichel2010} between two normalized quantum states $\rho$ and $\sigma$ is given by 
\begin{align*}
 \label{eq:purified_distance}
 d(\rho, \sigma) := \sqrt{1 - F(\rho, \sigma)^2},
\end{align*}
where $F$ is the fidelity between quantum states, 
\begin{align*} 
 F(\rho, \sigma) := \| \sqrt \rho \, \sqrt \sigma \|_1,
\end{align*}
and $\|A \|_1$ is the trace norm. 
\end{definition}
The approximation structure for the specification space is based on the $\eps$-balls of quantum states, according to the purified distance: 
 The $\eps$-approximation of a specification $W \in S^{\Omega}$ is the set
\begin{align*}
 W^\eps := \bigcup_{\omega \in W} \  \omega^\eps,
\end{align*}
where $ \omega^\eps = \{\rho \in \Omega: d(\rho, \omega) \leq \eps \}$ is the $\eps$-ball of $\omega$, according to the purified distance.

Note that $ V^1 = \Omega$ and $V^0 =V$, for all specifications $V \in S^\Omega$; the approximation structure induced by the purified distance is attainable.
\end{example}

\begin{example}[Animal approximations]
Continuing with $\E$ from Example~\ref{ex:zoo_approximation},  the approximation space $\appspace{\Omega}$ is the subset of $S^\Omega$  that consists of all phyla, classes, orders, etc., down to individual species. 
\end{example}

\begin{example}[Real approximations] 
\label{ex:real_approx}
Let $\Omega =[0,1]$, and $\E=\{10^{-n}\}_{n \in \mathbb N^0}$, with the usual order relation. Then we define the approximation structure $A_{\E}$ with  $x^\eps= \{\omega \in [0,1]: |x-\omega|< \eps \}$.  The approximation space induced by this structure, $\appspace{\Omega}$, consists of the $n$-digit approximations of all real numbers in $[0,1]$, for all $n$. Note that this is not an attainable approximation space, as we never reach the real line (because for all $\eps >0$ and all $x \in \mathbb R$, $x^\eps$ has an infinite number of elements). 
\end{example}

\newpage
\section{Convexity}
\label{appendix:convexity}

In the following we explore the structure that arises from convex state spaces for specifications. This convexity can sometimes be interpreted as probabilistic mixtures, such as in the case of quantum mechanics. On the other hand, convexity can also have a different meaning, such as a mixture of liquids in chemistry. Here, we shall not separate these two cases, since they induce the same structure on state space and specification space. In this sense, we are also not discussing how exactly we understand the notion of probability (frequencies, beliefs, bets?), which is a huge foundational topic by itself. Instead, we simply observe that some state spaces are convex, and could be interpreted as the spaces of probabilistic mixtures of extremal points. In the case where convexity does represent probabilities, we introduce more formally the notion of \emph{probabilistic equivalence} and we show in which sense it can be applied.

\subsection{General definitions and  results}

\subsubsection{Convexity}

Since we  leave the interpretation of convex mixtures open,  we  do not assume a particular structure on the state space other than that it may allow for convex mixtures. 
This implies that we require an  approach to convexity that is more general than standard definitions based on vector spaces. To this end, we shall define convexity solely based on a convex combination operation satisfying a certain set of axioms. There have been several similar proposals to generalize the concept of convexity \cite{Gudder1973, Flood1981, Fritz2009}.\footnote{See \url{https://golem.ph.utexas.edu/category/2009/04/convex_spaces.html} for a discussion. A very different approach (with some terminology clashes) is
followed e.g.\ in Ref.~\cite{Vel1993}.}
The following definition is adapted from Ref.~\cite{Fritz2009}.

\begin{definition}[Convex structure \cite{Fritz2009}]
\label{def:convexity}

A \emph{convex structure} is a set $\Omega$ equipped with a map 
\begin{align*}
f: [0,1] \times \Omega \times \Omega &\to \Omega \\ 
(p, \nu, \omega) &\mapsto \mix p \nu \omega
\end{align*}
satisfying:
\begin{enumerate}
\item $\mix 1 \nu \omega  = \nu$ (extremicity),
\item $\mix p \nu \omega  = \mix {1-p} \omega \nu $ (parametric commutativity),
\item $\mix  p {\mix q \nu \omega  } \tau = \mix{p \, q} \nu {\mix {\frac{1-p}{1- p\,q}} \tau \omega  }   $, for $p\, q\neq 1$ (parametric associativity), 

\item $\mix p \nu \nu = \nu$ (idempotence).
\end{enumerate}
We call $\mix p \nu \omega$ the \emph{convex mixture} of the two elements, and may represent the convex structure as $(\Omega, f)$. 
A pair $(\Omega, f)$ that satisfies all the above properties except idempotence is called a \emph{quasi-convex structure}. 

A subset $V \subseteq \Omega$ of a quasi-convex structure $(\Omega, f)$ is called a (quasi-)convex subset if  $(V, f_p)$ is itself a (quasi-)convex structure. (In particular, $V$ should be 
 closed under $f$, that is $\nu, \omega \in V \implies \mix p \nu \omega \in V, \forall p \in \{0, 1\}$.)
\end{definition}

For parametric associativity, note that for $p\,q=1$ the second coefficient is irrelevant due to extremicity.
The usual convex vector spaces are a special case of convex structures. 

\begin{lemma}[Convex vector spaces.]
Let $\Omega$ be a subset of a real vector space such that  for any two elements $\omega, \nu \in \Omega$ and any $p \in [0,1]$, the element $p\ \nu + (1-p)\  \omega =: \mix p \nu \omega$  is also in   $  \Omega$. Then $(\Omega, f)$  is an idempotent convex structure. 
\end{lemma} 

\begin{proof}
We have to prove:
\begin{enumerate}

\item Extremicity, $\mix 1 \nu \omega  = \nu$.

Follows directly from the fact that $0 \ \omega = 0$ is the identity element of addition  in a real vector space, and $1$ is the identity element of multiplication,  so  $1 \ \nu + 0\  \omega = \nu$.

\item Commutativity, $\mix p \nu \omega  = \mix {1-p} \omega \nu $.

Follows from commutativity of addition, $p \ \nu + (1-p) \ \omega = (1-p) \ \omega + p \nu$. 

\item Associativity, $\mix  p {\mix q \nu \omega  } \tau = \mix{p \, q} \nu {\mix {\frac{1-p}{1- p\,q}} \tau \omega  }   $, for $p\, q\neq 1$.

We have
\begin{align*}
p (q \nu + (1-q) \omega ) + (1-p) \tau 
&= p q \ \nu + p\ (1-q) \ \omega +(1-p) \tau \\
&= p q \ \nu + \frac{1- pq}{1-pq} \  p\ (1-q) \ \omega + \frac{1- pq}{1-pq} \ (1-p)\ \tau \\
&= p q \ \nu +  (1-pq)  \left( \frac{1- p}{1-pq}  \ \tau  +  \frac{ p\ (1-q)}{1-pq}  \ \omega  \right)\\
&= p q \ \nu +  (1-pq)  \left( \frac{1- p}{1-pq}\  \tau  +  \left(1- \frac{1- p}{1-pq}  \right)\ \omega  \right),
\end{align*}
where we used distributivity, commutativity of addition, and the identity element of multiplication|we can see you rolling your eyes.

\item Idempotence, $\mix p \nu \nu = \nu$. 

Follows from distributivity of scalar multiplication with respect to vector addition, $p \ \nu + (1-p) \ \nu = (p+ 1- p ) \ \nu$.

\end{enumerate}
\end{proof}

We can iterate the convex mixture of two elements to create a probabilistic mixture over a finite subset of elements, according to any finite probability distribution.

\begin{definition}[Mixture over finite sets]
\label{def:mixture_finite}
Let $(\Omega, f)$ be a convex structure and $\vec V  = (\nu_1, \nu_2, \dots, \nu_n )$ a finite vector of elements of $\Omega$. Finally, let $\vec P = (p_1, \dots, p_n)$ be a vector whose elements form a probability distribution  (that is, all $p_i \geq 0$ and $\sum_i p_i =1$). We define the mixture over $\vec V$ given by $\vec P$ as
$$f_{\vec P}(\vec V):= \mix {p'_1} {\nu_1} {\mix {p'_2} {\nu_2} {\mix {p'_3} {\nu_3} {\dots    }   }  } \in \Omega, $$
where the innermost nested term is $ \mix {p'_{n-1}} {\nu_{n-1}} {\nu_n}$, and the coefficients are given by $p'_1 = p_1$ and 
$p'_k =  {p_k}/ ({\prod_{i < k} (1-p'_i)})$, for $1<k<n$.
\end{definition}

\begin{lemma}
If $\Omega$ is a convex subset of a vector space then the above definition gives us the usual convex mixture, $f_{\vec P}(\vec V) = p_1 \ \nu_1 + p_2\ \nu_2 + \dots + p_n \ \nu_n$.
\end{lemma}

\begin{proof}

We  prove this by induction. First suppose that $\vec V_2 = (\nu_1, \nu_2)$ and $ \vec P = (p_1, p_2)$.  Then $f_{\vec P}(\vec V) = \mix{p_1}{\nu_1}{\nu_2} = p_1 \ \nu_1 + (1-p_1) \ \nu_2= p_1 \ \nu_1 + p_2\ \nu_2 $. 

Now suppose that for  vectors   $\vec V' = (\nu_{k+1},  \dots, \nu_n)$ and $ \frac{\vec  P'}{1-p_k}= \frac1{1-p_k}(p_{k+1}, \dots, p_n)$ it holds that 
$f_{\vec P'}(\vec V') =  \sum_{i=k+1}^{n} \frac{p_i}{1-p_k} \ \nu_i$.
Consider the extension to  vectors  $\vec V = (\nu_k, \nu_{k+1}, \dots, \nu_n)$ and $ \vec P = (p_k, p_{k+1},  \dots, p_n)$. 
Then we have 
\begin{align*}
\sum_{i=k}^{n} p_i \ \nu_i 
  &= p_k \ \nu_k + \sum_{i=k+1}^{n} p_i \ \nu_i\\
  &= p_k \ \nu_k + (1-p_k) \sum_{i=k+1}^{n}  \frac{p_i}{1-p_k} \ \nu_i\\ 
  &= p_k \ \nu_k + (1-p_k) f_{\frac{\vec P'}{1-p_k} }(\vec V') \\
  &= \mix {p_k} {\nu_k} {f_{\frac{\vec P'}{1-p_k} }(\vec V') } \\
  &= \mix {p_k} {\nu_k} 
  {\mix 
      {\frac{p_{k+1}}{1-p_k}} 
      {\nu_{k+1}} 
      {\mix 
        {\frac{p_{k+2}}{(1-p_k)(1-p_{k+1})}}  
        {\nu_{k+2}} 
        {\mix 
            {\frac{p_{k+2}}{(1-p_k)(1-p_{k+1})(1-p_{k+2})}} 
            {\nu_3} 
            {\dots }   
        }  
      }  
    }  \\
  &=f_{\vec P}(\vec V).
\end{align*}

\end{proof}

Finally, we may define convexity-preservation of maps between two quasi-convex structures.

\begin{definition}[Convexity-preserving maps]
Let $(\Omega, f)$ and $(\Sigma, f')$ be two quasi-convex structures. We say that a map $g: \Omega \to \Sigma$ is \emph{convexity-preserving} if
$$g\circ \mix p \nu \omega = \mixg p {g(\nu)}{g(\omega)}.$$
\end{definition}

\subsubsection{Properties of convex mixtures}

Here we show that convex mixtures over finite sets behave well as expected. In other words, in the following proofs we  expand the grisly nested mixtures for the last time. There is nothing particularly surprising or insightful in this subsection, so unless you are interested in this rediscovery of the wheel, we will not take it personally if you skip it.

\begin{lemma}[Combining mixtures]
\label{lemma:combining_mixtures}
Let $(\Omega, f)$ be a convex structure, let $\nu, \omega, \tau \in \Omega$, and let $r,p,q \in [0,1]$. Then 
$$ \mix r {\mix p \nu \omega } {\mix q \nu \tau} = \mix {r \, p + (1-r) \, q } \nu  {\mix \alpha   \omega \tau} ,  $$
with $\alpha = \frac{r (1-p)}{1- (rp+ (1-r) q)}$  (whenever $rp+ (1-r) q \neq 1$, otherwise we can take  $\alpha=0$ or anything else, as it does not matter).

\end{lemma}

\begin{proof}
We use associativity and commutativity of convex mixtures,
\begin{align*}
\mix r {\mix p \nu \omega } {\mix q \nu \tau} 
&= \mix {rp} \nu { \mix {\frac{1-r}{1-rp}}  {\mix q \nu \tau}  \omega   } \\
&= \mix {rp} \nu { \mix {\frac{q(1-r)}{1-rp}}  \nu {\mix   {\frac {1 - \frac{(1-r)}{1-rp} }{1 - \frac{q(1-r)}{1-rp}} }  \omega \tau}    }  \\
&= \mix {rp} \nu { \mix {\frac{q(1-r)}{1-rp}}  \nu {\mix   {\frac {r(1-p) }{  1- rp -q+qr } } \omega \tau}    } \\
&= \mix {1- rp} { \mix {1 - \frac{q(1-r)}{1-rp}}   {\mix   {\frac {r(1-p) }{  1- rp -q+qr } } \omega \tau} \nu   } \nu \\
&= \mix {
     (1- rp) \left( 1 - \frac{q(1-r)}{1-rp} \right) }
      {  \mix   {\frac {r(1-p) }{  1- rp -q+qr } }
           \omega    \tau
      }  
      {\mix {\dots} \nu \nu   }  \\
&= \mix {1 - (r p + (1-r) q)} { \mix {\frac{r (1-p)}{1- (rp+ (1-r) q)} } \omega \tau  } \nu \\
&=\mix {r p + (1-r) q} \nu { \mix {\frac{r (1-p)}{1- (rp+ (1-r) q)} } \omega \tau  }.
\end{align*}
\end{proof}

\begin{remark}
For vector spaces,  Lemma \ref{lemma:combining_mixtures} simply expresses that
\begin{align*}
& \quad r (p \ \nu + (1-p) \ \omega) + (1-r) (q \ \nu + (1-q) \ \tau) 
= (r \ p + (1-r) \ q)\ \nu + r \ (1-p) \ \omega + (1-r) \ (1-q) \ \tau.
\end{align*}

\end{remark}

\begin{lemma}[Mixture is convexity-preserving] \label{prop:idempotent_convex_structures}
Let $(\Omega, f)$ be a convex structure, let $\nu, \omega, \tau \in \Omega$, and let $r,p \in [0,1]$.
Then
$$ \mix r {\mix p \nu \omega } {\mix p \nu  \tau} = \mix p \nu  {\mix r  \omega \tau } .  $$
In other words, the function
\begin{align*}
\mix p \nu \cdot : \ \Omega &\to \Omega \\
\omega &\mapsto \mix p \nu \omega,
\end{align*}
parameterized by any $\nu \in \Omega$ and $p \in [0,1]$,
is convexity-preserving.
\end{lemma}

\begin{remark} \label{remark:combining_quasi_convex_structures}
If $(\Omega, f)$ is a quasi-convex structure without idempotence, then we can only ensure
$$ \mix q {\mix p \nu \omega } {\mix p \nu \tau} 
= \mix p  {\mix \beta \nu \nu} {\mix q  \omega \tau},$$
where $\beta \in [0,1]$ depends on $p$ and $q$.
\end{remark}

\begin{proof}

Follows from Lemma~\ref{lemma:combining_mixtures} when $p=q$, 
\begin{align*}
\mix r {\mix p \nu \omega } {\mix p \nu \tau} 
&=\mix {r p + (1-r) p} \nu { \mix {\frac{r (1-p)}{1- (rp+ (1-r) p)} } \omega \tau  } \\
&=\mix p \nu { \mix r \omega \tau  } .
\end{align*}

\end{proof}

The following lemma really tells us that the vector $\vec V$ in Definition~\ref{def:mixture_finite} can be re-ordered together with $\vec P$ at will.

\begin{lemma}[Permuting mixtures]
\label{lemma:order_mixture}
Let $(\Omega, f)$ be a convex structure and $\vec V  = (\nu_1, \nu_2, \dots, \nu_n )$ a finite vector of elements of $\Omega$. Finally, let $\vec P = (p_1, \dots, p_n)$ be a vector whose elements form a probability distribution over $\vec V$ (that is, all $p_i \geq 0$ and $\sum_i p_i =1$). Then
$$f_{\vec P}(\vec V)=f_{\pi(\vec P)}(\pi(\vec V)),$$
where $\pi$ denotes any permutation of elements in $\vec V$ and $\vec P$.
\end{lemma}
\begin{proof}
We will show that we  can  always swap two neighbouring elements, that is 
$f_{\vec P}(\vec V)=f_{{\pi_i}(\vec P)}(\pi_{i}(\vec V)),$
for permutations $\pi_i$  that swap the $i$-th and the $i+1$-th element. By repeated such swaps, we can then realize any permutation $\pi$  and so $f_{\vec P}(\vec V)=f_{\pi(\vec P)}(\pi(\vec V))$ in general.
We have
\begin{align*}
f_{\vec P}(\vec V)&=\mix {p'_1} {\nu_1} {\mix {p'_2} {\nu_2} {\dots \mix{p'_i}{\nu_i} {\mix {p'_{i+1}} {\nu_{i+1}} {\dots    }   }  }} \\
\flag{\text{commutativity}}&=\mix {p'_1} {\nu_1} {\mix {p'_2} {\nu_2} {\dots \mix{1-p'_i} {\mix {p'_{i+1}} {\nu_{i+1}} {\dots    }}{\nu_i}   }  } \\
\flag{\text{associativity}}&=\mix {p'_1} {\nu_1} {\mix {p'_2} {\nu_2} {\dots \mix{(1-p'_i)p'_{i+1}}{\nu_{i+1}} {\mix {p'_i/(1-p'_{i+1}(1-p'_i))} {\nu_i} {\mix{p'_{i+2}}{\nu_{i+2}}{\dots    }}   }  }}, \\
\end{align*}
We need to show that the new coefficients correspond to those of $\pi_i(\vec P)$. Indeed we find 
\begin{align*}
(1-p'_i)p'_{i+1}
&= \frac{(1-p'_i)p_{i+1}}{\Pi_{j<i+1}(1-p'_j)}
=\frac{p_{i+1}}{(\Pi_{j<i}(1-p'_j)}
= [\pi_i(\vec P)]_i' 
\end{align*}
and
\begin{align*}
\frac{p'_i}{1-p'_{i+1}(1-p'_i)}
&= \frac{p_i}{\Pi_{j<i}(1-p'_j)(1-[\pi_i(\vec P)]_i')}
= [\pi_i(\vec P)]_{i+1}'.
\end{align*}

\end{proof}

Since order does not matter, we may sometimes  denote $f_P(V) = f_{\vec P } (\vec V)$ for a subset $V \subset \Omega$ and a probability distribution $P$ over elements in $V$.

\begin{lemma} \label{lemma:mixture_of_distributions}
Let $(\Omega, f)$ be a convex structure and $\vec V  = (\nu_1, \nu_2, \dots, \nu_n )$ a finite vector of elements of $\Omega$. Let $\vec P= (p_1, \dots, p_n)$ and $\vec Q = (q_1, \dots, q_n)$ be two probability distributions over $\vec V$, and let $r \in [0,1]$. Then
$$\mix r {f_{\vec P}(\vec V)}{ f_{\vec Q}(\vec V)}  = f_{\vec R}(\vec V),$$
where $\vec R = (r_1, r_2, \dots, r_n)$ is a probability distribution given by  $r_i = r \ p_i + (1-r) \ q_i$.
\end{lemma}

\begin{proof} 
We prove this by induction. For two-element vectors, the claim follows directly from  Lemma~\ref{lemma:combining_mixtures}.  

Now assume that the claim holds for 
\begin{align*}
\vec V' = (\nu_{k+1}, \dots, \nu_n), 
\qquad
\frac{\vec P'}{1-p_k}= \frac1{1-p_k}(p_{k+1}, \dots, p_n) , 
\qquad
\frac{\vec Q'}{1-q_k}= \frac1{1-q_k} (q_{k+1}, \dots, q_n).
\end{align*}
Then we have, for the extended   vectors $\vec V = (\nu_k,\nu_{k+1}, \dots, \nu_n)$, $\vec P= (p_k, p_{k+1}, \dots, p_n) $ and $\vec Q= (q_k, q_{k+1}, \dots, q_n) $,
\begin{align*}
\mix r {f_{\vec P}(\vec V)}{ f_{\vec Q}(\vec V)}
&= \mix r 
   {\mix {p_k} {\nu_k} {f_{\frac{\vec P'}{1-p_k}} (\vec V')}}
   {\mix {q_k} {\nu_k} {f_{\frac{\vec Q'}{1-q_k}} (\vec V')}} \\
\flag{\text{Lemma~\ref{lemma:combining_mixtures}}} \quad
&= \mix{r \, p_k + (1-r) \, q_k}
   {\nu_k }
   {\mix
     {\frac{r\ (1-p_k)}{1- (r\ p_k + (1-r) \ q_k)}}
     {f_{\frac{\vec P'}{1-p_k}} (\vec V')}
     {f_{\frac{\vec Q'}{1-q_k}} (\vec V')}
    }\\
\flag{r_k:= r \, p_k + (1-r) \, q_k}
&= \mix{r_k}
   {\nu_k }
   {\mix
     {\frac{r \ (1-p_k)}{1-r_k}}
     {f_{\frac{\vec P'}{1-p_k}} (\vec V')}
     {f_{\frac{\vec Q'}{1-q_k}} (\vec V')}
    }\\
\flag{\text{induction hypothesis}} \quad
&= \mix{r_k}
   {\nu_k }
   {f_{\vec R'}(\vec V')}, 
\qquad \flag{\vec R' = \frac{r \ (1-p_k)}{1-r_k} \frac{\vec P'}{1-p_k} + \left(1-\frac{r \ (1-p_k)}{1-r_k}\right) \frac{\vec Q'}{1-q_k}} \\
&= \mix{r_k}
   {\nu_k }
   {f_{\vec R'}(V)}, 
\qquad \flag{\vec R' = \frac{1}{1 - r_k} (r \  \vec P' + (1-r) \ \vec Q' )} \\
&= f_{\vec R}(\vec V),   
\qquad \flag{\vec R = r \  \vec P + (1-r) \ \vec Q } .
\end{align*}

\end{proof}

\begin{lemma}[Nested mixtures as probability distributions]
\label{lemma:convex_combination_reverse}
Let $(\Omega, f)$ be a convex structure and $\vec V  = (\nu_1, \nu_2, \dots, \nu_n )$ a finite vector of elements of $\Omega$. Then, for any  $p_1',p_2',...,p_{n-1}'\in[0,1]$, there exists a  probability distribution vector $\vec P=(p_1,p_2,\dots,p_n)$ such that
$$\mix {p'_1}{\nu_1}{\mix{p'_2}{\nu_2}{ \dots  \mix{p'_{n-1}}{\nu_{n-1}}{\nu_n} \dots }}=f_{\vec P}(\vec V).$$
The terms of the probability vector are given by $p_k=p'_k\ (\Pi_{i<k}(1-p_i'))$ for $1<k<n$, $p_1= p_1'$, 
all $p_i\geq 0$ and $\sum_i p_i=1$.
\end{lemma}

\begin{proof}
We prove the lemma by induction.
For $\vec V_2 = (\nu_1, \nu_2)$  we trivially find that $\mix {p_1'} {\nu_1} {\nu_2}=f_{\vec P}(\vec V)$ where $p_1=p_1'$ and $p_2=1-p_1'$ as required.
Now assume that the claim holds for $\vec V'=(\nu_2, \dots \nu_n)$, and coefficients $p_2',  \dots, p_{n-1}'$, that is  
$$\mix{p'_2}{\nu_2}{\dots \mix{p'_{n-1}}{\nu_{n-1}}{\nu_n}\dots } = f_{\vec P^*}(\vec V'),$$
with $\vec P^* = (p_2, \dots, p_n)$, where  $p_k = p'_k\ (\Pi_{i<k}(1-p_i')) $ for $2<k<n$ and $p_2= p_2'$.
We will show that it also holds for the extensions $V = (\nu_1, \nu_2, \dots, \nu_n)$ and  $\mix {p'_1}{\nu_1}{\mix{p'_2}{\nu_2}{\dots \mix{p'_{n-1}}{\nu_{n-1}}{\nu_n}\dots }}$. To see this, note that
\begin{align*}
\mix {p'_1}{\nu_1}{\mix{p'_2}{\nu_2}{\dots \mix{p'_{n-1}}{\nu_{n-1}}{\nu_n}\dots }}
&=\mix{p'_1}{\nu_1}{f_{\vec P*}(\vec V')}\\
\flag{\vec Q=(1,0,\dots,0), \quad \vec Q'=(0,p_2,\dots,p_n)} \quad
&=\mix {p'_1}{f_{\vec Q}(\vec V)}{f_{\vec Q'}(\vec V)}\\
\flag{\text{Lemma~\ref{lemma:mixture_of_distributions}}}
&=f_{\vec R}(\vec V),
\end{align*} where $\vec R = p_1' \vec Q + (1-p_1') \vec Q'  = (p_1', (1-p_1') p_2', \dots , (1-p_1') p_n)$. Note that by construction  $\sum_i r_1 =1$  and  $r_k = p_k'/ (\Pi_{i<k}(1-p_i'))$. 
\end{proof}

\begin{lemma}[Repetitions in convex mixtures]
\label{lemma:mixture_repeated}
Let $(\Omega, f)$ be a convex structure and $\vec V  = (\nu_1, \nu_2, \dots, \nu_{n-2},\nu,\nu )$ a finite vector of elements of $\Omega$. Finally, let $\vec P = (p_1, \dots, p_n)$ be a probability distribution  vector. Then
\begin{align*}
f_{\vec P}(\vec V)=f_{\vec{\tilde P}}(\vec{\tilde V}),
\qquad
\vec {\tilde P}&= (p_1,p_2,\dots,p_{n-2},\ p_{n-1}+p_n), \\
\vec {\tilde V}&=(\nu_1,\nu_2,\dots,\nu_{n-2},\ \nu).
\end{align*}

\end{lemma}
\begin{proof}
We have
\begin{align*}
f_{\vec P}(\vec V)
&=\mix{ p'_1}{\nu_1}{\mix{ p'_2}{\nu_2}{\dots,\mix{ p'_{n-2}}{\nu_{n-2}}{\mix{ p'_{n-1}}{\nu}{\nu}}}}\\
\flag{\text{idempotence}}
&=\mix{\tilde p'_1}{\nu_1}{\mix{\tilde p'_2}{\nu_2}{\dots,\mix{\tilde p'_{n-2}}{\nu_{n-2}}{\nu}}}\\
\flag{\text{Lemma~\ref{lemma:convex_combination_reverse}}}
&=f_{\vec{\tilde P}}(\vec{\tilde V}),
\end{align*}
for $\tilde p'_i$ as in Definition~\ref{def:mixture_finite}. That the last element of $\vec{\tilde P}$ is indeed $p_{n-1}+p_n$ can be seen directly from the construction of $\vec{\tilde P}$ through a normalized probability distribution.

\end{proof}

\begin{lemma}
\label{lemma:convex_distributivity}
Let $(\Omega, f)$ be a convex structure,  $\vec V  = (\nu_1, \nu_2, \dots, \nu_{n-1},\mix{r}{\nu}{\omega} )$ a finite vector of elements of $\Omega$ and $\vec P = (p_1, \dots, p_n)$ be a probability distribution vector. Then
$$f_{\vec P}(\vec V)=f_{\vec{\tilde P}}(\vec{\tilde V})$$
with $\tilde{\vec V}=(\nu_1,\nu_2,\dots,\nu_{n-1},\nu,\omega)$ and $\vec{\tilde P}=(p_1,p_2,\dots,p_{n-1},rp_n,(1-r)p_n)$.
\end{lemma}
\begin{proof}
We have
\begin{align*}
f_{\vec P}(\vec V)&=\mix {p'_1} {\nu_1} {\mix {p'_2} {\nu_2} {\dots \mix{p'_{n-1}}{\nu_{n-1}} {\mix {r} {\nu} {\omega }   }  }} \\
\flag{\text{Lemma~\ref{lemma:convex_combination_reverse}}}&= f_{\vec{\tilde P}}(\tilde V),
\end{align*}
where the only thing left to check is that indeed the last two entries of $\vec{\tilde P}$ are $\tilde p_{n}=rp_n$ and $\tilde p_{n+1}=(1-r)p_n$ respectively. This can easily seen since by the above construction according to Definition~\ref{def:mixture_finite} $p_k'=p_k/(\Pi_{i<k}(1-p_i'))$, and so by Lemma~\ref{lemma:convex_combination_reverse}
$$\tilde p_n=r(\Pi_{i<n}(1-p_i'))=rp_n,$$
which guarantees by construction that $\tilde p_{n+1}=(1-r)p_n$ due to normalization.
\end{proof}

\subsubsection{Convex hull}

Taking the convex hull is essentially closing a subset under the operation of convex mixtures.

\begin{definition}[Convex hull]
\label{def:convex_hull}
Let $(\Omega, f)$ be a quasi-convex structure and $V$ a subset of $\Omega$.  The \emph{convex hull} $V^\P$ of $V$ is the smallest quasi-convex subset of $\Omega$ that contains $V$. 
\end{definition}

\begin{remark} \label{remark:convex_hull_properties}
  Let $(\Omega, f)$ be a quasi-convex structure, and $V, W \subseteq \Omega$. The operation of taking the convex hull has the following properties with respect to the partial order given by inclusion: 
  \begin{enumerate}
  \item  inflating, $V \subseteq V^\P$, 
  \item   order-preserving, $V \subseteq W \implies V^\P  \subseteq  W^\P$,
  \item $V$ is convex $\Leftrightarrow \ V^\P = V$,
  \item and in particular,  idempotent, $(V^\P)^\P = V^\P$,
  \item even if $V$ is finite, $V^\P$ may be uncountable,
  \item for any $\nu, \omega \in V$ and any $p \in [0,1]$, $\mix p \nu \omega \in V^\P$.
  \end{enumerate}
\end{remark}

\begin{lemma} \label{lemma:finite_convey_hull}
Let $(\Omega, f)$ be a convex structure and $V = \{\nu_1, \nu_2 , \dots, \nu_n\} \subset \Omega$ a finite subset. Then the convex hull of $V$ has the form 
$$V^\P = \bigcup_{P \in \P_V } \{ f_P(V) \}, $$
where $ \P_V$ is the set of all valid probability distributions over $V$ (that is, all $n$-element vectors $P =(p_1, p_2, \dots, p_n)$ satisfying all $p_i \geq 0$ and $\sum_i p_i =1$).
\end{lemma}

\begin{proof}
First we prove direction $\supseteq$, that is we show that if  $V^\P$ is a quasi-convex set that contains $V$, it must contain the right-hand side. For any valid probability distribution $P \in \P_V$, $f_P(V) \in V^\P$. To see this, observe firstly that all $q_i \in [0,1]$. Now  $x_{n-1} = \mix {q_{n-1}} {\nu_{n-1}} {\nu_n} \in V^\P$,
because it is the convex mixture of two elements of $V$. But then ${x_{n-2}} = \mix {q_{n-2}} {\nu_{n-1}} {x_n} $ 
is also in $V^\P$, because it is the convex mixture of two elements of $V^\P$. The same is true for $ \mix {q_{n-2}} {\nu_{n-2}} {x_{n-2}}  $ and so on until we reach $ x_1 = \mix {q_{1}} {\nu_1} {x_2} = f_P (V). $

To prove direction $\subseteq$, we need to show that the right-hand side is convex.
This is immediately given  by Lemma~\ref{lemma:mixture_of_distributions}. 
Indeed, for all $r \in [0,1]$ and any two elements $f_P(V), f_Q(V)$,  there exists a valid probability distribution $R = (r_1, r_2, \dots, r_n) \in \P_V$ such that $\mix r {f_P(V)}{ f_Q(V)}  = f_R(V)$. The coefficients of the new distribution are $r_i = r \ p_i + (r-1) \ q_i$.
\end{proof}

The following proposition will be used ahead to prove desirable properties of the convex hull.

\begin{proposition}[Finite cover for convex hull]
\label{prop:convex_hull_finite_subsets_cover}
Let $(\Omega, f)$ be a quasi-convex structure an $V$ a subset.  
We can take the convex hull of $V$ over a cover composed only of finite sets,   that is
$$V^\P=\bigcup_{\widetilde V\subseteq V, \ \widetilde V \text{finite}} \widetilde V^\P \quad = \bigcup_{\widetilde V\subseteq V, \ \widetilde V \text{finite}} \bigcup_{P \in \P_{\widetilde V}} \{f_P(\widetilde V)\}.$$
\end{proposition}

\begin{proof}

As a first step, we show that the set on the right-hand side is convex.
A convex combination of two elements on the right-hand side can be written as
$\omega= \mix p \nu {\nu'}$, 
with  $\nu\in\widetilde V_1^\P$ and $ \nu' \in\widetilde V_2^\P$, where $\widetilde V_1, \widetilde V_2 \subseteq V$ are finite subsets.
It follows that $\omega\in(\widetilde V_1\cup \widetilde V_2)^\P$. But $\widetilde V_1\cup \widetilde V_2$ is also finite, so $\omega\in \bigcup_{\widetilde V\subseteq V, \widetilde V \text{ finite}} \widetilde V^\P$.

Now we may prove direction $\subseteq$. We have
\begin{align*}
V &=\bigcup_{\widetilde V\subseteq V, \widetilde V \text{ finite}}\widetilde V \\
\flag{\cdot^\P \text{  inflating}}\qquad 
  &\subseteq \bigcup_{\widetilde V\subseteq V, \ \widetilde V \text{finite}} \widetilde V^\P \\
\flag{\cdot^\P \text{  order-preserving}} 
\implies V^\P 
  &\subseteq \left( \bigcup_{\widetilde V\subseteq V, \ \widetilde V \text{finite}} \widetilde V^\P  \right)^\P  \\
\flag{\text{right-hand side convex}} \qquad 
  &=  \bigcup_{\widetilde V\subseteq V, \ \widetilde V \text{finite}} \widetilde V^\P .
\end{align*}

Direction $\supseteq$ is straight-forward as any element of the right-hand side can be written as $\omega=\mix p \nu {\nu'}$, 
with   $\nu, \nu' \in\widetilde V \subseteq V$, therefore $\omega \in V^\P$. 
The final inequality follows directly from Lemma~\ref{lemma:finite_convey_hull}.
\end{proof}

\subsubsection{Connecting convex structures}

More generally, we can check whether a set is convex by mapping it to a known convex structure. Intuitively, the mapping is like a weaker version of an embedding that preserves convex mixtures. 

\begin{proposition}[Connecting convex structures]
\label{prop:connect_convex_structures}
Let $(\Sigma, f)$ be a  quasi-convex structure, and let $\Omega$ be a set. If there are two maps
$ \e: \Omega \to  \Sigma$ and $\h: \Sigma \to \Omega$ 
satisfying the two conditions
\begin{align*}
 \h \circ \e \, (\mu) = \mu, \qquad 
\h \circ \mix p {\e \circ \h ( \nu)} \omega  = \h \circ \mix p \nu \omega, 
\end{align*}
then $(\Omega,  f')$ is a quasi-convex structure, where we defined
\begin{align*}
f': [0,1] \times \Omega \times \Omega & \to \Omega \\
(p, \nu, \omega) &\mapsto  \h \circ  \mix p {\e(\nu) } {\e(\omega) } =: \mixg p \nu \omega.
\end{align*}
Furthermore, $\h$ is convexity-preserving, $\h \circ \mix p \nu \omega = \mixg p {\h(\nu) } {\h(\omega)}$.
If $(\Sigma, f)$ is a convex structure, then so is $(\Omega, \widetilde f)$.
\end{proposition}

\begin{proof}
Firstly,  $ f'$ is well defined, because  $\Sigma$ is convex, so  $\mix p {\e(\nu) } {\e(\omega) } \in \Sigma$.
The remaining properties follow directly from definition:

\begin{enumerate}

\item 
extremicity,
\begin{align*}
\mixg 1 \nu \omega  &=\h \circ  \mix 1 {\e(\nu) } {\e(\omega) } = \h \circ \e(\nu) = \nu;
\end{align*}

\item commutativity,
\begin{align*}
\mixg p \nu \omega  = 
 \h \circ \mix p {\e(\nu) } {\e(\omega) }  
 =  \h \circ (\mix {1-p}  {\e(\omega) } {\e(\nu) } = \mixg {1-p} \omega \nu  ;
\end{align*}

\item associativity,
\begin{align*}
\mixg  p {\mixg q \nu \omega  } \tau
&=  \h \circ 
   \mix p { \e \circ \h \circ \mix q {\e(\nu)} {\e(\omega)}} 
   {\e(\tau)}   \\
&=  \h \circ 
   \mix p {  \mix q {\e(\nu)} {\e(\omega)}} 
   {\e(\tau)}  \\
&= \h \circ
  \mix{p \, q} {\e(\nu)} {\mix {\frac{1-p}{1- p\,q}} {\e(\tau)} {\e(\omega)}  } \\
&= \h \circ
  \mix{1-pq} {\mix {\frac{1-p}{1- p\,q}} {\e(\tau)} {\e(\omega)}  } {\e(\nu)} \\
&= \h \circ
  \mix{1-pq} {\e\circ \mixg {\frac{1-p}{1- p\,q}} {\tau} {\omega}  } {\e(\nu)} \\
&= \mixg{1-pq} {\mixg {\frac{1-p}{1- p\,q}} \tau \omega  }  \nu \\
&= \mixg{p \, q} \nu {\mixg {\frac{1-p}{1- p\,q}} \tau \omega  }  ,
\end{align*}
 for $p\, q\neq 1$.
 
 \item idempotency,
\begin{align*}
\mixg p \nu \nu  &= \h \circ  \mix p {\e(\nu) } {\e(\nu) }  = \h \circ \e(\nu) = \nu;
\end{align*}

\end{enumerate}

Showing that $\h$ is convexity-preserving is also direct
\begin{align*}
\mixg p {\h(\nu) } {\h(\omega)}
&= \h \circ \mix p   {\e \circ \h (\nu) } {\e\circ \h (\omega) } \\
&= \h \circ \mix p  \nu  {\e\circ \h (\omega) } \\
&= \h \circ \mix {1-p}   {\e\circ \h (\omega) } \nu  \\
&= \h \circ \mix {1-p}   \omega \nu  \\
&= \h \circ \mix p \nu \omega.
\end{align*}

\end{proof}

\subsubsection{Convexity in specification spaces}

\begin{proposition}[Quasi-convex specification spaces]
\label{prop:convex_specifications}

Let $(\Omega, f)$ be a convex structure. Then $(S^\Omega , \widetilde f)$ is  quasi-convex, with the convex mixture of any two specifications $V, W \in S^\Omega$ defined as
\begin{align*}
  \mixt p V W
  := \bigcup_{\nu \in V} \bigcup_{ \omega \in W }
   \{ \mix p \nu \omega\}.
\end{align*}
\end{proposition}

\begin{proof} 

Firstly, it is clear that $\widetilde f$ is  well-defined, because $\Omega$ is convex.
Now, to show the properties:

\begin{enumerate}

\item 
extremity,
\begin{align*}
\mixt 1 V W  &=\bigcup_{\nu \in V} \bigcup_{\omega \in W}  \{\mix 1 \nu \omega \}
=\bigcup_{\nu \in V} \bigcup_{\omega \in W}  \{\nu\}
= V;
\end{align*}

\item commutativity,
\begin{align*}
\mixt p V W   = &=\bigcup_{\nu \in V} \bigcup_{\omega \in \Omega}  \{\mix p \nu \omega \}
= \bigcup_{\nu \in V} \bigcup_{\omega \in \Omega}  \{\mix {1-p} \omega \nu \}   = \mixt {1-p} W V ;
\end{align*}

\item associativity,
\begin{align*}
\mixt  p {\mixt q V W   } Z
&=  \mixt  p {\bigcup_{\nu \in V} \bigcup_{\omega \in \Omega}  \{\mix q \nu \omega \}} Z \\
 &= \bigcup_{\nu \in V} \bigcup_{\omega \in \Omega}  \bigcup_{\tau \in Z}  \{ \mix  p {\mix q \nu \omega  } \tau   \}  \\
 &=  \bigcup_{\nu \in V} \bigcup_{\omega \in \Omega}  \bigcup_{\tau \in Z} 
\left\{ \mix{p \, q} \nu {\mix {\frac{1-p}{1- p\,q}} \tau \omega  } \right\} \\
 &= \mixt{p \, q} V {\mixt {\frac{1-p}{1- p\,q}} Z W   }  ,
\end{align*}
 for $p\, q\neq 1$.
\end{enumerate}

\end{proof}

\begin{remark}

Specification spaces of convex state spaces are not idempotent. However,  
$\mixt p  V V  \subseteq V^\P. $
\end{remark}

\begin{proof} 
We have
\begin{align*}
\mixt p  V V   = \bigcup_{\nu \in V} \bigcup_{\omega \in Z}  \{\underbrace{\mix p \nu \omega}_{\flag {\in V^\P}} \} \subseteq   \bigcup_{\nu \in W} \bigcup_{\omega \in Z} V^\P = V^\P.
\end{align*}
\end{proof}

\begin{lemma}
\label{lemma:convex_mixture_order_preserving}
Let $(\Omega, f)$ be a convex structure,  and $(S^\Omega, \widetilde f)$ the corresponding convex specification space. 
Then $\widetilde f$ is order-preserving:
$V \subseteq W \implies \mixt p V Z \subseteq \mixt p W Z  .$

\end{lemma}

\begin{proof}
Follows directly from definition, as
\begin{align*}
\mixt p V Z &= \bigcup_{\nu \in V} \bigcup_{ \omega \in Z}  \{\mix p \nu \omega \} \subseteq \bigcup_{\nu \in W} \bigcup_{ \omega \in Z}  \{\mix p \nu \omega \} = \mixt p W Z.
\end{align*}
\end{proof}

\begin{lemma}
Let $(S^\Omega, \widetilde f)$  be a convex specification space and $V, W \in S^\Omega$ two specifications. Then
$$\mixt p V  W \subseteq (V\cup W)^\P.$$
\end{lemma}

\begin{proof}
 We use the fact that taking the convex hull is order-preserving, 
 \begin{align*}
  \mixt p V  W \subseteq  \mixt p {V \cup W}  {V\cup W} \subseteq  (V\cup W)^\P.
 \end{align*}
\end{proof}

\begin{proposition}[Convex embeddings] 
\label{prop:convex_embeddings}
Let $(\Sigma, f)$ be a convex structure and  $(S^\Sigma, \widetilde f)$ the corresponding quasi-convex specification space.
Let $S^\Omega$ be a specification space embedded in $S^\Sigma$, with intensive embedding $(\e, \h)$ and lumping $\Lump$, and let us define the family of functions $\{\widetilde f_p\}_{p \in [0,1]}$
\begin{align*}
\widetilde f'_p:  S^\Omega \times S^\Omega & \to S^\Omega \\
( V, W) &\mapsto  \h \circ  \mixt p {\e(V) } {\e(W) } =: \mixtg p V W.
\end{align*}
Then the following are equivalent:
\begin{enumerate}
\item $\Lump \circ \mixt p  V W  \supseteq \mixt p {\Lump V} {\Lump W}$,
\item $(S^\Omega, \widetilde f')$ is a quasi-convex structure and $\h$ is convexity-preserving, $\h \circ \mixt p V W = \mixtg p {\h(V)} {\h(W)}$.
\end{enumerate}
\end{proposition}

\begin{proof}
To show $1 \implies 2$, we will simply show that we are in the conditions of Proposition~\ref{prop:connect_convex_structures}, that is, that $ \h \circ \mixt p {\Lump (V)} W  = \h \circ \mixt p V W $.
We will do this by  sandwiching the left-hand side between two copies of the right-hand side. First, we apply $\h$ on both sides of the premise,
\begin{align*}
&\Lump \circ \mixt p  V W 
\supseteq \mixt p {\Lump (V)} {\Lump (W)}  
\supseteq \mixt p  {\Lump (V)} W  
\quad \flag{ f \text{ order-preserving}}\\
\implies 
&\h \circ \underbrace{\e \circ \h}_{\flag{\text{id}}} \circ \mixt p  V W \supseteq \h \circ \mixt p {\Lump( V)} W 
\supseteq  \h \circ \mixt p V W 
\quad \flag{\h, f \text{ order-preserving}}  \\
\iff 
&\h  \circ \mixt p  V W = \h \circ \mixt p {\Lump( V)} W .
\end{align*}

To show $2 \implies 1$, we  apply $\e$ on both sides, 
\begin{align*}
\h\circ \mixt p V W
 &= \mixtg p {\h(V)} {\h(W)}
 = \h \circ \mixt p {\e \circ \h (V) } {\e \circ \h (W) } \\
\implies 
\e \circ \h\circ \mixt p V W 
 &= \e \circ \h \circ \mixt p {\e \circ \h (V) } {\e \circ \h (W) } \\
\iff
\Lump \circ \mixt p V W 
 &= \Lump\circ  \mixt p {\Lump(V) } {\Lump(W) }  
 \supseteq  \mixt p {\Lump(V) } {\Lump(W) }.
\end{align*}

\end{proof}

\subsubsection{Probabilistic equivalence and convex quotient space}

The following lemma shows that taking the convex hull induces a cumulative equivalence relation. 
This means that there is a Galois insertion between a convex specification space and a quotient space based on probabilistic equivalence.

\begin{lemma} \label{lemma:convex_quasi}
Let $(\Omega, f)$ be a convex structure and  $(S^\Omega, \widetilde f)$  the corresponding quasi-convex specification space.
The function that takes specifications to their convex hulls,
\begin{align*}
\cdot^\P: \ S^\Omega &\to S^\Omega\\
V&\mapsto V^\P,
\end{align*}
is an idempotent, inflating quasi-endomorphism on $S^\Omega$ inducing a cumulative equivalence relation $\sim_{\P}$.
\end{lemma}

\begin{proof}
Idempotence and inflation follow from definition of convex hull (see Remark \ref{remark:convex_hull_properties}). 
To see that it is a quasi-endomorphism, we use the fact that it is order-preserving, 
\begin{align*}
V^\P \cup W^\P \subseteq (V \cup W)^\P  \cup (V \cup W)^\P   = (V \cup W)^\P .
\end{align*}
It also follows that $(V\cup W)^\P=(V^\P\cup W^\P)^\P$, as
\begin{align*}
\flag{\cdot^\P \text{ inflating and order-preserving}} \quad
(V\cup W)^\P & \subseteq(V^\P\cup W^\P)^\P   \\
\flag{\cdot^\P \text{ idempotent}} \quad 
(V\cup W)^\P &=((V\cup W)^\P)^\P\\
\flag{\cdot^\P \text{ quasi-endomorphism}} \qquad \qquad
&\supseteq(V^\P\cup W^\P)^\P.
\end{align*}
It remains to show that
$\sim_{\P}$ is cumulative, that is $V^\P = W^\P \implies V^\P = (V \cup W)^\P$.
Using the above, we  find
\begin{align*}
(V\cup W)^\P & =(V^\P\cup W^\P)^\P \\
\flag{V^\P=W^\P} \qquad \quad
 &=(V^\P)^\P=V^\P.
\end{align*}

\end{proof}

\begin{corollary}
There is a Galois insertion $(g_\P, h_\P)$ of the quotient space $S^\Omega/\P$ into $S^\Omega$, with \begin{align*}
g_\P: \quad S^\Omega/\P &\to S^\Omega \\
[V]_\P &\mapsto V^\P
\end{align*}
and
\begin{align*}
h_\P: \quad S^\Omega &\to S^\Omega/\P \\
V &\mapsto [V]_\P,
\end{align*}
and both $g_\P$ and $h_\P$ are quasi-homomorphisms. Furthermore,  $S^\Omega/\P$ is a join-semilattice, with the join induced by the Galois insertion.
\end{corollary}

\begin{proof}
Follows from Lemma \ref{lemma:convex_quasi} and Theorem \ref{thm:quasi_reduced_induced_Galois}.
\end{proof}\

The following remark is a direct application of Lemma~\ref{lemma:quasi_endomorphism_ss} that gives us some more properties of the convex hull.

\begin{remark}
Let $(\Omega, f)$ be a convex structure and let
$V, W \in S^\Omega$. Then
\begin{enumerate}
\item  $(V^\P \cup W^\P)^\P = (V \cup W)^\P$,
\item if $V \cap W \neq \emptyset$, then  $V^\P \cap W^\P \supseteq (V \cap W)^\P$,
\item if $V \cap W \neq \emptyset$, then  $(V^\P \cap W^\P)^\P = V^\P \cap W^\P$.
\end{enumerate}
\end{remark}

\subsubsection{Consistency of the probabilistic equivalence}

Whenever convex mixtures can be interpreted as probabilistic mixtures, it might be convenient to work in the quotient space: the equivalence class $[V]_\P$ may be represented by both  $V$ and $V^\P$. In some instances, it might be useful to take $V^\P$ as the representative of the equivalence class, as it displays the probabilistic interpretation of specifications most explicitly. However, $V_\P$ may be an uncountable set, and so often it will be easier to work with another countable representative, such as the specification that contains only the extremal points of the set $V_\P$. Given that we work with the quotient space, in particular moving between elements of $[V]_\P$ should always be allowed. In the following, we show that this is indeed possible. However, some of the following propositions show that there are some conditions that the resource theory needs to satisfy to respect probabilistic equivalence.

\begin{proposition}[Consistency of convex hull with convexity-preserving homomorphisms]
\label{prop:convex_hull_preserving_homomorphism} Let $(\Omega, f)$ and $(\Sigma, f')$ be two convex structures and $(S^\Omega, \widetilde f )$, $(S^\Sigma, \widetilde f' )$  the respective quasi-convex specification spaces. 
Let $g$ be a convexity-preserving homomorphism from $S^\Omega$ to  $S^\Sigma$. Then $g(V^\P)=(g(V))^\P$ for any specification $V\in S^\Omega$.
\end{proposition}

\begin{proof}
First we will show $\subseteq$, using Proposition~\ref{prop:convex_hull_finite_subsets_cover} (that is, by reducing to the case of finite covers of a set $V$; this is a trick to avoid uncountable problems, and there may exist a more direct proof for those better versed in infinities and measures than us). We find
\begin{align*}
g(V^\P)
 &=g\left(
   \bigcup_{\widetilde V\subseteq V,\widetilde V\text{ finite}} \bigcup_{P \in \P_{\widetilde V}} \{f_P(\widetilde V) \} \right)
   \quad \flag{\text{Proposition \ref{prop:convex_hull_finite_subsets_cover}  }}\\
\flag{g \text{ homomorphism}}
 &=\bigcup_{\widetilde V\subseteq V,\widetilde V\text{ finite}} \bigcup_{P \in \P_{\widetilde V}} g\circ f_P(\widetilde V)\\
 \flag{\widetilde V=\{\nu_1,\dots,\nu_n\}}
 &=\bigcup_{\widetilde V\subseteq V,\widetilde V\text{ finite}} \bigcup_{P \in \P_{\widetilde V}} g(\{ \mix{p_1}{\nu_1}{\dots,\mix{p_{n-1}}{\nu_{n-1}}{\nu_n}}\})\\
 \flag{g\text{ convexity-preserving}}
 &=\bigcup_{\widetilde V\subseteq V,\widetilde V\text{ finite}} \bigcup_{P \in \P_{\widetilde V}} \mixt{p_1}{g(\{\nu_1\})}{\dots,\mixt{p_{n-1}}{g(\{\nu_{n-1}\})}{g(\{\nu_n\})}}\\
 &=\bigcup_{\widetilde V\subseteq V,\widetilde V\text{ finite}} \bigcup_{P \in \P_{\widetilde V}}\bigcup_{\nu'_1\in g(\nu_1)}\dots \bigcup_{\nu'_n\in g(\nu_n)} \{\mix{p_1}{\nu'_1}{\dots,\mix{p_{n-1}}{\nu'_{n-1}}{\nu'_n}}\}\\
&\subseteq \bigcup_{\widetilde V\subseteq g(V),\widetilde V\text{ finite}} \widetilde V^\P\\
&= (g(V))^\P
\end{align*}
To see  direction $\supseteq$, note that $g(V^\P)$ is convex: take any two  elements
$$
x_P, x_{P'} \in g(V^\P) =
   \bigcup_{\widetilde V\subseteq V,\widetilde V\text{ finite}} \bigcup_{P \in \P_{\widetilde V}}  g \circ f_P (\widetilde V)$$
and consider their convex combination, $\mix{r}{x_P}{x_{P'}}$, which for some finite $\tilde V_1,\tilde V_2\subseteq V$ satisfies
\begin{align*}
 \mix r  {x_P}  {x_{P'}}
 &\in \mixt r {g \circ f_P (\tilde V_1)} { g\circ f_{P'} (\tilde V_2)}\\
\flag{g\text{ convexity-preserving}}
&= g \circ  \mix r { f_P (\tilde V_1)} {  f_{P'} (\tilde V_2)} \\
\flag{\text{Lemma~\ref{lemma:mixture_of_distributions}}}&= g \circ f_{\vec R}(\vec {\widetilde {V'}}),
\end{align*}
where $\tilde V_1=\{\nu_1,\dots,\nu_n\},\tilde V_2=\{\omega_1,\dots,\omega_m\},
\vec {\tilde {V'}}:=(\nu_1,\dots \nu_n,\omega_1,\dots \omega_m)$,
$r_i = r \tilde p_i + (1-r) \tilde q_i$, and $\tilde P, \tilde Q \in \P_{\widetilde{V'}}$ are extensions of the original probability distributions on $\vec{\widetilde {V'}}$ (such that $\tilde P$ restricted to   $\tilde V_1$ equals $P$, and is zero for elements $\omega_i$ of $\tilde V_2$, and vice-versa for $\tilde Q$).
This element is again in $g(V^\P)$ because $f_{\vec R}(\vec{\widetilde {V'}})\subseteq V^\P$.
Then, $g(V^\P)=(g(V^\P))^\P\supseteq (g(V))^\P$.

\end{proof}

\begin{corollary}[Consistency of convex hull with approximations]
If the endomorphisms making up an approximation structure $A_\E$ are convexity-preserving, then for any $V\in S^\Omega,\epsilon\in\E$, 
$$(V^\epsilon)^\P=(V^\P)^\epsilon.$$
\end{corollary}

\begin{corollary}[Consistency of convex resource theories]
\label{corollary:consistency_convex_resource_theories}
Let $(S^\Omega, \cT)$ be a convex resource theory, and let $V, W \in S^\Omega$. Then $V \to W$ implies $V^\P \to W^\P$. 
\end{corollary}
\begin{proof}
Let $g \in \cT$ be the transformation that achieves $g(V) \subseteq W$. Then $g(V^\P) = g(V)^\P \subseteq W^\P$.
\end{proof}

\begin{corollary}[Consistency of convex hull with embeddings]
\label{corollary:consistency_convex_embeddings}
Let $S^\Omega$ be a convex specification space and let there be an intensive embedding $(\e,\h)$ with  lumping $\Lump=\e\circ\h$. If this lumping is quasi-convexity preserving, that is 
$$\Lump(p\ V+(1-p)\ W)\supseteq p\ \Lump(V)+(1-p)\ \Lump(W),$$
then the embedding preserves the action of taking the convex hull:
$\Lump(V^\P)=\Lump((\Lump(V))^\P)$, or equivalently
$\h(V^\P)=(\h(V))^\P$.
\end{corollary}
\begin{proof}
From Proposition~\ref{prop:convex_embeddings}, we know that $\h$ is convexity-preserving.  By Proposition~\ref{prop:convex_hull_preserving_homomorphism} this immediately implies $\h(V^\P)=(\h(V))^\P$.
\end{proof}

\begin{proposition}[Consistency of convex hull with convex combination] \label{prop:consistency_hull_combination}
Let $(\Omega, f)$ be a convex structure and $(S^\Omega, \widetilde f)$ the corresponding convex specification space. Then
$$\mixt r V W^\P = \mixt r {V^\P} {W^\P}, $$
 for any $V,W\in S^\Omega$ and $r\in[0,1]$.
\end{proposition}

\begin{proof} 
For the case of finite specifications $V$ and $W$ of size $n$ and $m$ respectively, we find 

\begin{align*}
\mixt r V W ^\P
&= \left(
\bigcup_{\nu \in V} \bigcup_{\omega \in W}
\{ \mix r \nu \omega \}
\right)^\P \\ 
&= 
\{\mix r {\nu_1 } {\omega_1 },\mix r {\nu_1 } {\omega_2 }, 
\dots,
\mix r {\nu_n } {\omega_m } \}^\P \\ 
\flag{*}
&= \bigcup_{P\in \P_{V\times W}} f_{\vec P} 
\big(
( \mix r {\nu_1 } {\omega_1 },\mix r {\nu_1 } {\omega_2 }, 
\dots,
\mix r {\nu_n } {\omega_m })
\big) \\
\flag{**}
&=\bigcup_{P\in \P_{V\times W}}
f_{\vec{\tilde P}} 
\left(
{\vec Z}
\right) \\
\flag{***}
&=\bigcup_{P\in \P_{V\times W}}
f_{\pi(\vec{\tilde P})} 
\left(
\pi({\vec Z})
\right) \\
\flag{****}
&=\bigcup_{P\in \P_{V\times W}}
f_{\vec Q} 
\left( \vec {X} \right)\\
&=\bigcup_{P_V\in \P_{V},P_W\in \P_W}
\mix r {f_{P_V}(V)}{f_{P_W}(W)}\\
&=\mixt r
{\bigcup_{P_V\in\P_V} f_{P_V}(V)}
{\bigcup_{P_W\in\P_W} f_{P_W}(W)}\\
&=\mixt r{V^\P}{W^\P},\\
\end{align*}
where in $\flag{*}$ we technically need  Lemma~\ref{lemma:mixture_repeated} as the elements $f_r(\nu,\omega)$ need not be different for different $\nu$ and $\omega$,
and take 
$
\vec P= (p_{1,1}, p_{1,2}, \dots, p_{n,m})
$.
For \flag{**} we  combine the distributions, defining 
\begin{align*}
\vec{ Z}
&=(\nu_1, &&\omega_1, &&\nu_1, &&\omega_2, &&\dots, &&\nu_n, &&\omega_m &),\\
\vec{\tilde P}&=(r\  p_{1,1}, &&(1-r) \ p_{1,1}, &&r\ p_{1,2}, &&(1-r) \ p_{1,2}, &&\dots, &&r\ p_{n,m}, &&(1-r) \ p_{n,m} &).
\end{align*}
For \flag{***} we take a permutation that orders all the identical elements together,
\begin{align*}
\pi(\vec{ Z})
&=(\underbrace{\nu_1,\dots,\nu_1}_{\flag{m\text{ terms}}},
\underbrace{\nu_2,\dots,\nu_2}_{\flag{m\text{ terms}}},
\dots,
\underbrace{\nu_n,\dots,\nu_n}_{\flag{m\text{ terms}}},
\underbrace{\omega_1,\dots,\omega_1}_{\flag{n\text{ terms}}},
\underbrace{\omega_2,\dots,\omega_2}_{\flag{n\text{ terms}}},
\dots,
\underbrace{\omega_m,\dots,\omega_m}_{\flag{n\text{ terms}}}),
 \\
\pi(\vec{\tilde P})
&=(\underbrace{r \ p_{1,1},\dots, r \ p_{1,m}}_{\flag{m\text{ terms}}},
\dots,
\underbrace{r \ p_{n,1},\dots, r \ p_{n,m}}_{\flag{m\text{ terms}}},
\underbrace{(1-r) \ p_{1,1},\dots, (1-r) \ p_{n,1}}_{\flag{n\text{ terms}}},
\dots,
\underbrace{(1-r) \ p_{1,m},\dots, (1-r) \ p_{n,m}}_{\flag{n\text{ terms}}}).
\end{align*}
For \flag{****} we use Lemma~\ref{lemma:mixture_repeated} to sum over the probabilities of the repeated elements, such that
\begin{align*}
\vec X
&=(\nu_1,\nu_2,\dots,\nu_n,\omega_1,\omega_2,\dots,\omega_m) ,\\
\vec Q &= \left( \underbrace{r \sum_{j}^m p_{1, j}}_{\flag{r \ p_1^V}}, \underbrace{r \sum_{j}^m p_{2, j}}_{\flag{r \ p_2^V}}, \dots,  
\underbrace{r \sum_{j}^m p_{n, j}}_{\flag{r \ p_n^V}}, \underbrace{(1-r) \sum_{i}^n p_{i, 1}}_{\flag{(1-r) \ p_1^W}}, 
\underbrace{(1-r) \sum_{i}^n p_{i, 2}}_{\flag{(1-r) \ p_2^W}} , \dots, 
\underbrace{(1-r) \sum_{i}^n p_{i, m}}_{\flag{(1-r) \ p_m^W}} \right) ,\\
\vec{P_V} &= (p_1^V, p_2^V, \dots, p_n^V), 
\qquad
\vec{P_W} = (p_1^W, p_2^W, \dots, p_m^W).
\end{align*}
For the infinite case, we first show that $(\mixt p{V} {W})^\P\subseteq \mixt p {V^\P}{W^\P}$. To see this, note that the expression on the right-hand side describes a convex set, because
$$\mixt p{V^\P}{W^\P}=\bigcup_{\nu\in V^\P, \omega\in W^\P}\{\mix p \nu \omega\},$$
and so for any $\gamma,\tau\in \mixt p{V^\P}{W^\P}$, and any $r\in[0,1]$, there exist $\nu,\nu'\in V^\P$ and $\omega,\omega'\in W^\P$ such that
\begin{align*}
\mix r \gamma \tau
&=\mix r{\mix p \nu \omega}{\mix p {\nu'} {\omega'}}\\
\flag{\text{Lemma~\ref{lemma:convex_distributivity}}}
&=f_{\vec R}(\vec X)\\
\flag{\text{Lemma~\ref{lemma:convex_distributivity}}}
&=\mix p{\mix q \nu {\nu'}}{\mix q\omega{\omega'}}\\
&\in \mixt p {V^\P}{W^\P},
\end{align*}
where
\begin{align*}
\vec X&=(\nu,\nu',\omega,\omega'), \qquad
\vec R=(pq,p(1-q),(1-p)q,(1-p)(1-q)).
\end{align*} 
Next, we can use convexity of $\mixt p {V^\P}{W^\P}$ to conclude that
$(\mixt p{V} {W})^\P\subseteq \mixt p {V^\P}{W^\P}$,
as taking the convex hull is order-preserving:
$$(\mixt p V W)^\P\subseteq (\mixt p{V^\P}{W^\P})^\P=\mixt p{V^\P}{W^\P}.$$
To show the other direction, we use Proposition \ref{prop:convex_hull_finite_subsets_cover} and find
\begin{align*}
\mixt p{V^\P}{W^\P}
&=\bigcup_{V'\subseteq V, V'\text{ finite}}\bigcup_{W'\subseteq W, W'\text{ finite}}
\mixt p{V'^\P}{W'^\P}\\
\flag{\text{finite proof applies}}
&=\bigcup_{V'\subseteq V, V'\text{ finite}}\bigcup_{W'\subseteq W, W'\text{ finite}}
(\mixt p{V'}{W'})^\P\\
\flag{\mixt p{V'}{W'}\text{ finite subset of } \mixt p{V}{W}}
&\subseteq\bigcup_{X\subseteq \mixt p{V}{W}, X\text{ finite}}
X^\P\\
&=(\mixt p V W)^\P
\end{align*}
\end{proof}

\begin{remark}
\label{remark:ambiguous_convexity}

The above consistency checks might lead us to believe that we may always work in the probabilistic quotient space. One must proceed with caution though, because  taking the convex hull of a specification might make us lose information about the larger picture, and in particular about correlations.  For example, imagine that $S^\Omega$ is the specification space of all two-qubit density matrices, and $S^{\Omega_A}$ the reduced space of the first qubit.  Now imagine that we had the local specification $V_A =$ `first qubit is in a pure state', corresponding to the surface of a Bloch sphere. The global version of this specification is $\e_A(V_A) =$ `first qubit is in a pure state, and so in particular it is uncorrelated with the other'. Now, if we took the convex hull of the local specification we would obtain all mixtures of pure states, that is the whole Bloch sphere, $V_A^\P= \Omega_A$. Taking the global version of this specification yields $\e_A(V_A^\P) = \Omega$, because now we allow for correlated states that purify (or more generally extend) mixed states on the first qubit. In taking the convex hull we lose the information that this qubit was uncorrelated: the new global specification $\Omega$ allows for classical correlations as well as entanglement. While we might have also obtained classical correlations by taking the convex hull of the initial global specification $\e_A(V_A)$, and so we could understand them as lack of knowledge already inherent in $\e_A(V_A)$,
we are now also including physical correlations, i.e.\ entanglement, into the specification.  
This is due to the dual character of convexity in quantum theory: the density matrix formalism may reflect both a probabilistic, subjective lack of knowledge as well as physical correlations, i.e.\ entanglement, with external systems.

\end{remark}

\subsubsection{Convex resource theories}

\begin{definition}[Convex resource theories]
We say that a resource theory $(S^\Omega,\cT)$ is \emph{doubly convex} if $(\Omega, f)$ is a convex structure and the transformations in $\cT$ preserve convexity, that is 
\begin{align*}
\mixt p {g(V)} {g(W)}  =  g (\mixt pVW) ,
\end{align*}
for all $V, W \in S^\Omega$ and $g \in \cT$.
\end{definition}

\begin{lemma}
Let $(S^\Omega, \cT)$ be a convex resource theory, let $V, W, Z \in S^\Omega$, and let $0\leq p \leq 1$. Then $(V\cup W) \to Z$ implies $\mixt p V W \to Z^\P$. 
\end{lemma}
  
\begin{proof}
It follows from $\mixt p V W \subseteq (V\cup W)^\P$ together with Corollary~\ref{corollary:consistency_convex_resource_theories}.
\end{proof}

\begin{definition}[Doubly convex resource theory]
We say that a resource theory $(S^\Omega,\cT)$ is \emph{doubly convex} if $(S^\Omega,\cT)$ is convex and $(\cT,c)$ is a convex structure, where $c$ stands for a family of functions $\{c_p\}_{p \in [0,1]}$ ,
\begin{align*}
c_p: \cT\times\cT&\to\cT\\
(g,g')\to \mixf p g {g'},
\end{align*}
which in addition to the properties of a convex structure, satisfy
$$\mixf p  g {g'}(V)\sim_\P \mixt p {g(V)}{g'(V)}$$
for all $V\in S^\Omega$, 
and respect the monoid structure of transformations as
\begin{align*}
\mixf p {f \circ g}{f \circ g'}
&= f \circ \mixf p g {g'}\\
\mixf p {g \circ f}{g' \circ f}
&= \mixf p g {g'} \circ f.
\end{align*}
\end{definition}

\begin{lemma}\label{lemma:double_convex}
Let $(S^\Omega, \cT)$ be a doubly convex resource theory, let $V, W, Z \in S^\Omega$, and let $0\leq p \leq 1$. Then $V \to W$ and $V \to Z$ implies $V \to (\mixt p W  Z)^\P$. 
\end{lemma}

\begin{proof}
Let $f$ be the transformation that achieves $f(V)\subseteq W$ and $g$ the transformation that achieves $g(V)\subseteq Z$. Then
$$\mixf p f g(V) 
\subseteq [\mixf p f g(V)]^\P =    (\mixt p{f(V)} { g(V)})^\P
\subseteq (\mixt p W  Z)^\P.$$
\end{proof}

\subsection{Proofs of claims from the main text}

\begin{proposition}[Quantum theory is doubly convex]
\label{prop:quantum_doubly_convex}
The resource theory of quantum mechanics is doubly convex.
\end{proposition}
\begin{proof}
We first show that the specification space $S^\Omega$ for a state space $\Omega$ of density matrices on a finite-dimensional Hilbert space is convex. To see this, we take the usual convex combination
$$\mix p \rho \sigma = p \  \rho + (1-p) \  \sigma,$$
and note that, because density matrices form a subspace of a real vector space, $(\Omega,f)$ is a convex structure as required, with $(S^\Omega,\widetilde f)$ the corresponding quasi-convex structure on specification space. Next, we show that the allowed transformations in quantum mechanics are convexity preserving:
$$ \E (p  \  \rho + (1-p) \  \sigma) = p \  \E (\rho) + (1-p) \  \E (\sigma)$$
by linearity for any CPTP map $\E$. Therefore the  extensions on the specification space,  $\widetilde \E(V) = \bigcup_{\rho \in V} \{ \E (\rho)\}$,  satisfy the required property
$$ \widetilde \E (\mixt p V W) = \mixt p {\widetilde \E (V)} {\widetilde \E (W)}.$$
Finally, we need to show that the transformations are themselves convex, and that their convex combinations satisfy
$$\mixf p {\widetilde \E} {\widetilde \E'} (V)\sim_\P \mixt p {\widetilde \E(V)}{\widetilde \E'(V)}$$
for all $V\in S^\Omega$, 
and respect the monoidal structure.
Since CPTP maps take quantum states to quantum states, we can simply define the convex combination $c$ of transformations through
$$\mixf p {\widetilde \E} {\widetilde \E'} (\{\rho\})
= \{ p \ \E (\rho) + (1-p) \ \E' (\rho) \}.$$
This definition clearly gives a convex structure since $\Omega$ is a subset of a vector space (a convex structure), and by linearity  
\begin{align*}
\mixf p {\widetilde \E \circ \widetilde \E'}{\widetilde \E \circ \widetilde \E''}
&= \widetilde \E \circ \mixf p {\widetilde \E'} {\widetilde \E''}\\
\mixf p {\widetilde \E' \circ \widetilde \E}{\widetilde \E'' \circ \widetilde \E}
&= \mixf p {\widetilde \E} {\widetilde \E''} \circ \widetilde \E
\end{align*}
as required. Finally we note that
\begin{align*}
\mixf p {\widetilde \E} {\widetilde \E'} (V) 
&= \bigcup_{\rho \in V} 
\{ p \ \E (\rho) + (1-p) \ \E'(\rho) \} \\
&= \mixt p {\widetilde \E (V)} {\widetilde \E' (V)}
\end{align*}
and so the required convex hull equivalence between the two sides trivially holds.
\end{proof}

 \thmFreeConvex*

\begin{proof}
Since by assumption 
$\Omega \to V$ and 
$\Omega \to W$, Lemma~\ref{lemma:double_convex} tells us that 
$\Omega \to (\mixt p {V} {W})^\P$.
For the special case of free states, this becomes 
$\Omega \to  (\mixt p {\{ \nu\} } {\{\omega\}})^\P = \{\mix p \nu \omega\}^\P = \{\mix p \nu \omega\}$.
\end{proof}

\thmConvexReduced*

\begin{proof}

We denote the Galois insertion of $S^T$ in $S^\Omega$ by $(\e,\h)$. Note that the condition on $\Lump$ for general convex structures $(\Omega,f)$ reads
$$\Lump (\mixt p V  W)\  \supseteq
\mixt p {\Lump (V)}{\Lump (W)}.$$
We follow Proposition~\ref{prop:convex_embeddings} in defining the quasi-convex mixture in the reduced specification space,
\begin{align*}
\widetilde f'_p:\quad {S^T}\times {S^T}&\to S^T\\
(V_T,W_T)
&\mapsto \h(\mixt p{\e(V_T)}{\e(W_T)}).
\end{align*}
Note that according to Proposition~\ref{prop:convex_embeddings}  $\h$ is convexity-preserving.
We will show that $\mixtg  p {\{\omega_T\}}{\{\nu_T\}}=\{\mu_T\}$ for some $\mu_T\in T$ and that  there is a family of functions $\{f'_p\}_{p \in [0,1]}$ such that  $(T, f')$ is a convex structure and  
$$\mixtg p {V_T} { W_T} = \bigcup_{\nu_T \in V_T} \bigcup_{\omega_T \in W_T} \{\mixg p {\nu_T}{ \omega_T}\} .$$
By assumption  we have
$$\e\circ \h(\mixt p {\{\omega\}} {\{\nu\}}) \supseteq \mixt p {\e\circ \h(\{\omega\})}{\e\circ \h(\{\nu\})}.$$
Acting on both sides with $\h$ and using the facts that $\h$ is order-preserving and that $\h\circ\e$ is the identity gives
$$\h(\mixt p{\{\omega\}}{\{\nu\}}) \supseteq \h(\mixt p{\e\circ \h(\{\omega\})}{\e\circ \h(\{\nu\})}).$$
By construction of the embedding, singletons are reduced into singletons: $\omega_T\in {T}\iff \{\omega_T\}=\h(\{\omega\})$
for some $\omega\in \Omega$, and in turn, for all $\omega\in \Omega$, $\h(\{\omega\}) = \{\omega_T\}$ for some $\omega_T\in T$.
Therefore, 
the right-hand side  equals $\widetilde f'_p(\{\omega_T\},\{\nu_T\})$. 
On the left-hand side, convexity of $\Omega$ implies
$\mixt p {\{\omega\}}{\{\nu\}} = \{\mu\}$
for some $\mu\in\Omega$,  so the left-hand side is
$\h(\{\mu\})=\{\mu_T\}$. Now, since this is a singleton,
$\supseteq$ becomes $=$ and so
$$\h(\mixt p{\{\omega\}}{\{\nu\}})
= \widetilde f'_p(\{\omega_T\},\{\nu_T\}).$$
With this, we have shown that
$\widetilde f'_p\ (\{\omega_T\},\{\nu_T\})=\{\mu_T\} =: \{\mixg{p}{\omega_T}{\nu_T}\}$
for some $\mu_T\in T$.
It is left to verify that $(T,f'_p)$ is indeed a convex structure, namely that it is idempotent. This is guaranteed because
$$\widetilde f'_p(\{\omega_T\},\{\omega_T\})
= \h(\mixt p{\{\omega\}}{\{\omega\}})
= \h(\{\omega\}) = \{\omega_T\}$$
for some $\omega\in\Omega$ such that $\h(\{\omega\})=\{\omega_T\}$.
Hence $(T, f')$ is a convex structure and $(S^T, \widetilde f')$ its quasi-convex specification space.

\end{proof}

\subsection{Additional results}

The following theorem shows that, under suitable conditions, not only is the reduced state space of a convex state space also convex, but also we can find restricted resource theories $(S^T,\widetilde F)$ that are also doubly convex. We thank the ITP-ETHZ Oktoberfest for inspiration for this theorem.
The condition imposed on the embedding guarantees that we can define the convex combination of reduced transformations uniquely.  

\begin{theorem}[Double convex embeddings]
\label{thm:embedding_convex2}

Let $(S^\Omega, \cT)$ be a doubly convex resource theory, let $\Lump$ be a lumping on $S^\Omega$ inducing a reduced specification space $S^T$ with an intensive embedding $\e$ that preserves a convex set of endomorphisms $G\subseteq\cT$. 
Then if $\Lump$ satisfies  
\begin{align*}
\Lump (\mixt p V  W)\  \supseteq
\mixt p {\Lump (V)}{\Lump (W)},
\end{align*}
and the convex combinations of functions satisfies 
\begin{align*}
&\h \circ f \circ \e = \h \circ f' \circ \e
\implies 
\h\circ \mixf p f g \circ \e = \h \circ \mixf p {f'} g \circ \e,
\end{align*}
then the restricted resource theory $(S^T, \widetilde G)$ is also doubly convex.
\end{theorem}

\begin{proof}
It follows from Theorem \ref{thm:embedding_convex} that $G$ is convex. It is left to show that
\begin{enumerate}
\item the operations $\widetilde g=\h \circ g \circ \e \in \widetilde F$ for $g\in G$ preserve convexity,
\item the set of functions $\widetilde g\in\widetilde G$ is convex.
\end{enumerate}

We will now prove the above points.

\begin{enumerate}
\item Since
$\widetilde g (V_T) = \h \circ g\circ\e(V_T)$,
we have
\begin{align*}
\widetilde g(\widetilde f'_p (V_T,W_T))
&= \h\circ g\circ \e\circ\h(\mixt p {\e(V_T)} {\e(W_T)})\\
\flag{\Lump \text{ preserves } G}&=\h\circ g\circ \mixt p {\e(V_T)} {\e(W_T)}
\\
\flag{g \text{ convexity preserving}}
&=\h\circ \mixt p {g\circ\e(V_T)}{g\circ \e(W_T)}\\
\flag{\h \text{ convexity preserving, Proposition~\ref{prop:convex_embeddings}}}
&=\widetilde f'_p (\h\circ g\circ \e(V_T), \h\circ g\circ\e(W_T))\\
&=\widetilde f'_p (\widetilde g(V_T),\widetilde g(W_T)).
\end{align*}

\item 
Let $f, g \in G$.
Since $G$ is convex, $\mixf p f {g} \in G$. 

Now, since we require that
\begin{align*}
&\h \circ f \circ \e = \h \circ f' \circ \e
\implies 
\h\circ \mixf p f g \circ \e = \h \circ \mixf p {f'} g \circ \e,
\end{align*}
we can define the convex combination of elements in $\widetilde G$ through
\begin{align*}
\widetilde c_p: \widetilde G \times \widetilde G
&\to \widetilde G\\
(\widetilde f,\widetilde g)
&\mapsto \h \circ \mixf p f g \circ \e
\end{align*}
for any $f,g\in G$ such that
$\h\circ f \circ \e = \widetilde f$ and
$\h \circ g \circ \e = \widetilde g$.
Now, we have to show that $(\widetilde G, \widetilde{c})$ is idempotent, extremal, commutative and associative, that it satisfies the probabilistic equivalence as
$\mixft p {\widetilde f} {\widetilde g} (V_T)\sim_\P \mixtg p {\widetilde f(V_T)}{\widetilde g(V_T)}$, and that  $\widetilde{c_p}$ respects the monoidal structure of  $G$.

 Idempotence, commutativity and extremicity follow directly from the same properties in $c$,
$$\mixft p {\widetilde f}{\widetilde f}
= \h \circ \mixf p f f \circ \e 
= \h \circ f \circ \e = \widetilde f,$$ 
$$\mixft p {\widetilde f}{\widetilde g}
= \h \circ \mixf p f g \circ \e
= \h \circ \mixf {1-p} g f \circ \e
= \mixft {1-p} {\widetilde g}{\widetilde f},$$
$$\mixft 1 {\widetilde f}{\widetilde g}
= \h \circ \mixf 1 f g \circ \e
= \h \circ f \circ \e
= \widetilde f.$$

For associativity, we have
\begin{align*}
\mixft p {\mixft q {\widetilde f} {\widetilde g}}{\widetilde t}
&= \mixft p {\h \circ \mixf q f g \circ \e} {\widetilde t} \\
&= \h \circ \mixf p {\mixf q f g}{t} \circ \e\\
&= \h \circ \mixf {1-pq} {f}
{\mixf {\frac{1-p}{1-pq}} {t} {g}} \circ \e \\
&= \mixft {1-pq} {\h \circ f \circ \e} 
{\h \circ \mixf {\frac{1-p}{1-pq}} {t} {g} \circ \e} \\
&= \mixft {1-pq} {\widetilde f} 
{\mixft {\frac{1-p}{1-pq}} {\widetilde t} {\widetilde g}}.
\end{align*}

 To see that the probabilistic equivalence condition is still satisfied, we calculate
\begin{align*}
\left[\mixft p {\widetilde f} {\widetilde g} (V_T)\right]^\P
&= \left[\h \circ \mixf p f g \circ \e (V_T)\right]^\P\\
\flag{\text{Corollary~\ref{corollary:consistency_convex_embeddings}}}
&= \h \left( \left[ \mixf p f g \circ \e (V_T) \right]^\P \right)\\
&= \h \left( \left[ \mixt p {f \circ \e (V_T) } {g \circ \e (V_T)} \right]^\P \right) \\
&= \left[ \h \left( \mixt p {f \circ \e (V_T) } {g \circ \e (V_T)} \right) \right]^\P  \\
\flag{\text{Proposition~\ref{prop:convex_embeddings}}}
&= \left[ \mixtg p {\h \circ f \circ \e (V_T) } {\h \circ g \circ \e (V_T)} \right]^\P\\
&= \left[ \mixtg p {\widetilde f (V_T)} {\widetilde g (V_T)}\right]^\P.
\end{align*}

Finally, for the monoidal structure, we take the equivalent property from  $c$ and compute
\begin{align*}
\mixft p {\widetilde f \circ \widetilde g}{\widetilde f \circ \widetilde g'}
&= \h \circ \mixf p {f \circ g} {f \circ g'} \circ \e \ \flag{\text{embedding preserves } f}\\
&= \h \circ f \circ \mixf p {g} {g'} \circ \e \\
\flag{\text{embedding preserves } f}
&= \widetilde f \circ \mixft p {\widetilde g} {\widetilde g'},
\end{align*}
and similarly,
\begin{align*}
\mixft p {\widetilde g \circ \widetilde f}{\widetilde g' \circ \widetilde f}
&= \h \circ \mixf p {g \circ f} {g' \circ f} \circ \e \ \flag{\text{embedding preserves } g, g'}\\
&= \h \circ \mixf p {g} {g'} \circ f \circ \e \\
\flag{\text{embedding preserves } \mixf p {g} {g'}}
&= \mixft p {\widetilde g} {\widetilde g'} \circ \widetilde f,
\end{align*}
where the preservation steps were allowed because
$\h \circ f \circ g \circ \e 
= \h \circ f \circ \e \circ \h \circ g \circ \e
= \widetilde f \circ \widetilde g$.

\end{enumerate}
\end{proof}

The above is in particular true for for embeddings that preserve subsystems of specifications.

\begin{lemma}\label{lemma:subsystem_doubly_convex}
Let $(S^\Omega, \cM)$ be a doubly convex resource theory. Then every subsystem $A \in \Sys(\cM)$ is convex.
\end{lemma}

\begin{proof}
Let $g_A, g'_A \in A$. We need to show that $\mixf p {g_A} {g'_A} \in A$. Let $ f_{\com A} \in \com A$. Then, we have
\begin{align*}
f_{\com A} \circ \mixf p {g_A} {g'_A} (V)
&= \mixf p{ f_{\com A} \circ g_A} {f_{\com A} \circ g'_A } (V)\\
&= \mixf p {g_A \circ  f_{\com A}} {g'_A \circ f_{\com A}} (V) \\
&= (\mixf p {g_A} {g'_A}  ) \circ   f_{\com A} (V).
\end{align*}
So we have $\mixf p {g_A} {g'_A}  \in \bic A = A$.
\end{proof}

\begin{corollary}  \label{corollary:preserving_convex_subsystem}
Let $(S^\Omega,\cT)$ be a doubly convex resource theory. For an intensive embedding $\e_A$ inducing a specification space $S^{\Omega_A}$ that preserves a subsystem $A\in\Sys(\cT)$, 
and such that
$$\Lump (\mixt p V  W)\  \supseteq
\mixt p {\Lump (V)}{\Lump (W)}$$
and for all $f,f',g \in A$
\begin{align*}
&\h_A \circ f \circ \e_A = \h_A \circ f' \circ \e_A
\implies 
\h_A\circ \mixf p f g \circ \e_A = \h_A \circ \mixf p {f'} g \circ \e_A,
\end{align*}
the corresponding restricted resource theory $(S^{\Omega_A},\widetilde A)$ is also doubly convex.
\end{corollary}

When we trace out a subsystem in quantum theory  quantum mechanics, local state spaces and transformations are convex. 

\begin{proposition}[Partial trace induces doubly convex local theories]
Let $(S^{\Omega_{AB}}, \cT)$ be quantum theory on a bipartite Hilbert space $\hilbert_A \otimes \hilbert_B$, and $A \in \Sys( \cT)$ the submonoid  of transformations acting on the first space.
The embedding determined by the partial trace over a subsystem, 
\begin{align*}
\e_A (V_A) = \{\sigma\in\Omega_{AB}: \sigma_A \in V_A\},
\qquad \qquad
\h_A (V) &= \{\sigma_A: \sigma \in V \},
\end{align*}
with $\sigma_A=\tr_B(\sigma)$, 
induces a doubly convex local resource theory $(S^{\Omega_A}, \widetilde A)$.

\end{proposition}

\begin{proof}
We use the same definitions as in Proposition~\ref{prop:quantum_doubly_convex} for the convex combinations of CPTP maps $\E$ and their counterparts $\widetilde \E$ on the specification space. 
We will show that the conditions of Theorem~\ref{corollary:preserving_convex_subsystem} are satisfied. 
Firstly, note that 
$$\Lump_A (V)= \bigcup_{\rho \in V} \{\sigma \in \Omega_{AB}: \tr_B(\sigma) = \tr_B(\rho) \},$$
and linearity of the partial trace ensures that the first condition is met,
\begin{align*}
\Lump_A (\mixt p V W)\
&= \bigcup_{\rho \in V} \bigcup_{\omega \in W} \Lump_A (\{p \ \rho + (1-p) \omega \})
\\
&= \bigcup_{\rho \in V} \bigcup_{\omega \in W}  \{\sigma \in \Omega_{AB}: \tr_B(\sigma) = \tr_B(p \ \rho + (1-p) \omega ) \}
\\
&= \bigcup_{\rho \in V} \bigcup_{\omega \in W}  \{\sigma \in \Omega_{AB}: \tr_B(\sigma) = p \ \tr_B(\rho) + (1-p) \ \tr_B(  \omega ) \}
\\
& =
\bigcup_{\rho \in V}\bigcup_{\omega \in W}  \{p \ \sigma + (1-p)\ \sigma' \in \Omega_{AB}:  p \ \tr_B(\sigma) + (1-p) \ \tr_B(  \sigma' ) = p \ \tr_B(\rho) + (1-p) \ \tr_B(  \omega ) \}
\\
&\supseteq
p \left( \bigcup_{\rho \in V}  
\{\sigma  \in \Omega_{AB}:  p \ \tr_B(\sigma)  =  \tr_B(\rho) \}
\right) 
+ (1-p) \left( \bigcup_{\omega \in W} \{\sigma  \in \Omega_{AB}:  p \ \tr_B(\sigma)  =  \tr_B(\omega) \}  \right) 
\\
&=
\mixt p  {\Lump_A (V)} {\Lump_A (W)}.
\end{align*}
For the second condition, assume that  $\h_A \circ \widetilde \E \circ \e_A = \h_A \circ \widetilde \E'' \circ \e_A$. We have
\begin{align*}
\h_A\circ \mixf p {\widetilde \E} {\widetilde \E'} \circ \e_A (V_A)
&= \bigcup_{\rho_A\in V_A}
\h_A\circ \mixf p {\widetilde \E} {\widetilde \E'} \circ \e_A (\{\rho_A\}) \\
&= \bigcup_{\rho_A\in V_A} \bigcup_{\sigma \in \Omega: \sigma_A=\rho_A}
\h_A\circ \mixf p {\widetilde \E} {\widetilde \E'} (\{\sigma\}) \\
&= \bigcup_{\rho_A\in V_A} \bigcup_{\sigma \in \Omega: \sigma_A=\rho_A}
\h_A (\{ p \ \E (\sigma) + (1-p) \ \E'(\sigma) \}) \\
\flag{\text{partial trace linear}}
&= \bigcup_{\rho_A\in V_A} \bigcup_{\sigma \in \Omega: \sigma_A=\rho_A}
\{ p \ \tr_B (\E (\sigma)) + (1-p) \ \tr_B (\E' (\sigma)) \} \\
&=  \mixtg p  {\h_A\circ \widetilde \E \circ \e_A(V_A)} {\h_A\circ \widetilde \E'\circ \e_A(V_A)} \\
\flag{\text{assumption}}
&=  \mixtg p  {\h_A\circ \widetilde \E'' \circ \e_A(V_A)} {\h_A\circ \widetilde \E' \circ \e_A(V_A)} \\
\flag{\text{revert previous steps}}
&= \h_A \circ \mixf p {\widetilde \E''} {\widetilde \E'} \circ \e_A (V_A).
\end{align*}

\end{proof}


\newpage

\renewcommand{\lofname}{List of theorems and propositions}

\listoftheorems[ignoreall,
show={theorem,proposition}]

\newpage 
\renewcommand{\lofname}{List of definitions}

\listoftheorems[ignoreall,
show={definition}]

\newpage


\begin{thebibliography}{10}
\expandafter\ifx\csname url\endcsname\relax
  \def\url#1{\texttt{#1}}\fi
\expandafter\ifx\csname urlprefix\endcsname\relax\def\urlprefix{URL }\fi
\providecommand{\bibinfo}[2]{#2}
\providecommand{\eprint}[2][]{\url{#2}}

\bibitem{Maxwell1871}
\bibinfo{author}{Maxwell, J.~C.}
\newblock \emph{\bibinfo{title}{{Theory of heat}}}
  (\bibinfo{publisher}{Longmans}, \bibinfo{year}{1871}).
\newblock
  \urlprefix\url{http://books.google.com/books?id=DqAAAAAAMAAJ\&pgis=1}.

\bibitem{Bennett1996}
\bibinfo{author}{Bennett, C.~H.} \emph{et~al.}
\newblock \bibinfo{title}{Purification of noisy entanglement and faithful
  teleportation via noisy channels}.
\newblock \emph{\bibinfo{journal}{Phys. Rev. Lett.}}
  \textbf{\bibinfo{volume}{76}}, \bibinfo{pages}{722--725}
  (\bibinfo{year}{1996}).
\newblock \urlprefix\url{http://link.aps.org/doi/10.1103/PhysRevLett.76.722}.

\bibitem{Horodecki2009}
\bibinfo{author}{Horodecki, R.}, \bibinfo{author}{Pawel, H.},
  \bibinfo{author}{Horodecki, M.} \& \bibinfo{author}{Horodecki, K.}
\newblock \bibinfo{title}{{Quantum entanglement}}.
\newblock \emph{\bibinfo{journal}{Rev. Mod. Phys.}}
  \textbf{\bibinfo{volume}{81}}, \bibinfo{pages}{865--942}
  (\bibinfo{year}{2009}).
\newblock \urlprefix\url{http://link.aps.org/doi/10.1103/RevModPhys.81.865}.

\bibitem{Brandao2015}
\bibinfo{author}{Brand\~{a}o, F. G. S.~L.} \& \bibinfo{author}{Gour, G.}
\newblock \bibinfo{title}{{The general structure of quantum resource theories}}
  \bibinfo{pages}{5} (\bibinfo{year}{2015}).
\newblock \urlprefix\url{http://arxiv.org/abs/1502.03149}.
\newblock \eprint{1502.03149}.

\bibitem{Coecke2014}
\bibinfo{author}{Coecke, B.}, \bibinfo{author}{Fritz, T.} \&
  \bibinfo{author}{Spekkens, R.~W.}
\newblock \bibinfo{title}{{A mathematical theory of resources}}
  \bibinfo{pages}{31} (\bibinfo{year}{2014}).
\newblock \urlprefix\url{http://arxiv.org/abs/1409.5531}.
\newblock \eprint{1409.5531}.

\bibitem{Fritz2015}
\bibinfo{author}{Fritz, T.}
\newblock \bibinfo{title}{{The mathematical structure of theories of resource
  convertibility I}} \bibinfo{pages}{67} (\bibinfo{year}{2015}).
\newblock \urlprefix\url{http://arxiv.org/abs/1504.03661}.
\newblock \eprint{1504.03661}.

\bibitem{Janzing2000}
\bibinfo{author}{Janzing, D.}, \bibinfo{author}{Wocjan, P.},
  \bibinfo{author}{Zeier, R.}, \bibinfo{author}{Geiss, R.} \&
  \bibinfo{author}{Beth, T.}
\newblock \bibinfo{title}{{The thermodynamic cost of reliability and low
  temperatures: Tightening Landauer's principle and the Second Law}}
  \bibinfo{pages}{18} (\bibinfo{year}{2000}).
\newblock \urlprefix\url{http://arxiv.org/abs/quant-ph/0002048}.

\bibitem{Brandao2013b}
\bibinfo{author}{Brandao, F. G. S.~L.}, \bibinfo{author}{Horodecki, M.},
  \bibinfo{author}{Ng, N. H.~Y.}, \bibinfo{author}{Oppenheim, J.} \&
  \bibinfo{author}{Wehner, S.}
\newblock \bibinfo{title}{{The second laws of quantum thermodynamics}}
  \bibinfo{pages}{39} (\bibinfo{year}{2013}).
\newblock \urlprefix\url{http://arxiv.org/abs/1305.5278}.
\newblock \eprint{1305.5278}.

\bibitem{Horodecki2013a}
\bibinfo{author}{Horodecki, M.} \& \bibinfo{author}{Oppenheim, J.}
\newblock \bibinfo{title}{{Fundamental limitations for quantum and nanoscale
  thermodynamics}}.
\newblock \emph{\bibinfo{journal}{Nature communications}}
  \textbf{\bibinfo{volume}{4}}, \bibinfo{pages}{2059} (\bibinfo{year}{2013}).
\newblock \urlprefix\url{http://www.ncbi.nlm.nih.gov/pubmed/23800725}.

\bibitem{Goold2015}
\bibinfo{author}{Goold, J.}, \bibinfo{author}{Huber, M.},
  \bibinfo{author}{Riera, A.}, \bibinfo{author}{del Rio, L.} \&
  \bibinfo{author}{Skrzypczyk, P.}
\newblock \bibinfo{title}{{The role of quantum information in thermodynamics
  --- a topical review}} \bibinfo{pages}{31} (\bibinfo{year}{2015}).
\newblock \urlprefix\url{http://arxiv.org/abs/1505.07835}.

\bibitem{Iten2015}
\bibinfo{author}{Iten, R.}, \bibinfo{author}{Colbeck, R.},
  \bibinfo{author}{Kukuljan, I.}, \bibinfo{author}{Home, J.} \&
  \bibinfo{author}{Christandl, M.}
\newblock \bibinfo{title}{{Quantum Circuits for Isometries}} \bibinfo{pages}{8}
  (\bibinfo{year}{2015}).
\newblock \urlprefix\url{http://arxiv.org/abs/1501.06911}.

\bibitem{Baumgratz2013}
\bibinfo{author}{Baumgratz, T.}, \bibinfo{author}{Cramer, M.} \&
  \bibinfo{author}{Plenio, M.~B.}
\newblock \bibinfo{title}{{Quantifying Coherence}}.
\newblock \emph{\bibinfo{journal}{Arxiv preprint}} \bibinfo{pages}{10 pages}
  (\bibinfo{year}{2013}).
\newblock \urlprefix\url{http://arxiv.org/abs/1311.0275}.
\newblock \eprint{1311.0275}.

\bibitem{Lostaglio2015}
\bibinfo{author}{Lostaglio, M.}, \bibinfo{author}{Jennings, D.} \&
  \bibinfo{author}{Rudolph, T.}
\newblock \bibinfo{title}{{Description of quantum coherence in thermodynamic
  processes requires constraints beyond free energy.}}
\newblock \emph{\bibinfo{journal}{Nature communications}}
  \textbf{\bibinfo{volume}{6}}, \bibinfo{pages}{6383} (\bibinfo{year}{2015}).
\newblock
  \urlprefix\url{http://www.nature.com/ncomms/2015/150310/ncomms7383/full/ncomms7383.html}.

\bibitem{Winter2015}
\bibinfo{author}{Winter, A.} \& \bibinfo{author}{Yang, D.}
\newblock \bibinfo{title}{{Operational Resource Theory of Coherence}}
  \bibinfo{pages}{10} (\bibinfo{year}{2015}).
\newblock \urlprefix\url{http://arxiv.org/abs/1506.07975}.

\bibitem{Gour2008}
\bibinfo{author}{Gour, G.} \& \bibinfo{author}{Spekkens, R.~W.}
\newblock \bibinfo{title}{{The resource theory of quantum reference frames:
  manipulations and monotones}}.
\newblock \emph{\bibinfo{journal}{New Journal of Physics}}
  \textbf{\bibinfo{volume}{10}}, \bibinfo{pages}{033023}
  (\bibinfo{year}{2008}).
\newblock
  \urlprefix\url{http://stacks.iop.org/1367-2630/10/i=3/a=033023?key=crossref.82bca9378b047fb39f6b06f2f8078c1f}.

\bibitem{Marvian2013}
\bibinfo{author}{Marvian, I.} \& \bibinfo{author}{Spekkens, R.~W.}
\newblock \bibinfo{title}{{Modes of asymmetry: the application of harmonic
  analysis to symmetric quantum dynamics and quantum reference frames}}
  \bibinfo{pages}{1--23} (\bibinfo{year}{2013}).
\newblock \urlprefix\url{http://arxiv.org/abs/1312.0680}.
\newblock \eprint{arXiv:1312.0680v1}.

\bibitem{Marvian2014}
\bibinfo{author}{Marvian, I.} \& \bibinfo{author}{Spekkens, R.~W.}
\newblock \bibinfo{title}{{Extending Noether's theorem by quantifying the
  asymmetry of quantum states.}}
\newblock \emph{\bibinfo{journal}{Nature communications}}
  \textbf{\bibinfo{volume}{5}}, \bibinfo{pages}{3821} (\bibinfo{year}{2014}).
\newblock \urlprefix\url{http://arxiv.org/abs/1404.3236}.
\newblock \eprint{1404.3236}.

\bibitem{Marvian2015}
\bibinfo{author}{Marvian, I.}, \bibinfo{author}{Spekkens, R.~W.} \&
  \bibinfo{author}{Zanardi, P.}
\newblock \bibinfo{title}{{Quantum speed limits, coherence and asymmetry}}
  \bibinfo{pages}{14} (\bibinfo{year}{2015}).
\newblock \urlprefix\url{http://arxiv.org/abs/1510.06474}.

\bibitem{Fuchs2015}
\bibinfo{author}{Fuchs, C.~A.} \& \bibinfo{author}{Schack, R.}
\newblock \bibinfo{title}{Qbism and the greeks: why a quantum state does not
  represent an element of physical reality}.
\newblock \emph{\bibinfo{journal}{Physica Scripta}}
  \textbf{\bibinfo{volume}{90}}, \bibinfo{pages}{015104}
  (\bibinfo{year}{2015}).
\newblock \urlprefix\url{http://stacks.iop.org/1402-4896/90/i=1/a=015104}.

\bibitem{Christandl2012}
\bibinfo{author}{Christandl, M.} \& \bibinfo{author}{Renner, R.}
\newblock \bibinfo{title}{Reliable quantum state tomography}.
\newblock \emph{\bibinfo{journal}{Phys. Rev. Lett.}}
  \textbf{\bibinfo{volume}{109}}, \bibinfo{pages}{120403}
  (\bibinfo{year}{2012}).
\newblock
  \urlprefix\url{http://link.aps.org/doi/10.1103/PhysRevLett.109.120403}.

\bibitem{Oppenheim02}
\bibinfo{author}{Oppenheim, J.}, \bibinfo{author}{Horodecki, M.},
  \bibinfo{author}{Horodecki, P.} \& \bibinfo{author}{Horodecki, R.}
\newblock \bibinfo{title}{{Thermodynamical approach to quantifying quantum
  correlations}}.
\newblock \emph{\bibinfo{journal}{Physical Review Letters}}
  \textbf{\bibinfo{volume}{89}}, \bibinfo{pages}{180402}
  (\bibinfo{year}{2002}).

\bibitem{Zanardi2004}
\bibinfo{author}{Zanardi, P.}, \bibinfo{author}{Lidar, D.~A.} \&
  \bibinfo{author}{Lloyd, S.}
\newblock \bibinfo{title}{Quantum tensor product structures are observable
  induced}.
\newblock \emph{\bibinfo{journal}{Phys. Rev. Lett.}}
  \textbf{\bibinfo{volume}{92}}, \bibinfo{pages}{060402}
  (\bibinfo{year}{2004}).
\newblock
  \urlprefix\url{http://link.aps.org/doi/10.1103/PhysRevLett.92.060402}.

\bibitem{Taylor2005}
\bibinfo{author}{Taylor, J.~M.} \emph{et~al.}
\newblock \bibinfo{title}{Fault-tolerant architecture for quantum computation
  using electrically controlled semiconductor spins}.
\newblock \emph{\bibinfo{journal}{Nat Phys}} \textbf{\bibinfo{volume}{1}},
  \bibinfo{pages}{177--183} (\bibinfo{year}{2005}).
\newblock \urlprefix\url{http://dx.doi.org/10.1038/nphys174}.

\bibitem{Brunner2012}
\bibinfo{author}{Brunner, N.}, \bibinfo{author}{Linden, N.},
  \bibinfo{author}{Popescu, S.} \& \bibinfo{author}{Skrzypczyk, P.}
\newblock \bibinfo{title}{Virtual qubits, virtual temperatures, and the
  foundations of thermodynamics}.
\newblock \emph{\bibinfo{journal}{Phys. Rev. E}} \textbf{\bibinfo{volume}{85}},
  \bibinfo{pages}{051117} (\bibinfo{year}{2012}).
\newblock \urlprefix\url{http://link.aps.org/doi/10.1103/PhysRevE.85.051117}.

\bibitem{Leibniz1739}
\bibinfo{author}{Leibniz, G.}
\newblock \emph{\bibinfo{title}{Tentamina theodic{\ae}{\ae} de bonitate dei:
  libertate hominis et origine mali}} (\bibinfo{publisher}{C.H. Bergerus},
  \bibinfo{year}{1739}).
\newblock \urlprefix\url{https://books.google.co.uk/books?id=J7FVAAAAYAAJ}.

\bibitem{Kripke1963}
\bibinfo{author}{Kripke, S.~A.}
\newblock \bibinfo{title}{{S}emantical {C}onsiderations on {M}odal {L}ogic}.
\newblock \emph{\bibinfo{journal}{{A}cta {P}hilosophica {F}ennica}}
  \textbf{\bibinfo{volume}{16}}, \bibinfo{pages}{83--94}
  (\bibinfo{year}{1963}).

\bibitem{Lewis2013}
\bibinfo{author}{Lewis, D.}
\newblock \emph{\bibinfo{title}{{Counterfactuals}}} (\bibinfo{publisher}{John
  Wiley \& Sons}, \bibinfo{year}{2013}).
\newblock
  \urlprefix\url{https://books.google.com/books?id=bCvnk3JMvfAC\&pgis=1}.

\bibitem{Goldblatt2003}
\bibinfo{author}{Goldblatt, R.}
\newblock \bibinfo{title}{{Mathematical modal logic: A view of its evolution}}.
\newblock \emph{\bibinfo{journal}{Journal of Applied Logic}}
  \textbf{\bibinfo{volume}{1}}, \bibinfo{pages}{309--392}
  (\bibinfo{year}{2003}).
\newblock
  \urlprefix\url{http://www.sciencedirect.com/science/article/pii/S1570868303000089}.

\bibitem{Fagin1995}
\bibinfo{author}{Fagin, R.}, \bibinfo{author}{Halpern, J.~Y.},
  \bibinfo{author}{Moses, Y.} \& \bibinfo{author}{Y., V.~M.}
\newblock \bibinfo{title}{{Reasoning About Knowledge}}  (\bibinfo{year}{1995}).
\newblock \urlprefix\url{http://philpapers.org/rec/FAGRAK-3}.

\bibitem{Halpern1995}
\bibinfo{author}{Halpern, J.~Y.}
\newblock \bibinfo{title}{{Reasoning About Knowledge: A Survey}}
  \urlprefix\url{http://citeseerx.ist.psu.edu/viewdoc/summary?doi=10.1.1.40.8823}.

\bibitem{Hintikka1962}
\bibinfo{author}{Hintikka, J.}
\newblock \bibinfo{title}{{Knowledge and Belief}}.
\newblock \emph{\bibinfo{journal}{Philosophical Review}}
  \textbf{\bibinfo{volume}{65}}, \bibinfo{pages}{61 178--189}
  (\bibinfo{year}{1962}).

\bibitem{Ditmarsch2007}
\bibinfo{author}{van Ditmarsch, H.}, \bibinfo{author}{van~der Hoek, W.} \&
  \bibinfo{author}{Kooi, B.}
\newblock \emph{\bibinfo{title}{Dynamic Epistemic Logic}}
  (\bibinfo{publisher}{Springer Netherlands}, \bibinfo{year}{2007}).
\newblock \urlprefix\url{http://link.springer.com/10.1007/978-1-4020-5839-4}.

\bibitem{Baltag2007}
\bibinfo{author}{Baltag, A.}, \bibinfo{author}{Coecke, B.} \&
  \bibinfo{author}{Sadrzadeh, M.}
\newblock \bibinfo{title}{{Epistemic Actions as Resources}}.
\newblock \emph{\bibinfo{journal}{Journal of Logic and Computation}}
  \textbf{\bibinfo{volume}{17}}, \bibinfo{pages}{555--585}
  (\bibinfo{year}{2007}).
\newblock
  \urlprefix\url{http://logcom.oxfordjournals.org/content/17/3/555.short}.

\bibitem{Girard1987}
\bibinfo{author}{Girard, J.-Y.}
\newblock \bibinfo{title}{{Linear logic}}.
\newblock \emph{\bibinfo{journal}{Theoretical Computer Science}}
  \textbf{\bibinfo{volume}{50}}, \bibinfo{pages}{1--101}
  (\bibinfo{year}{1987}).
\newblock
  \urlprefix\url{http://www.sciencedirect.com/science/article/pii/0304397587900454}.

\bibitem{Girard1995}
\bibinfo{author}{Girard, J.-Y.}
\newblock \bibinfo{title}{{Linear logic: its syntax and semantics}}
  \bibinfo{pages}{1--42} (\bibinfo{year}{1995}).
\newblock \urlprefix\url{http://dl.acm.org/citation.cfm?id=212876.212880}.

\bibitem{Dempster1967}
\bibinfo{author}{Dempster, A.~P.}
\newblock \bibinfo{title}{{Upper and Lower Probabilities Induced by a
  Multivalued Mapping}}.
\newblock \emph{\bibinfo{journal}{The Annals of Mathematical Statistics}}
  \textbf{\bibinfo{volume}{38}}, \bibinfo{pages}{325--339}
  (\bibinfo{year}{1967}).
\newblock \urlprefix\url{http://projecteuclid.org/euclid.aoms/1177698950}.

\bibitem{Shafer1976}
\bibinfo{author}{Shafer, G.}
\newblock \emph{\bibinfo{title}{{A Mathematical Theory of Evidence}}}
  (\bibinfo{publisher}{Princeton University Press}, \bibinfo{year}{1976}).
\newblock
  \urlprefix\url{https://books.google.com/books?id=sJZwQgAACAAJ\&pgis=1}.

\bibitem{Yager2008}
\bibinfo{author}{Yager, R.} \& \bibinfo{author}{Liping, L.}
\newblock \emph{\bibinfo{title}{Classic Works of the {D}empster-{S}hafer Theory
  of Belief Functions}} (\bibinfo{publisher}{Springer Berlin Heidelberg},
  \bibinfo{year}{2008}).
\newblock \urlprefix\url{http://link.springer.com/10.1007/978-3-540-44792-4}.

\bibitem{Fritz2015a}
\bibinfo{author}{Fritz, T.} \& \bibinfo{author}{Leifer, M.}
\newblock \bibinfo{title}{{Plausibility measures on test spaces}}
  \bibinfo{pages}{12} (\bibinfo{year}{2015}).
\newblock \urlprefix\url{http://arxiv.org/abs/1505.01151}.

\bibitem{Josang2001}
\bibinfo{author}{{J\o sang}, A.}
\newblock \bibinfo{title}{A logic for uncertain probabilities}.
\newblock \emph{\bibinfo{journal}{International Journal of Uncertainty,
  Fuzziness and Knowledge-Based Systems}} \textbf{\bibinfo{volume}{09}},
  \bibinfo{pages}{279--311} (\bibinfo{year}{2001}).
\newblock \urlprefix\url{http://dl.acm.org/citation.cfm?id=565980.565981}.

\bibitem{Josang2006}
\bibinfo{author}{J{\o}sang, A.}, \bibinfo{author}{Hayward, R.} \&
  \bibinfo{author}{Pope, S.}
\newblock \bibinfo{title}{Trust network analysis with subjective logic}.
\newblock In \emph{\bibinfo{booktitle}{Proceedings of the 29th Australasian
  Computer Science Conference - Volume 48}}, ACSC '06, \bibinfo{pages}{85--94}
  (\bibinfo{publisher}{Australian Computer Society, Inc.},
  \bibinfo{address}{Darlinghurst, Australia, Australia}, \bibinfo{year}{2006}).
\newblock \urlprefix\url{http://dl.acm.org/citation.cfm?id=1151699.1151710}.

\bibitem{Novak1999}
\bibinfo{author}{Nov\'{a}k, V.}, \bibinfo{author}{Perfilieva, I.} \&
  \bibinfo{author}{Mo\v{c}koř, J.}
\newblock \emph{\bibinfo{title}{{Mathematical Principles of Fuzzy Logic}}}
  (\bibinfo{publisher}{Springer Science \& Business Media},
  \bibinfo{year}{1999}).
\newblock
  \urlprefix\url{https://books.google.com/books?id=pJeu6Ue65S4C\&pgis=1}.

\bibitem{Zadeh1965}
\bibinfo{author}{Zadeh, L.}
\newblock \bibinfo{title}{Fuzzy sets}.
\newblock \emph{\bibinfo{journal}{Information and Control}}
  \textbf{\bibinfo{volume}{8}}, \bibinfo{pages}{338 -- 353}
  (\bibinfo{year}{1965}).
\newblock
  \urlprefix\url{http://www.sciencedirect.com/science/article/pii/S001999586590241X}.

\bibitem{Zimmermann2001}
\bibinfo{author}{Zimmermann, H.-J.}
\newblock \emph{\bibinfo{title}{Fuzzy Set Theory—and Its Applications}}
  (\bibinfo{publisher}{Springer Netherlands}, \bibinfo{year}{2001}).
\newblock
  \urlprefix\url{http://link.springer.com/book/10.1007%2F978-94-010-0646-0}.

\bibitem{Munkres2000}
\bibinfo{author}{Munkres, J.}
\newblock \emph{\bibinfo{title}{Topology}}.
\newblock Featured Titles for Topology Series (\bibinfo{publisher}{Prentice
  Hall, Incorporated}, \bibinfo{year}{2000}).
\newblock \urlprefix\url{https://books.google.co.uk/books?id=XjoZAQAAIAAJ}.

\bibitem{Naimpally1971}
\bibinfo{author}{Naimpally, S.} \& \bibinfo{author}{Warrack, B.}
\newblock \emph{\bibinfo{title}{Proximity Spaces}}.
\newblock Cambridge Tracts in Mathematics and Mathematical Physics
  (\bibinfo{publisher}{Cambridge University Press}, \bibinfo{year}{1971}).
\newblock \urlprefix\url{https://books.google.co.uk/books?id=56lxuAAACAAJ}.

\bibitem{Pawlak1973}
\bibinfo{author}{Pawlak, Z.}
\newblock \bibinfo{title}{{Mathematical Foundation of Information Retrieval.}}
\newblock In \emph{\bibinfo{booktitle}{Mathematical Foundations of Computer
  Science: Proceedings of Symposium and Summer School, Strbsk\'{e} Pleso, High
  Tatras, Czechoslovakia, September 3-8, 1973.}}, \bibinfo{pages}{135--136}
  (\bibinfo{year}{1973}).
\newblock
  \urlprefix\url{http://www.researchgate.net/publication/220975215\_Mathematical\_Foundation\_of\_Information\_Retrieval}.

\bibitem{Pawlak1982}
\bibinfo{author}{Pawlak, Z.}
\newblock \bibinfo{title}{{Rough sets}}.
\newblock \emph{\bibinfo{journal}{International Journal of Computer \&
  Information Sciences}} \textbf{\bibinfo{volume}{11}},
  \bibinfo{pages}{341--356} (\bibinfo{year}{1982}).
\newblock \urlprefix\url{http://link.springer.com/10.1007/BF01001956}.

\bibitem{vonNeumann1947}
\bibinfo{author}{Von~Neumann, J.} \& \bibinfo{author}{Morgenstern, O.}
\newblock \emph{\bibinfo{title}{Theory of Games and Economic Behavior}}
  (\bibinfo{publisher}{Princeton University Press}, \bibinfo{year}{1947}).
\newblock \urlprefix\url{https://books.google.ch/books?id=AUDPAAAAMAAJ}.

\bibitem{Lieb1999}
\bibinfo{author}{Lieb, E.~H.} \& \bibinfo{author}{Yngvason, J.}
\newblock \bibinfo{title}{{The physics and mathematics of the second law of
  thermodynamics}} \bibinfo{pages}{1--101} (\bibinfo{year}{1999}).
\newblock \eprint{9708200v2}.

\bibitem{Maurer2011}
\bibinfo{author}{Maurer, U.} \& \bibinfo{author}{Renner, R.}
\newblock \bibinfo{title}{{Abstract Cryptography}}.
\newblock In \bibinfo{editor}{Chazelle, B.} (ed.) \emph{\bibinfo{booktitle}{The
  Second Symposium on Innovations in Computer Science, ICS 2011}},
  \bibinfo{pages}{1--21} (\bibinfo{publisher}{Tsinghua University Press},
  \bibinfo{year}{2011}).
\newblock
  \urlprefix\url{http://conference.itcs.tsinghua.edu.cn/ICS2011/content/papers/14.html}.

\bibitem{HorodeckiOppenheim2013}
\bibinfo{author}{Horodecki, M.} \& \bibinfo{author}{Oppenheim, J.}
\newblock \bibinfo{title}{(quantumness in the context of) resource theories}
  (\bibinfo{year}{2012}).
\newblock
  \urlprefix\url{http://www.worldscientific.com/doi/abs/10.1142/S0217979213450197?queryID=57/1541482}.

\bibitem{delRio2011}
\bibinfo{author}{del Rio, L.}, \bibinfo{author}{Aberg, J.},
  \bibinfo{author}{Renner, R.}, \bibinfo{author}{Dahlsten, O.} \&
  \bibinfo{author}{Vedral, V.}
\newblock \bibinfo{title}{{The thermodynamic meaning of negative entropy.}}
\newblock \emph{\bibinfo{journal}{Nature}} \textbf{\bibinfo{volume}{474}},
  \bibinfo{pages}{61--3} (\bibinfo{year}{2011}).
\newblock \urlprefix\url{http://dx.doi.org/10.1038/nature10123}.

\bibitem{Turgut2007}
\bibinfo{author}{Turgut, S.}
\newblock \bibinfo{title}{{Catalytic transformations for bipartite pure
  states}}.
\newblock \emph{\bibinfo{journal}{Journal of Physics A: Mathematical and
  Theoretical}} \textbf{\bibinfo{volume}{40}}, \bibinfo{pages}{12185--12212}
  (\bibinfo{year}{2007}).
\newblock \urlprefix\url{http://stacks.iop.org/1751-8121/40/i=40/a=012}.

\bibitem{Aberg2013}
\bibinfo{author}{Aberg, J.}
\newblock \bibinfo{title}{{Catalytic Coherence}}  (\bibinfo{year}{2013}).
\newblock \urlprefix\url{http://arxiv.org/abs/1304.1060}.
\newblock \eprint{1304.1060}.

\bibitem{Ng2014}
\bibinfo{author}{Ng, N. H.~Y.} \emph{et~al.}
\newblock \bibinfo{title}{{Limits to catalysis in quantum thermodynamics}}
  \bibinfo{pages}{1--18} (\bibinfo{year}{2014}).
\newblock \urlprefix\url{http://arxiv.org/abs/1405.3039}.
\newblock \eprint{arXiv:1405.3039v1}.

\bibitem{Lieb2014}
\bibinfo{author}{Lieb, E.~H.} \& \bibinfo{author}{Yngvason, J.}
\newblock \bibinfo{title}{Entropy meters and the entropy of non-extensive
  systems}.
\newblock \emph{\bibinfo{journal}{Proceedings of the Royal Society of London A:
  Mathematical, Physical and Engineering Sciences}}
  \textbf{\bibinfo{volume}{470}} (\bibinfo{year}{2014}).

\bibitem{Chiribella2016}
\bibinfo{author}{Chiribella, G.}, \bibinfo{author}{D'Ariano, G.~M.} \&
  \bibinfo{author}{Perinotti, P.}
\newblock \bibinfo{title}{Quantum from principles}.
\newblock In \bibinfo{editor}{Chiribella, G.} \& \bibinfo{editor}{Spekkens,
  R.~W.} (eds.) \emph{\bibinfo{booktitle}{Quantum Theory: Informational
  Foundations and Foils}}, \bibinfo{pages}{171--222}
  (\bibinfo{publisher}{Springer Netherlands}, \bibinfo{address}{Dordrecht},
  \bibinfo{year}{2016}).

\bibitem{Masanes2013}
\bibinfo{author}{Masanes, L.}, \bibinfo{author}{M\"{u}ller, M.~P.},
  \bibinfo{author}{Augusiak, R.} \& \bibinfo{author}{P\'{e}rez-Garc\'{\i}a, D.}
\newblock \bibinfo{title}{{Existence of an information unit as a postulate of
  quantum theory.}}
\newblock \emph{\bibinfo{journal}{Proceedings of the National Academy of
  Sciences of the United States of America}} \textbf{\bibinfo{volume}{110}},
  \bibinfo{pages}{16373--7} (\bibinfo{year}{2013}).
\newblock \urlprefix\url{http://www.pnas.org/content/110/41/16373}.

\bibitem{Henson2015}
\bibinfo{author}{Henson, J.} \& \bibinfo{author}{Sainz, A.~B.}
\newblock \bibinfo{title}{{Macroscopic noncontextuality as a principle for
  almost-quantum correlations}}.
\newblock \emph{\bibinfo{journal}{Physical Review A}}
  \textbf{\bibinfo{volume}{91}}, \bibinfo{pages}{042114}
  (\bibinfo{year}{2015}).
\newblock \urlprefix\url{http://link.aps.org/doi/10.1103/PhysRevA.91.042114}.

\bibitem{Tomamichel2010}
\bibinfo{author}{Tomamichel, M.}, \bibinfo{author}{Colbeck, R.} \&
  \bibinfo{author}{Renner, R.}
\newblock \bibinfo{title}{{Duality between smooth min- and max-entropies}}.
\newblock \emph{\bibinfo{journal}{IEEE Transactions on Information Theory}}
  \textbf{\bibinfo{volume}{56}}, \bibinfo{pages}{4674--4681}
  (\bibinfo{year}{2010}).
\newblock
  \urlprefix\url{http://ieeexplore.ieee.org/xpl/articleDetails.jsp?arnumber=5550419}.

\bibitem{Gudder1973}
\bibinfo{author}{Gudder, S.}
\newblock \bibinfo{title}{{Convex structures and operational quantum
  mechanics}}.
\newblock \emph{\bibinfo{journal}{Communications in Mathematical Physics}}
  \textbf{\bibinfo{volume}{29}}, \bibinfo{pages}{249--264}
  (\bibinfo{year}{1973}).
\newblock \urlprefix\url{http://projecteuclid.org/euclid.cmp/1103858551}.

\bibitem{Flood1981}
\bibinfo{author}{Flood, J.}
\newblock \bibinfo{title}{{Semiconvex geometry}}.
\newblock \emph{\bibinfo{journal}{Journal of the Australian Mathematical
  Society}} \textbf{\bibinfo{volume}{30}}, \bibinfo{pages}{496}
  (\bibinfo{year}{2009}).
\newblock
  \urlprefix\url{http://journals.cambridge.org/action/displayFulltext?type=8\&fid=4896424\&jid=JAZ\&volumeId=30\&issueId=04\&aid=4896420\&bodyId=\&membershipNumber=\&societyETOCSession=}.

\bibitem{Fritz2009}
\bibinfo{author}{Fritz, T.}
\newblock \bibinfo{title}{{Convex Spaces I: Definition and Examples}}
  \bibinfo{pages}{24} (\bibinfo{year}{2009}).
\newblock \urlprefix\url{http://arxiv.org/abs/0903.5522}.

\bibitem{Vel1993}
\bibinfo{author}{van~de Vel, M.}
\newblock \emph{\bibinfo{title}{{Theory of Convex Structures}}}
  (\bibinfo{publisher}{Elsevier}, \bibinfo{year}{1993}).
\newblock
  \urlprefix\url{https://books.google.com/books?hl=en\&lr=\&id=oDUmytVzSPMC\&pgis=1}.

\end{thebibliography}

\end{document}